\documentclass[10pt,a4paper,twoside,reqno]{amsart}

\pdfoutput=1
\usepackage{subcaption}
\usepackage[T1]{fontenc}
\usepackage{float}
\usepackage{latexsym}
\usepackage{amsmath,amstext,amssymb,amsfonts,
amscd,bm,array,multirow,amsbsy,mathrsfs,calc}
\usepackage{t1enc}
\usepackage{indentfirst}
\usepackage{pb-diagram}
\usepackage{graphicx}
\usepackage{enumerate}
\usepackage{color}
\usepackage[all]{xy}
\usepackage{hyperref}
\hypersetup{colorlinks=false,pdfborderstyle={/S/U/W 0}}
\usepackage[usenames,dvipsnames]{xcolor}
\usepackage{url}
\usepackage{upgreek}

 \theoremstyle{plain}
\newtheorem{thm}{Theorem}[section]

 \newtheorem{prop}[thm]{Proposition}

 \newtheorem{lemma}[thm]{Lemma}
\newtheorem{cor}[thm]{Corollary}

\theoremstyle{definition}
 \newtheorem{definition}[thm]{Definition}

 \theoremstyle{remark}
 \newtheorem{remark}[thm]{Remark}
\newtheorem{eg}[thm]{Example}

\newcommand{\be}{\begin{equation*}}
\newcommand{\ee}{\end{equation*}}
\newcommand{\ben}{\begin{equation}}
\newcommand{\een}{\end{equation}}
\newcommand{\beqa}{\begin{eqnarray*}}
\newcommand{\eeqa}{\end{eqnarray*}}
\newcommand{\beqan}{\begin{eqnarray}}
\newcommand{\eeqan}{\end{eqnarray}}
\newcommand{\nn}{\nonumber}

\def\Z{\mathbb{Z}}

\def\R{\mathbb{R}}

\def\Hess{\mathrm{Hess}}

\def\Crit{\mathrm{Crit}}

\def\pd{\partial}

\def\reg{\mathrm{reg}}
\def\sing{\mathrm{sing}}
\def\crit{\mathrm{crit}}
\def\noncrit{\mathrm{noncrit}}

\def\Int{\mathrm{Int}}

\def\fM{\mathfrak{M}}

\def\fX{\mathfrak{X}}

\def\dd{\mathrm{d}}

\def\bepsilon{\boldsymbol{\varepsilon}}
\def\veta{\boldsymbol{\eta}}
\def\hveta{{\hat \veta}}

\newcommand{\id}{\mathrm{id}}
\newcommand{\Tr}{\mathrm{Tr}}

\newcommand{\sign}{\mathrm{sign}}

\def\cB{\mathcal{B}}
\def\cC{\mathcal{C}}

\def\cE{\mathcal{E}}
\def\cF{\mathcal{F}}
\def\cG{\mathcal{G}}
\def\cH{\mathcal{H}}
\def\cL{\mathcal{L}}
\def\cM{\mathcal{M}}
\def\cN{\mathcal{N}}

\def\cX{\mathcal{X}}
\def\cU{\mathcal{U}}

\def\Hom{\mathrm{Hom}}

\def\heta{{\hat \eta}}
\def\hetap{\heta^\parallel}

\def\homega{\hat{\omega}}
\def\hOmega{\Omega}
\def\hepsilon{{\hat \bepsilon}}

\def\fM{\mathfrak{M}}
\def\fS{\mathfrak{S}}
\def\fs{\mathfrak{s}}

\def\f{{\bf f}}

\def\kin{\mathrm{kin}}
\def\pot{\mathrm{pot}}

\def\v{\mathrm{v}}

\def\bvarphi{\boldsymbol{\varphi}}
\def\bpsi{\boldsymbol{\psi}}

\newcommand{\eqdef}{\stackrel{{\rm def.}}{=}}

\def\red{\mathrm{red}}

\def\grad{\mathrm{grad}}

\def\f{\mathfrak{f}}
\def\hc{{\hat c}}

\def\cZ{\mathcal{Z}}
\def\constant{\mathrm{constant}}

\def\vol{\mathrm{vol}}

\def\rk{\mathrm{rk}}
\def\ry{\mathrm{y}}
\def\rzeta{\upzeta}
\begin{document}


\title[Natural coordinates and horizontal approximations]{Natural coordinates and horizontal approximations in two-field cosmological models}

\author{Calin Iuliu Lazaroiu}

\address{Departamento de Matematicas, Universidad UNED - Madrid\\
 Calle de Juan del Rosal 10, 28040, Madrid, Spain\\
clazaroiu@mat.uned.es\\
and\\
Horia Hulubei National
  Institute of Physics and Nuclear Engineering,\\
  Reactorului 30, Bucharest-Magurele, 077125, Romania\\
lcalin@theory.nipne.ro}

\maketitle

\begin{abstract}
We construct natural local coordinate systems on the phase space of
two-field cosmological models with orientable target space, which
allow for a description of cosmological flows through quantities of
direct physical interest. Such coordinates are induced by the {\em
  fundamental observables} of the model, which we formulate geometrically
using the tautological bundle of the tangent bundle of the scalar
manifold.  We also describe a large class of geometric dynamical
approximations induced by the choice of an Ehresmann connection in the
tangent bundle of the scalar manifold. Such approximations take a
conceptually simple form in natural coordinates and we illustrate one
of them as an application.
\end{abstract}



\section*{Introduction}

Cosmological models with more than one scalar field and arbitrary
target space geometry provide natural generalizations of classical
one-field inflationary cosmology and are expected to play a crucial
role in cosmological applications of fundamental theories of gravity
and matter such as superstring theory and its low energy limit (which
to leading order is described by supergravity). Such {\em multifield
  models} arise naturally in supersymmetric string compactifications
(which typically produce a large number of moduli fields) and might be
preferred or even required in fundamental theories which allow for a
consistent quantization of gravity \cite{GK,OPSV,AP,BPR}. This led to
renewed interest in such models, which traditionally were much
less studied than their one-field counterparts. Multifield models allow for
rich behavior which is relevant to various questions in
phenomenological cosmology \cite{PSZ,FRPRW,Lilia1,ASSV,Lilia2, Sch1,
  Sch2, Sch3}.

Despite their crucial conceptual and theoretical importance,
multifield cosmological models are poorly understood --- mainly
because their dynamics (even in the absence of cosmological
perturbations) can be extremely nontrivial already in the two-field
case, as illustrated for example in \cite{genalpha,
  elem,modular}. Even when resorting to numerical methods, the
identification of solutions with various desirable features is
unlikely to be feasible in generic multifield models without
sophisticated mathematical tools\footnote{ A notable exception is
provided by models admitting ``hidden symmetries'' \cite{Noether1,
  Noether2, Hesse}, for which exact solutions can sometimes be found
.} pertaining to the geometric theory of dynamical systems
\cite{Palis,Katok}. It is precisely for this reason that such models
are of interest in mathematical physics and require the development of
a systematic theory. A few steps towards this were were made in
\cite{ren,grad}.

In the absence of cosmological perturbations, a multifield
cosmological model is defined by a nonlinear geometric second order
ODE associated to a {\em scalar triple} $(\cM,\cG,\Phi)$, where
$(\cM,\cG)$ is a connected (but generally non-compact) smooth and
complete Riemannian manifold without boundary (called the
{\em scalar manifold} of the model) and $\Phi$ is a smooth real-valued
function defined on $\cM$, which we assume to be non-constant and
strictly positive. To such data one associates the {\em cosmological
  equation} of the scalar triple:
\be
\nabla_t\dot{\varphi}(t)+ \frac{1}{M_0}\left[||\dot{\varphi}(t)||^2+2\Phi(\varphi(t))\right]^{1/2}\dot{\varphi}(t)+ (\grad \Phi)(\varphi(t))=0~~,
\ee
where $\varphi:I\rightarrow \cM$ is a smooth curve in $\cM$ defined on
a nondegenerate interval $I$ and the dot stands for derivation with
respect to $t\in I$. Here $||~||:T\cM\rightarrow \R_{\geq 0}$ and
$\grad:\cC^\infty(\cM)\rightarrow \cX(\cM)$ are the norm function and
gradient operator defined by $\cG$, $\nabla$ is the Levi-Civita
connection of $(\cM,\cG)$ and $\nabla_t\eqdef
\nabla_{\dot{\varphi}(t)}$. The {\em rescaled Planck mass} $M_0>0$ is
a parameter which can be absorbed by redefining $\cG$ --though it is
more convenient for various arguments to treat it as a free
parameter. The problem of mathematical and physical interest is to
study the solutions of this geometric ODE for general scalar triples
$(\cM,\cG,\Phi)$ and to develop geometric methods for extracting the
qualitative features of solutions and for approximating them {\em
  efficiently}.  Among other questions, one is interested in the
asymptotic behavior of solutions which (depending on the asymptotics of
$\Phi$) might ``escape towards infinity'' at the Freudenthal ends of
$\cM$.  We refer the reader to \cite{genalpha, ren, grad} for a
conceptual discussion of certain aspects of multifield cosmological
models with topologically non-trivial target spaces $\cM$.

The first nontrivial case of interest arises for two-field
cosmological models, which comprise the case $\dim\cM=2$ -- i.e. when
$(\cM,\cG)$ is a (generally non-compact) Riemann surface. Such models
can already display intricate behavior which depends strongly on the
topology of $\cM$ and on the asymptotic behavior of $\Phi$ and $\cG$
towards its Freudenthal ends (see
\cite{genalpha,elem,modular}). Despite being the focus of much
recent interest, little is known about the behavior of solutions of the
cosmological equation for such models, except for special examples
which are generally studied using naive approximations or by numerical
means. In particular, no systematic theory exists which provides
geometric approximation schemes that allow for efficient control of
the error terms. Moreover, no global geometric description of
quantities of physical interest was developed in the physics
literature.

In the present paper, we address some of these issues for two-field
models with orientable target space $\cM$ by developing the geometric
theory of natural cosmological observables and associated notions
while introducing a wide class of approximation schemes (which we call
``horizontal approximations'') for the cosmological equation, which
are defined by the choice of an Ehresmann connection on an open subset
of $T\cM$. As applications, we write the first order form of the
cosmological equation in certain ``universal'' local coordinate
systems on $T\cM$ which have direct physical interpretation and
describe the horizontal approximation and some of its attending
objects defined by one of of these. This turns out to correspond to
a proposal made previously in the physics literature, some
aspects of which were partially clarified in \cite{cons}. Given space
limitations, we postpone the analysis of two-field cosmological
dynamics in the natural coordinates identified in this
paper (as well as an analysis of the error of horizontal
approximations) for further work.

The paper is organized as follows. In Section \ref{sec:mf}, we
describe the concept of horizontal approximation of a second order
geometric ODE which is induced by a choice of an Ehresmann connection
in the tangent bundle of the underlying manifold. We also discuss
fiberwise and ``special'' coordinate systems on the tangent bundle,
which naturally define such approximations and will arise in latter
sections. In Section \ref{sec:models}, we summarize the mathematical
theory of two-field cosmological models in the absence of
perturbations and give a geometric treatment of local cosmological
observables and their dynamical reduction. Section \ref{sec:adapted}
studies a certain local frame of $T\cM$ which is naturally defined by a
scalar triple. Section \ref{sec:fobs} gives a geometric formulation of
the {\em fundamental observables} of a scalar triple, which are
dynamical reductions of certain natural observables used in the
physics literature. Section \ref{sec:fibcoords} discusses natural
fiberwise coordinates on $T\cM$ which are induced by the fundamental
observables. In Section \ref{sec:phase}, we discuss local coordinate
systems of $T\cM$ induced by natural fiberwise coordinates and compute
the semispray of the cosmological equation in those coordinate
systems. Section \ref{sec:SR} discusses constant roll (and in
particular slow roll and ultra slow roll) manifolds while Section
\ref{sec:cons} discusses the quasi-conservative and strongly
dissipative regimes, giving some insight into the meaning of the
conservative function introduced in Section \ref{sec:fobs}. In Section
\ref{sec:horadapted}, we illustrate the general ideas of Section
\ref{sec:mf} by discussing the horizontal approximation in adapted
coordinates. Appendix \ref{app:signed} recalls some geometric notions
associated to curves in oriented Riemann surfaces and introduces the
universal Frenet frame which plays a crucial role in Section
\ref{sec:fobs}. Appendix \ref{app:param} gives the relation between
various quantities used in this paper and certain quantities
considered in the physics literature.

\subsection{Notations and conventions}
\label{subsec:notations}
\noindent All manifolds considered in this paper are smooth and
paracompact. For any manifold $\cM$, we let ${\dot T}\cM$ denote the
slit tangent bundle of $\cM$, which is defined as the complement of
the image of the zero section in $T\cM$. We denote by $\R_\cM$ the
trivial real line bundle on $\cM$. For any smooth vector bundle $E$
over $\cM$ and any $k\in \Z_{\geq 0}\sqcup \{\infty\}$, we denote by
$\cC^k(E)$ the space of those sections of $E$ over $\cM$ which are of
class $\cC^k$. For ease of notation, we let $\Gamma(E)\eqdef
\cC^\infty(E)$ denote the space of smooth sections of $E$. We denote
by $\fX(\cM)\eqdef \Gamma(T\cM)$ the space of smooth vector fields
defined on $\cM$.  We denote by $\pi:T\cM\rightarrow \cM$ the bundle
projection of $T\cM$ and by $F\eqdef \pi^\ast(T\cM)$ the tautological
vector bundle over $T\cM$ (sometimes called the ``Finsler bundle'' of
$\cM$, see \cite{SLK}).

Let:
\be
C(\cM)\eqdef \sqcup_{I\in \Int}\cC^\infty(I,\cM)
\ee
be the set of all smooth curves in $\cM$, where $\Int$ is the set of all
nondegenerate intervals on the real axis and let $C(T\cM)$ be the set of
all smooth curves in $T\cM$. The {\em canonical lift} of curves from
$\cM$ to $T\cM$ is the map $c:C(\cM)\rightarrow C(T\cM)$ defined
through:
\be
c(\varphi)\eqdef \dot{\varphi}~~\forall \varphi\in C(\cM)~~,
\ee
while the {\em canonical projection} of curves from $T\cM$ to $\cM$
is the map $p:C(T\cM)\rightarrow C(\cM)$ defined through:
\be
p(\gamma)\eqdef \pi\circ \gamma~~\forall \gamma\in C(T\cM)~~.
\ee
We have $p\circ c=\id_{C(\cM)}$, which means that $c$ is a section of
$p$. We denote by $C_0(T\cM)$ the image of the map $c$. The set
$C_0(T\cM)$ consists of those curves $\gamma:I\rightarrow T\cM$ which
have the property that $\dot{\gamma}(t)\in TT\cM$ is a second order tangent
vector for all $t\in I$ (i.e. an element of $TT\cM$ which is invariant
under the flip involution of $TT\cM$). Then the maps $p$ and $c$ restrict to
mutually-inverse bijections between $C(\cM)$ and $C_0(T\cM)$.

For any Riemannian manifold $(\cM,\cG)$, we denote
$\cN:T\cM\rightarrow \R_{\geq 0}$ the norm function induced by
$\cG$ on $T\cM$:
\ben
\label{Ndef}
\cN(u)\eqdef ||u||~~\forall u\in T\cM~~.
\een

\begin{definition}
\label{def:pointed}
A curve $\varphi:I\rightarrow N$ (where $N$ is any manifold) is called
pointed if $0\in \mathring{I}$, where $\mathring{I}$ is the interior of the
interval $I$. 
\end{definition}

\section{Horizontal approximations for second order geometric ODEs}
\label{sec:mf}

\subsection{Approximations of a geometric dynamical system induced by complementary distributions}

\noindent Let $N$ be a manifold and $S\in \cX(N)$ a vector field
defined on $N$. We view $(N,S)$ as an autonomous dynamical system
whose flow curves $\gamma:I\rightarrow N$ are the solutions of the
integral curve equation of $S$:
\be
\dot{\gamma}(t)=S(\gamma(t))~~\forall t\in I~~.
\ee
Two Frobenius distributions $V,W\in TN$ are called {\em complementary}
if $TN=V\oplus W$. We think of $V$ as being ``vertical'' and of $W$ as
being ``horizontal'' and use notation justified by this for reasons
which will become clear in the next subsection. Given a pair of
complementary distributions $(V,W)$ on $N$, any vector field $X\in
\cX(N)$ decomposes uniquely as $X=X_V\oplus X_H$ with $X_V\in
\Gamma(V)$ and $X_H\in \Gamma(W)$. This defines projectors $P_V\in
\Hom(TN,V)$ and $P_H\in \Hom(TN,W)$ such that:
\be
P_VX=X_V~~,~~P_HX=X_H~~.
\ee
We will use terminology which doesn't refer to $V$ explicitly, for
reasons which will become clear in our discussion of second order
geometric ODEs below.

\begin{definition}
The $W$-flow curves of $S$ are the integral curves of the vector field
$S_H$, i.e. the solutions $\lambda:I\rightarrow N$ of the equation:
\be
\dot{\lambda}(t)=S_H(\lambda(t))~~\forall t\in I~~.
\ee
\end{definition}

\noindent In the situation above, the {\em $W$-approximation} of the
dynamical system $(N,S)$ consists of replacing it with the dynamical
system $(N,S_H)$, which we call the {\em $W$-approximant} of $(N,S)$.
Accordingly, one approximates the integral curves $\gamma$ of $S$ by
the integral curves $\lambda$ of $S_H$. The approximation can be
accurate only for those integral curves $\gamma$ of $S$ which have the
property that $S_V(\gamma(t))$ is ``small'' for all $t\in I$ with
respect to some appropriate notion of smallness for sections of $W$;
this requires the choice of an appropriate topology on $\Gamma(W)$
(for example, the Whitney topology or the distance topology induced by
a metric on $TN$). Intuitively, the approximation is good if $\gamma$
``stays close'' to the zero locus of $S_V$ in $N$. One can of course
also consider the ``complementary'' approximation in which one
replaces $S$ by $S_V$.

\begin{definition}
The {\em $W$-mean field locus} $\cZ_W$ is the zero set of $S_V$:
\be
\cZ_W\eqdef \{n\in N~\vert~S_V(n)=0\}~~.
\ee
\end{definition}

Notice that $\cZ_W$ is a manifold if $S_V$ is transverse
to the zero section of $V$. Let $n_0$ be a point of $Z_W$ and
$\gamma:I\rightarrow N$ be a pointed\footnote{See Definition
\ref{def:pointed}.}  integral curve of $S$ such that
$\gamma(0)=n_0$. Further restricting $I$ if necessary, let
$\lambda:I\rightarrow N$ be an integral curve of $S_H$ such that
$\lambda(0)=n_0$; of course, $\lambda$ and $\gamma$ are uniquely
determined by $n_0$ and the interval $I$ (which need not be the
maximal interval of definition of any of these two integral
curves). Then the approximation $\gamma(t)\approx \lambda(t)$ is exact
at $t=0$ and remains accurate for $t\in (-\epsilon,\epsilon)\subset I$
for some sufficiently small $\epsilon>0$.

An important special case arises if the distribution $W$ is Frobenius
integrable.  In this case, we have $W=T\cF$, where $\cF$ is the
foliation of $N$ which integrates $W$. Since $S_H\in \Gamma(W)$ is
tangent to the leaves of $\cF$, each integral curve $\lambda$ of $S_H$
lies on a leaf $\cL$ of this foliation. Let $\dim N=d$ and $\dim
\cL=\rk W=r$.  If $(U,(x^i,y^j)_{i=1,\ldots r, j=1,\ldots d-r})$ is a
foliated chart of $\cF$ with $x^i$ and $y^j$ corresponding
respectively to the leaf and transverse directions, then the plaques
of $\cF$ in this chart are locally described the condition
$y^i=\constant$ and the integral curve equation of $S_H$ amounts to
the system of ODEs:
\be
\dot{x}^i(t)=S_H^i(x(t), y(t))~~,~~\dot{y}^j(t)=0~~(i=1,\ldots,r~,~j=1,\ldots d-r)~~,
\ee
where $S_H=S_H^i\frac{\pd}{\pd x^i}$. In particular, the ``vertical
speeds'' $\dot{y}^j$ are neglected in the $W$-approximation. This
parallels the separation into ``slow'' and ``fast variables'' (where
the time variation of the ``slow variables'' is
neglected) which is common in various physics approximations such as
the Born-Oppenheimer approximation for the multiparticle Schrodinger
equation. 

\subsection{Formulation for second order geometric ODEs}

\noindent Let $\cM$ be a connected manifold and $S\in \cX(T\cM)$ be a
second order vector field \cite{SLK} (a.k.a. ``semispray'') defined on
the total space of the tangent bundle $\pi:T\cM\rightarrow \cM$; this
means that $S$ is invariant under the ``flip'' involution of the
double tangent bundle $TT\cM$.  It is well-known (see loc. cit.) that
such a vector field describes the first order formulation of an
autonomous {\em geometric} second order ODE defined on $\cM$ whose
second order term separates. In any local coordinate system $(x^i)$ on
$\cM$, this corresponds to a system of $n$ second order ODEs of the
form:
\be
\ddot{x}^i=f^i(x,\dot{x})~~.
\ee
The adjective ``geometric'' means that this second order system is
invariant under coordinate transformations on $\cM$; equivalently, it
means that it can be formulated as the condition
$j^2(\varphi)=\fs(j^1(\varphi))$ for curves $\varphi:I\rightarrow
\cM$, where $j^1(\varphi)=\dot{\varphi}:I\rightarrow J^1(\cM)=T\cM$
and $j^2(\varphi):I\rightarrow J^2(\cM)$ are the first and second jet
prolongations of $\varphi$ and $\fs$ is a section of the jet bundle
fibration $\pi_{21}:J^2(\cM)\rightarrow J^1(\cM)$. Here
$J^k(\cM)\rightarrow \cM$ denotes the $k$-th jet bundle of curves in
$\cM$ (see Subsection \ref{subsec:obs}).

By definition, a {\em flow curve} of $S$ is an integral curve
$\gamma:I\rightarrow T\cM$ of $S$ while a {\em solution curve} of $S$
is a solution $\varphi:I\rightarrow \cM$ of the corresponding second
order geometric ODE defined on $\cM$. Let $C_s(\cM)\subset C(\cM)$ be
the set of solution curves and $C_f(T\cM)\subset C(T\cM)$ be the set
of flow curves, where $C(N)$ denotes the set of smooth curves in a
manifold $N$.  Since $S$ is a semispray, the maps $c:C(\cM)\rightarrow
C(T\cM)$ and $p:C(T\cM)\rightarrow C(\cM)$ defined through
$c(\varphi):=\dot{\varphi}$ and $p(\gamma):=\pi\circ \gamma$ (see
Subsection \ref{subsec:notations}) restrict to mutually-inverse
bijections between $C_s(\cM)$ and $C_f(T\cM)$.

Let $\pi:T\cM\rightarrow \cM$ be the projection of $T\cM$, $V:=\ker
(\dd\pi)\subset TT\cM$ be the vertical distribution of $T\cM$ and $W$
be a complete Ehresmann connection on $T\cM$, i.e. a distribution
$W\subset TT\cM$ on $T\cM$ which is complementary to $V$ in $TT\cM$:
\be
TT\cM=V\oplus W
\ee
and which defines a horizontal lift of any curve in $\cM$ through any
point of $T\cM$ lying above that curve. We denote by $X_V, X_H$ the
vertical and horizontal projections of any vector field field $X\in
\Gamma(TT\cM)$ defined on $T\cM$. For any pointed curve
$\varphi:I\rightarrow \cM$ and any $u\in T\cM$ with
$\pi(u)=\varphi(0)$, let $\bvarphi_u:I\rightarrow T\cM$ be the
horizontal lift of $\varphi$ through $u$ at
$t=0$, which is uniquely determined by the conditions:
\be
\dot{\bvarphi}_u(t)\in W_{\bvarphi_u(t)}~~\forall t\in I~~,~~\bvarphi_u(0)=u~~\mathrm{and}~~\pi\circ \bvarphi_u=\varphi~~.
\ee

\subsection{The mean field locus defined by $W$}

\noindent The {\em mean field locus} $\cZ_W$ of $S$ relative to the
Ehresmann connection $W$ is the $W$-mean field locus of $S$ defined by
the complementary distributions $V$ and $W$ on $T\cM$, i.e.  the zero
set of $S_V$ in $T\cM$:
\be
\cZ_W\eqdef \{u\in T\cM~\vert~S_V(u)=0\}\subset T\cM~~.
\ee

\begin{definition}
We say that $S$ is {\em regular with respect to $W$} if $S_V$ is
transverse to the zero section of the vertical bundle $V\rightarrow
T\cM$.
\end{definition}

\noindent When $S$ is regular with respect to $W$, the mean field
locus $\cZ_W$ is a submanifold of $T\cM$ of dimension equal to that of
$\cM$. In this case, $\cZ_W$ is called the {\em mean field manifold}
of $S$ relative to $W$. In this case, the implicit function theorem
implies that $\cZ_W$ is locally a multisection of the fiber bundle
$\pi:T\cM\rightarrow \cM$ (because the restriction of $S_V$ to each
fiber $F$ of $T\cM$ is a vector field tangent to $F$ which is
transverse to the zero section of $TF$ and hence its vanishing locus
is a discrete set of points). In particular, $T\cZ_W$ is complementary
to the distribution $V\vert_{\cZ_W}$ inside $TT\cM\vert_{\cZ_W}$.  We
have $S(u)=S_H(u)\in W$ for all $u\in \cZ_W$. From now on, we assume
that $S$ is regular relative to $W$.

\begin{remark}
Notice that $S\vert_{\cZ_W}=S_H\vert_{\cZ_W}$ need not be tangent to
$\cZ_W$ and hence the integral curves of $S_H$ starting from a point 
of the mean field manifold need not be contained in that manifold.
\end{remark}

\subsection{The horizontal approximation defined by $W$}

\begin{definition}
A curve $\lambda:I\rightarrow T\cM$ is called a {\em phase space
  $W$-curve} of $S$ if it is an integral curve of the vector field
$S_H$, i.e. if it satisfies:
\ben
\label{Weq}
\dot{\lambda}(t)=S_H(\lambda(t))~~\forall t\in I~~.
\een
\end{definition}

\noindent Notice that any phase space $W$-curve $\lambda:I\rightarrow
T\cM$ is horizontal. Let $C_H(T\cM)$ be the set of curves in $T\cM$
which are horizontal with respect to $W$ and $C_W(T\cM)\subset
C_H(T\cM)$ be the subset of phase space $W$-curves. Recall that $C_f(T\cM)$
denotes the set of integral curves of $S$. 

\begin{definition}
A curve $\psi:I\rightarrow \cM$ is called a {\em $W$-curve} of $S$ if
$\psi=\pi\circ \lambda$ for some phase space $W$-curve
$\lambda:I\rightarrow T\cM$ of $S$, i.e. if there exists a horizontal lift
of $\psi$ which is a phase space $W$-curve. 
\end{definition}

\noindent Let $C_W(\cM)$ be the set of $W$-curves of $S$ and recall
that $C_s(\cM)$ denotes the set of solution curves of the second order
geometric ODE defined by $S$. Since a horizontal lift $\lambda$ of
$\psi$ is a horizontal curve which satisfies $\pi\circ \lambda=\psi$,
we have $C_W(\cM)=p(C_W(T\cM))$, where $p\eqdef \pi \circ
(-):C(T\cM)\rightarrow C(\cM)$. Since $\pi\circ \lambda=\psi$, we have
$(\dd\pi)\circ \dot{\lambda}=\dot{\psi}$. When $\lambda$ is a phase
space $W$-curve, this gives $(\dd\pi)\circ S\circ \lambda=\dot{\psi}$,
where we used the relation $\pi_\ast(S_H)=\pi_\ast(S)$.

Since the horizontal lift a pointed curve $\psi$ at $t=0$ is uniquely determined
by its value $u\in T_{\psi(0)}$ at zero time, this gives a first order
ODE for $\psi$ parameterized by $u$ which characterizes
$W$-curves. More precisely, a pointed curve $\psi:I\rightarrow \cM$
with $\psi(0)=m\in \cM$ is a $W$-curve iff there exists $u\in T_m\cM$
such that the horizontal lift $\bpsi_u:I\rightarrow T\cM$ of $\psi$
through $u$ at $t=0$ is a phase space $W$-curve. Equation \eqref{Weq}
for $\psi_u$ takes the form:
\ben
\label{Weq2}
\dot{\bpsi}_u(t)=S_H(\bpsi_u(t))~~\forall t\in I
\een
and can be viewed as a first order geometric ODE for $\psi$ with
initial condition $\psi(0)=m$ and parameter $u\in T_m\cM$. This ODE
has a unique maximal solution for each parameter $u\in T_{\psi(0)}\cM$
and $\psi$ can be recovered by projecting that solution to $\cM$.  Let
$I_u$ be the interval of definition of the maximal solution of
\eqref{Weq2} at a fixed value of $u$. Then the maximal solution of
\eqref{Weq2} for parameter $u$ projects to a $W$-curve
$\psi^{(u)}:I_u\rightarrow \cM$ which satisfies $\psi^{(u)}(0)=m$.
This gives a family of $W$-curves passing through each point $m$ of
$\cM$ and parameterized by the vectors tangent to $\cM$ at that
point. Notice that this parameterization need not be single-valued. A
specific $W$-curve $\psi^{(u)}$ which passes through $m$ at $t=0$ can
be specified by selecting a tangent vector $u\in T_m\cM$. This
determines the speed of $\psi^{(u)}$ at $t=0$ through the relation:
\be
\dot{\psi}^{(u)}(0)=(\dd\pi)(S(u))~~.
\ee

\begin{remark}
Two phase space $W$-curves $\lambda_1,\lambda_2:I\rightarrow T\cM$
determine the same $W$-curve $\psi:I\rightarrow \cM$ iff both of them
project to $\psi$, which requires $(\dd\pi)\circ S\circ
\lambda_1=(\dd\pi)\circ S\circ \lambda_2=\dot{\psi}$.
\end{remark}

The {\em horizontal approximation} defined by the Ehresmann connection
$W$ consists of approximating integral curves of $S$ by phase space
$W$-curves of $S$, i.e. approximating the flow of $S$ on $T\cM$ by the
flow of $S_H$. As explained below, this amounts to approximating
solution curves $\varphi$ of the second order geometric ODE defined by
$S$ on $\cM$ by $W$-curves of $S$ whose speed approximates the
projection of the acceleration of $\varphi$. The approximation is accurate for those
integral curves $\gamma:I\rightarrow T\cM$ of $S$ which are almost
horizontal in the sense that $S_V\circ \gamma$ is close to zero (where
``closeness to zero'' is measured, for example, by the fiberwise norm
induced on $TT\cM$ by a metric on $T\cM$).

Let $\varphi:I\rightarrow \cM$ be a pointed solution curve such that
the vector $u:=\dot{\varphi}(0)$ lies on the mean field manifold
$\cZ_W$. Then $S(u)\in W$ and the acceleration of $\varphi$ at
$t=0$ is horizontal:
\be
\ddot{\varphi}(0)=S(\dot{\varphi}(0))=S(u)\in W~~\mathrm{in}~~TT\cM~~.
\ee
After passing to a subinterval of $I$ if necessary, let $\psi=\pi\circ
\lambda:I\rightarrow \cM$ be the $W$-curve which is the projection of
the unique phase space $W$-curve $\lambda:I\rightarrow T\cM$ which
satisfies $\lambda(0)=\dot{\varphi}(0)=u$. The curve $\lambda$ and the
canonical lift $\gamma=\dot{\varphi}$ of $\varphi$ satisfy:
\be
\dot{\lambda}(t)=S_H(\lambda(t))~~,~~\dot{\gamma}(t)=S(\gamma(t))~~\forall t\in I
\ee
with the initial conditions:
\be
\gamma(0)=\lambda(0)=u~~.
\ee
Thus:
\be
\psi(0)=\pi(u)=\varphi(0)~~\mathrm{and}~~\dot{\psi}(0)=(\dd\pi)(S(u))=(\dd\pi)(\ddot{\varphi}(0))
\ee
and we have:
\be
\dot{\psi}(t)=(\dd\pi)(\dot{\lambda}(t))=(\dd\pi)(S(\lambda(t)))~~\forall t\in I~~.
\ee
The approximation $\gamma(t)\approx \lambda(t)$ implies
$S(\lambda(t))\approx S(\gamma(t))=\dot{\gamma}(t)=\ddot{\varphi}(t)$
and is equivalent with the joint approximation $\varphi(t)\approx
\psi(t)~\&~\dot{\varphi}(t)\approx \dot{\lambda}(t)$, which implies
$\varphi(t)\approx \psi(t)~~\&~(\dd\pi)(\ddot{\varphi}(t))\approx
\dot{\psi}(t)$. Notice that the ``mean field vector''
$u=\dot{\varphi}(0)\in \cZ_W$ determines the $W$-curve $\psi$ which
passes through the point $\varphi(0)$ at $t=0$ and is used to
approximate the solution curve $\varphi$. Intuitively, the
approximation neglects the vertical component of the acceleration
$\ddot{\varphi}$ of $\varphi$. This approximation is exact at $t=0$
and remains accurate on a small enough neighborhood of the origin
inside the interval $I$. In particular, the phase space
$W$-approximation of flow curves is exact on the mean field manifold
for all flow curves which intersect $\cZ_W$.

\begin{definition}
A {\em mean field} for $S$ relative to $W$ is a locally defined vector
field $X\in \Gamma(TU_0)$ such that:
\ben
\label{mf}
X(m)\in \cZ_W~~\forall m\in U_0~~,
\een
where $U_0$ is an open subset of $\cM$. 
\end{definition}

\noindent If $U$ is connected, then the image $X(U)$ of a mean field
$X\in \Gamma(U,T\cM)$ is a subset of a ``branch'' of the mean field
manifold $\cZ_W$ (when the latter is viewed as a multisection of
$\pi$).  This allows one to give local parameterizations of that
branch of the mean field manifold by choosing local coordinates
defined on open subsets $U\subset \cM$. The name ``mean field'' was
chosen by analogy with similar objects which appear in condensed
matter physics.

\subsection{The mean field approximation defined by $W$}

\noindent Let $S_W\in \Gamma(TW)$ be the projection of
$S\vert_{\cZ_W}$ to the integrable distribution $T\cZ_W$, where the
projection is taken parallel to $V$.

\begin{definition}
A curve $\rho:I\rightarrow \cZ_W$ is called a {\em phase mean field
  curve} of $S$ relative to $W$ if it is an integral curve of the
vector field $S_W\in \Gamma(TW)$, i.e. if it satisfies:
\ben
\label{mfeq}
\dot{\rho}(t)=S_W(\rho(t))~~\forall t\in I~~.
\een
\end{definition}

\noindent Notice that any phase mean field curve lies on $\cZ_W$

The {\em mean field approximation} consists of approximating solution
curves of the geometric ODE defined by $S$ by phase mean field curves
of $S$ relative to $W$. This is a singular approximation in the sense
that it restricts the phase space of the dynamical system from $T\cM$
to $T\cZ_W$. The approximation is accurate for those solution curves
$\varphi$ whose flow curve $\gamma=\dot{\varphi}$ lies very close to
the mean field manifold $\cZ_W$, {\em provided that} $S\vert_{\cZ_W}$
is well-approximated by $S_W$.

Let $\varphi:I\rightarrow \cM$ be a pointed solution curve of $S$ such that
the vector $u:=\dot{\varphi}(0)$ lies on the mean field manifold
$\cZ_W$. Let $\gamma:=\dot{\varphi}:I\rightarrow T\cM$ be the canonical
lift of $\varphi$. Then:
\be
\dot{\gamma}(0)=S(\gamma(0))=S(u)\in W_{u}~~.
\ee
After passing to a subinterval of $I$ if necessary, let
$\rho:I\rightarrow \cZ_W$ be the unique phase mean field curve which
satisfies:
\be
\rho(0)=u~~.
\ee
Then:
\be
\dot{\rho}(t)=S_W(\rho(t))~~\forall t\in I
\ee
and we have:
\be
\dot{\rho}(0)=S_W(u)~~.
\ee Let $\mu\eqdef \pi\circ \rho:I\rightarrow \cM$. The approximation
$\gamma(t)\approx \rho(t)$ implies $\varphi(t)\approx \mu(t)$ and
$(\dd\pi)(S_W(\dot{\varphi}(t)))\approx \dot{\mu}(t)$ and is accurate
for $t$ close to zero only if $\dot{\gamma}(0)\approx \dot{\rho}(0)$,
i.e. only if the horizontal vector $S(u)$ is almost tangent to $\cZ_W$
and hence is well approximated by $S_W(u)$. This approximation is much
coarser than the horizontal approximation and is not guaranteed to be
accurate for any flow curve or non-degenerate time interval.

\begin{remark}
The discussion above can be adapted easily to the restriction of $S$
and $W$ to a non-empty open subset of $T\cM$. We leave the details of this
generalization to the reader.
\end{remark}

\subsection{Fiberwise coordinates and special coordinate systems on $\dot{T}\cM$}
\label{subsec:fiberwise}

\noindent Suppose for simplicity that $\dim \cM=2$, which is the case
relevant for the remainder of the paper. Let $U$ be a non-empty open
subset of $T\cM$. Then $U_0\eqdef \pi(U)$ is an open subset of $\cM$
since $\pi$ is a submersion and hence an open map.

\begin{definition}
\label{def:fiberwise}
A {\em system of fiberwise coordinates} defined on $U$ is a (generally
non-surjective) submersion $y=(y^1,y^2):U\rightarrow \R^2$ whose
restriction to every non-empty intersection of $U$ with a fiber of
$T\cM$ is a system of coordinates on that fiber. That is, the map:
\be
y\vert_{\pi^{-1}(m)\cap U}:\pi^{-1}(m)\cap U\rightarrow \R^2
\ee
is a diffeomorphism to its image $y(\pi^{-1}(m)\cap U)$ for all $m\in U_0$~~.
\end{definition}

\noindent A particular case arises when $y(\pi^{-1}(m)\cap U)=y(U)$ for all $m\in U$.

\begin{prop}
Let $y:U\rightarrow \R^2$ be a smooth map. Then the
following statements are equivalent:
\begin{enumerate}[(a)]
\item $y$ is a fiberwise coordinate system such that
  $y(\pi^{-1}(m)\cap U)=y(U)$ for all $m\in U_0\eqdef \pi(U)$.
\item The restriction $\pi_U\eqdef \pi\vert_U:U\rightarrow U_0$ is a
  topologically trivial fiber bundle with typical fiber $y(U)$ and the
  map $(\pi_U,y):U\rightarrow U_0\times y(U)$ is a trivialization of
  this fiber bundle, where $U_0\times y(U)$ is viewed as trivial fiber
  bundle over $U_0$ using projection on the first factor.
\end{enumerate}
In this case, $\pi_U$ is fiber sub-bundle of the vector bundle $T\cU_0$.
\end{prop}

\begin{proof}
To see that $(a)$ implies $(b)$, notice that all fibers of the surjective
submersion $\pi_U$ are diffeomorphic with $y(U)$ and $(\pi_U,y)$
is an isomorphism of fiber bundles to the trivial $y(U)$-bundle over
$U_0$. The converse implication follows from the definition of
trivializations of fiber bundles.
\end{proof}

\noindent An interesting class of examples arises from local frames of $T\cM$.

\begin{definition}
Let $U_0$ be a non-empty open subset of $\cM$. A surjective system of fiberwise
coordinates $y=(y^1,y^2):\pi^{-1}(U_0)=TU_0\rightarrow \R^2$ is called
{\em linear} if the restriction of $y$ to every fiber of $TU_0$ is a
linear isomorphism.
\end{definition}

\begin{prop}
The following statements are equivalent for a map $y:TU_0\rightarrow \R^2$.
\begin{enumerate}[(a)]
\item $y$ is a linear system of fiberwise coordinates.
\item There exists a frame $(e_1,e_2)$ of $TU_0$ such that:
\be
u=y^1(u)e_1(\pi(u))+y^2(u)e_2(\pi(u))~~\forall u\in TU_0
\ee
\end{enumerate}
\end{prop}

\noindent Notice that the frame $(e_1,e_2)$ need not be holonomic
(i.e. it need not be induced by a local coordinate system on $\cM$).

\begin{proof}
Suppose that $(a)$ holds. For any $m\in U_0$, the restrictions $y^i_m:
T_m\cM\rightarrow \R$ give linear coordinates on $T_m\cM$. If $e_1(m),
e_2(m)\in T_m\cM$ are the unique vectors which satisfy
$y(e_1(m))=(1,0)$ and $y(e_2(m))=(0,1)$ then any $u\in T_m\cM$ can be
written as: \be u=y^1_m(u) e_1(m)+y^2_m(u) e_2(m)~~ \ee and hence
$(e_1(m),e_2(m))$ forms a basis of the fiber $T_m\cM$. This gives a
frame $(e_1,e_2)$ of $TU_0$ and hence $(b)$ holds. Now suppose that
$(b)$ holds. Then we have $y^i_m(u)=u^i$ for $m\in U_0$ and $u\in
T_mU_0$, where $u=u^1 e_1(m)+u^2 e_2(m)$ is the frame expansion of
$u$. It is clear that $y=(y^1,y^2)$ is a linear system of fiberwise
coordinates.
\end{proof}

\begin{definition}
\label{def:special}
A coordinate system $x=(y^1,y^2,\xi^1,\xi^2):U\rightarrow \R^4$
defined on an open subset $U\subset T\cM$ is called {\em special} if
it satisfies the following two conditions:
\begin{enumerate}[1.]
\item $(y^1,y^2):U\rightarrow \R^2$ is a system of fiberwise
  coordinates defined on $U$.
\item There exists a system of coordinates
  $x=(x^1,x^2):U_0\eqdef\pi(U)\rightarrow \R^2$ such that
  $\xi^1=(x^1)^\v\eqdef x^1\circ \pi_U$ and $\xi^2=(x^2)^\v\eqdef
  x^2\circ \pi_U$, where $\pi_U\eqdef \pi\vert_U$.
\end{enumerate}
\end{definition}

\noindent Given any system of fiberwise coordinates $(y^1,y^2)$ on
$U\subset T\cM$ and any system of coordinates $(x^1,x^2)$ on
$U_0=\pi(U)$, it is clear that $(y^1,y^2,\xi^1\eqdef
(x^1)^\v,\xi^2\eqdef (x^2)^\v)$ is a special coordinate system on
$U$.

\subsection{The horizontal approximation defined by a special system of coordinates}

\noindent Let $(y^1,y^2,\xi^1=(x^1)^\v,\xi^\v=(x^2)^\v):U\rightarrow
\R^4$ be a special system of coordinates defined on the open subset
$U\subset T\cM$ and let $U_0\eqdef \pi(U)$. The rank two integrable
distribution $W$ spanned by the vector fields $\frac{\pd}{\pd \xi^1}$
and $\frac{\pd}{\pd \xi^2}$ on $U$ is complementary to the vertical
distribution $V$ of the submersion $\pi_U:U\rightarrow U_0$ and
defines a complete {\em flat} Ehresmann connection on $U$.
As in the previous subsection, this gives a decomposition $S=S_V+S_H$ of any
semispray $S$ defined on $U$. Write:
\be
S=S^{y^i}\frac{\pd}{\pd y^i}+S^{\xi^i}\frac{\pd}{\pd \xi^i}~~,
\ee
where we use implicit summation over repeated indices and $i\in
\{1,2\}$.  Since the vector fields $\frac{\pd}{\pd y^i},\frac{\pd}{\pd
  \xi^i}\in \cX(U)$ are respectively vertical and horizontal, we
have:
\be
S_V=S^{y^i}\frac{\pd}{\pd y^i}~\mathrm{and}~~S_H=S^{\xi^i}\frac{\pd}{\pd \xi^i}~~.
\ee
To identify the horizontal part of $S$, let
$\zeta=(\zeta^1,\zeta^2):TU_0\rightarrow \R^2$ be the linear system of
fiberwise coordinates defined by the frame $\left(\frac{\pd}{\pd x^1},
\frac{\pd}{\pd x^2}\right)$ of $TU_0$. This is known as the {\em
  tangent chart} of the chart $(x^1,x^2):U_0\rightarrow \R^2$ of
$\cM$. The change of coordinates from $(\zeta,\xi)$ to $(y,\xi)$ and
its inverse have the forms:
\ben
\label{coordchange}
y^i=\ry^i(\zeta,\xi)~,~~\xi^i=\xi^i~~\mathrm{and}~~\zeta^i=\rzeta^i(y,\xi)~~,~~\xi^i=\xi^i~,
\een
where $\ry^i$ and $\rzeta^i$ are functions. In particular, we have $\frac{\pd \xi^i}{\pd y^j}=0$ and hence:
\be
\frac{\pd}{\pd y^i}=\frac{\pd \rzeta^j}{\pd y^i}\frac{\pd}{\pd \zeta^j}~~.
\ee
Since $S$ is a semispray, its horizontal part is given by the
following formula in the tangent chart $(\zeta,\xi)$ (see \cite{SLK}):
\be
S_H=\zeta^i\frac{\pd}{\pd \xi^i}~~.
\ee
Thus:
\be
S^{\xi^i}=\zeta^i~~.
\ee

For any curve $\gamma:I\rightarrow U$, we
have:
\be
\dot{\gamma}(t)=\dot{y}^i(t)\frac{\pd}{\pd y^i}\Big{\vert}_{\gamma(t)}+\dot{\xi}^i(t)\frac{\pd}{\pd \xi^i}\Big{\vert}_{\gamma(t)}~~,
\ee
where $y^i(t)\eqdef y^i(\gamma(t))$ and $\xi^i(t)\eqdef
\xi^i(\gamma(t))$. Hence:
\be
\dot{\gamma}(t)_V=\dot{y}^j(t)\frac{\pd}{\pd y^i}\Big{\vert}_{\gamma(t)}~~,~~\dot{\gamma}(t)_H=\dot{\xi}^i(t)\frac{\pd}{\pd \xi^i}\Big{\vert}_{\gamma(t)}
\ee
and the integral curve equation of $S$ amounts to
the following system of ODEs:
\ben
\label{yxisys}
\dot{y}^i(t)=S^{y^i}(y(t),\xi(t))~~,~~\dot{\xi}^i(t)=S^{\xi^i}(y(t),\xi(t))=\rzeta^i(y(t),\xi(t))~~.
\een

The mean field locus $\cZ_W$ of $S$ relative to $W$ has equations:
\be
S^{y_1}(y,\xi)=S^{y_2}(y,\xi)=0~~.
\ee
As before, we assume that $S$ is regular relative to $W$, i.e. $S_V$
is transverse to the zero section of $V$. This amounts to the
condition that the map $(S^{y^1},S^{y^2}):U\rightarrow \R^2$ restricts
to a submersion on a tubular neighborhood of $\cZ_W$ in $U$. Writing a vector
field $X\in \cX(U_0)$ as $X=X^i\frac{\pd}{\pd x^i}$ with $X^i\in
\cC^\infty(U_0)$, we have $\xi^i\circ X=x^i$ and $\zeta^i\circ X=X^i$.
Hence the mean field equation \eqref{mf} reads:
\be
S^{y^i}(\ry^1(X^1(x),X^2(x),x),\ry^2(X^1(x),X^2(x),x),x)=0 ~~(i=1,2)~~,
\ee
where $\ry^i$ are the functions appearing in the coordinate
transformations \eqref{coordchange}. This system of functional
equations locally determines $X^i$ as functions of $x$, though the
solution need not be globally unique.

Let $\lambda:I\rightarrow U$ be a curve in $U$. Setting $y^i(t)\eqdef
y^i(\lambda(t))$ and $\xi^i(t)\eqdef \xi^i(\lambda(t))$, the phase space
$W$-curve equation \eqref{Weq} amounts to the system:
\be
\dot{y}^i(t)=0~~\mathrm{and}~\dot{\xi}^i(t)=S^{\xi^i}(y(t),\xi(t))=\rzeta^i(y(t),\xi(t))~~\forall t\in I~~.
\ee
Intuitively, the horizontal approximation defined by the special
coordinates $(y,\xi)$ amounts to treating the variables $y^1$ and
$y^2$ as `slow' in the sense that $\dot{y}^1$ and $\dot{y}^2$
are neglected. This corresponds to neglecting the following combinations:
\be
\frac{\partial
  \ry^i}{\partial \zeta^j}(\dot{x}(t),x(t))\ddot{x}^j(t)+\frac{\partial
  \ry^i}{\partial \xi^j}(\dot{x}(t),x(t))\dot{x}^j(t)~~(i=1,2)
\ee
of the components of the speed and acceleration of a solution curve
$\varphi$ of the geometric second order ODE
$\ddot{\varphi}(t)=S(\dot{\varphi})(t)$ defined by $S$ on $U_0$. Here
$x^i(t)\eqdef x^i(\varphi(t))$ and we noticed that
$\xi^i(\dot{\varphi}(t))=x^i(t)$ and
$\zeta^i(\dot{\varphi}(t))=\dot{x}^i$.

\section{Two-field cosmological models}
\label{sec:models}

\noindent By a {\em two-field cosmological model} we mean a
classical cosmological model with two real scalar fields derived from the
following Lagrangian on a spacetime with topology $\R^4$:
\ben
\label{cL}
\cL[g,\varphi]=\left[\frac{M^2}{2} \mathrm{R}(g)-\frac{1}{2}\Tr_g \varphi^\ast(\cG)-\Phi\circ \varphi\right]\vol_g~~.
\een
Here $M$ is the reduced Planck mass, $g$ is the spacetime metric on
$\R^4$ (taken to be of ``mostly plus'' signature) and $\vol_g$ and
$\mathrm{R}(g)$ are the volume form and scalar curvature of $g$. The
scalar fields are described by a smooth map $\varphi:\R^4\rightarrow
\cM$, where $\cM$ is a (generally non-compact) connected, smooth and
paracompact surface without boundary which we assume to be oriented
and endowed with a smooth Riemannian metric $\cG$, while
$\Phi:\cM\rightarrow \R$ is a smooth function which plays the role of
potential for the scalar fields. One usually requires that $\cG$ is
complete to ensure conservation of energy. For simplicity, we also
assume that $\Phi$ is strictly positive and non-constant on $\cM$. The
triplet $(\cM,\cG,\Phi)$ is called the {\em scalar triple} of the
model, while $(\cM,\cG)$ is called the {\em scalar manifold}.  The
model is parameterized by the quadruplet $\fM\eqdef
(M_0,\cM,\cG,\Phi)$, where:
\be
M_0\eqdef M\sqrt{\frac{2}{3}}
\ee
is the {\em rescaled Planck mass}. We denote by:
\be
\Crit\Phi\eqdef \{m\in \cM~\vert~(\dd \Phi)(m)=0\}~~\mathrm{and}~~\cM_0\eqdef \cM\setminus \Crit\Phi
\ee
the {\em critical and noncritical subsets} of $\cM$. Notice that
$\cM_0$ is open in $\cM$. It is also non-empty since we assume that
$\Phi$ is not constant.

The two-field model parameterized by the quadruplet $\fM$ is obtained
by assuming that $g$ is an FLRW metric with flat spatial section:
\ben
\label{FLRW}
\dd s^2_g=-\dd t^2+a(t)^2\sum_{i=1}^3 \dd x_i^2
\een
(where $a$ is a smooth and positive function of $t$) and that
$\varphi$ depends only on the {\em cosmological time} $t\eqdef
x^0$. Here $(x^0,x^1,x^2,x^3)$ are the Cartesian coordinates on
$\R^4$.  Define the {\em Hubble parameter} through:
\be
H(t)\eqdef \frac{\dot{a}(t)}{a(t)}~~,
\ee
where the dot indicates derivation with respect to $t$.

\subsection{The cosmological equation}

\noindent Let $\pi:T\cM\rightarrow \cM$ be bundle projection. Let
$\cH:T\cM\rightarrow \R_{>0}$ be the {\em rescaled Hubble function} of the scalar triple
$(\cM,\cG,\Phi)$, which is defined through:
\ben
\label{cHdef}
\cH(u)\eqdef \frac{1}{M_0}\left[||u||^2+2\Phi(\pi(u))\right]^{1/2}~~\forall u\in T\cM~~
\een
and let $\nabla$ be the Levi-Civita connection of $(\cM,\cG)$. When $H>0$
(which we assume for the remainder of the paper), the Euler-Lagrange
equations of \eqref{cL} reduce (see \cite{genalpha, Noether1, Noether2})
to the {\em cosmological equation}:
\ben
\label{eomsingle}
\nabla_t \dot{\varphi}(t)+ \cH(\dot{\varphi}(t))\dot{\varphi}(t)+ (\grad \Phi)(\varphi(t))=0~~,
\een
(where $\nabla_t\eqdef \nabla_{\dot{\varphi}(t)}$) together with the condition:
\ben
\label{Hvarphi}
H(t)= \frac{1}{3}\cH(\dot{\varphi}(t))~~.
\een
The smooth solutions $\varphi:I\rightarrow \cM$ of \eqref{eomsingle}
(where $I$ is a non-degenerate interval) are called {\em cosmological
  curves}, while their images in $\cM$ are called {\em cosmological
  orbits}.  Given a cosmological curve $\varphi$, relation
\eqref{Hvarphi} determines the function $a$ up to a multiplicative
constant.

A cosmological curve $\varphi:I\rightarrow \cM$ need not be an
immersion.  Accordingly, we define the {\em singular and regular
  parameter sets} of $\varphi$ through:
\beqa
I_\sing\eqdef \{t\in I~\vert~\dot{\varphi}(t)=0\}~~,~~I_\reg\eqdef I\setminus I_\sing= \{t\in I~\vert~\dot{\varphi}(t)\neq 0\}~~.
\eeqa
When $\varphi$ is not constant, it can be shown that $I_\sing$ is a discrete subset of $I$.
The sets of {\em critical and noncritical times} of $\varphi$ are defined through:
\beqa
I_\crit\eqdef \{t\in I~\vert~\varphi(t)\in \Crit\Phi \}~~,~~I_\noncrit \eqdef I\setminus I_\crit \eqdef \{t\in I~\vert~\varphi(t)\not\in \Crit\Phi\}~~.
\eeqa
The cosmological curve $\varphi$ is called {\em nondegenerate} if
$I_\sing=\emptyset$ and {\em noncritical} if
$I_\crit=\emptyset$. Noncriticality means that the orbit of $\varphi$
does not meet the critical set of $\Phi$, i.e. $\varphi(I)$ is
contained in $\cM_0$. It is easy to see that a cosmological curve
$\varphi:I\rightarrow \cM$ is constant iff its image coincides with a
critical point of $\Phi$, which in turn happens iff there exists some
$t\in I_\sing\cap I_\crit$.  Hence for any non-constant cosmological
curve we have:
\be
I_\sing\cap I_\crit=\emptyset~~.
\ee

\subsection{The cosmological dynamical system}

\noindent The cosmological equation \eqref{eomsingle} is {\em
  geometric} in the sense that it can be written as the condition
$j^2(\varphi)=\fs(j^1(\varphi))$, where $\fs$ is a section of the
fiber bundle $J^2(\cM)\stackrel{\pi_{21}}{\rightarrow} J^1(\cM)$ (see
Subsection \ref{subsec:obs}). This statement follows from the fact
that \eqref{eomsingle} is expressed in terms of
diffeomorphism-covariant operators and from the fact that the double
derivative term in \eqref{eomsingle} separates from the lower order
terms in any local system of coordinates.

The geometric character of \eqref{eomsingle} implies that the
cosmological equation is invariant under changes of local coordinates
on $\cM$ and hence its first order form defines a geometric dynamical
system on the total space of $T\cM$ in the sense of
\cite{Palis,Katok}, which we call the {\em cosmological dynamical
  system} defined by the scalar triple $(\cM,\cG,\Phi)$. As explained
for example in \cite{SLK}, the defining vector field $S\in \cX(T\cM)$
of this dynamical system is a second-order vector field which we call
the {\em cosmological semispray}. We have (see \cite{ren}):
\be
S=\mathbb{S}-Q~~,
\ee
where $\mathbb{S}$ is the geodesic spray of the Riemannian manifold
$(\cM,\cG)$ and $Q\in \cX(T\cM)$ is the {\em cosmological correction},
which is the vertical vector field:
\be
Q=\cH C+(\grad \Phi)^v~~.
\ee
Here $C\in \cX(T\cM)$ is the Euler vector field of $T\cM$ and
$(\grad\Phi)^v\in T\cM$ is the vertical lift of the vector field
$\grad\Phi\in \cX(\cM)$ to a vector field defined on $T\cM$. The
cosmological equation \eqref{eomsingle} is equivalent with the
integral curve equation of $S$ in the following sense (see
\cite{ren}):

\begin{prop}
Let $\varphi:I\rightarrow \cM$ be a curve in $\cM$. Then the following
statements are equivalent:
\begin{enumerate}[(a)]
\item $\varphi$ satisfies the cosmological equation \eqref{eomsingle}
\item The canonical lift $\gamma\eqdef \dot{\varphi}:I\rightarrow T\cM$ of
  $\varphi$ is an integral curve of the vector field $S$, i.e. it
  satisfies the equation:
\ben
\label{Seq}  
\dot{\gamma}(t)=S(\gamma(t))~~\forall t\in I~~.
\een
\end{enumerate}
\end{prop}

\noindent The flow of $S$ on the total space of $T\cM$ is called the
          {\em cosmological flow} of the model. Following dynamical
          system terminology, we will refer to $T\cM$ as the {\em
            phase space} of the cosmological dynamical system. The
          integral curves of $S$ are called {\em cosmological flow
            curves} while their images in $T\cM$ are called {\em
            cosmological flow orbits}. It can be shown that any
          cosmological {\em flow} orbit is aperiodic and without
          self-intersections. This follows from the fact that the
          cosmological dynamical system is dissipative (see
          Proposition \ref{prop:dissip} below).

\subsection{Local cosmological observables}
\label{subsec:obs}

\noindent For any $k\in \Z_{\geq 0}$, let $E^k\eqdef J^k(\cM)$ be the
$k$-th jet bundle of curves in $\cM$, whose projection to $\cM$ we
denote by $\pi_k$. By definition, the fiber of $E^k$ at $m\in \cM$ is
the $k$-th jet space $E_m^k=J^k_m(\cM)$ of curves in $\cM$ which pass
through the point $m$ at time $t=0$. Let $F^k\eqdef \pi_k^\ast(T\cM)$
be the {\em $k$-th Finsler bundle} of $\cM$, whose projection to $E^k$
we denote by $p_k$. We have $E^0=\cM$, $E^1=T\cM$ and $F^0=T\cM$,
$F^1=F$, where $F=\pi^\ast(T\cM)$ is the ordinary Finsler bundle of
$\cM$, i.e.  the tautological bundle of $T\cM$. For any $k> l\geq 0$,
let $\pi_{kl}:E^k\rightarrow E^l$ be the natural projection, which is
a fiber bundle. We have $\pi_{k,0}=\pi_k$ and $\pi_{l,m}\circ
\pi_{k,l}=\pi_{l,m}$ for $k>l>m$. Notice that
$\pi_{k,k-1}:E^k\rightarrow E^{k-1}$ is an affine bundle of rank $2$
which is modeled on the vector bundle $F^{k-1}$ (see
\cite{Saunders}). For any smooth curve $\varphi:I\rightarrow \cM$, let
$j^k(\varphi):I\rightarrow J^k(\cM)$ denote the $k$-th jet
prolongation of $\varphi$. We have $j^0(\varphi)=\varphi$ and
$j^1(\varphi)=\dot{\varphi}$.

\begin{remark}
The product $\R\times E^k\rightarrow \R$ identifies with the jet
bundle of sections of the trivial bundle $\R\times \cM\rightarrow \R$,
where the projection to $\R$ is the first projection of the Cartesian
product.
\end{remark}

\noindent Let $\cE$ be a vector bundle over $\cM$ and set $\cE^k\eqdef
\pi_k^\ast(\cE)$ for all $k\in \Z_{\geq 0}$. 

\begin{definition}
A (continuous and off-shell) {\em $\cE$-valued local cosmological
  observable} of order $k\in \Z_{\geq 0}$ is a continuous section
$s\in \cC^0(U,\cE^k)$, where $U$ is a non-empty open subset of
$E^k=J^k(\cM)$. A local cosmological observable of order one is called
{\em basic}; in this case, $s$ is a section of $\pi^\ast(\cE)$ defined
on an open subset of $J^1(\cM)=T\cM$.
\end{definition}

\noindent Local cosmological observables valued in the trivial real
line bundle $\R_\cM$ of $\cM$ are called {\em scalar local
  observables}, while $T\cM$-valued local observables are called {\em
  vector local observables}. In the first case, $\cE^k$ is isomorphic
with the trivial line bundle $\R_{E^k}$ over $E^k$, so a scalar local
observable is a function $f:U\rightarrow \R$, where $U$ is an open
subset of $E^k$. In the second case, we have $\cE^k=F^k$, so a local
vector observable is a section $s\in \Gamma(U,F^k)$, where $U$ is an
open subset of $E^k$. Notice that $\cE$-valued local observables of
order zero are local sections of the vector bundle $\cE$.

\begin{eg}
A simple example of scalar basic cosmological observable is the norm
function $||~||:T\cM\rightarrow \R_{\geq 0}$ of $(\cM,\cG)$, which is
continuous everywhere but smooth only on the slit tangent bundle
${\dot T}\cM$. Another simple example of such is the rescaled Hubble
function $\cH=\frac{1}{M_0}\sqrt{\cN^2+2\Phi\circ \pi}$, which is
smooth on $T\cM$.  Here $\cN:T\cM\rightarrow \R_{\geq 0}$ is the norm
function induced by $\cG$ (see Subsection \ref{subsec:notations}).
\end{eg}

\begin{definition}
\label{def:normred}
Let $U\subset T\cM$ be a non-empty open set and $U_0\eqdef \pi(U)$. A
basic scalar local observable $f:U\rightarrow \R$ is
called {\em norm reducible} if there exists a map $g:U_0\times
\R_{\geq 0}\rightarrow \R$ such that: $f=g\circ (\pi\times\cN)$, i.e.:
\be
f(u)=g(\pi(u),||u||)~~\forall u\in U~~.
\ee
In this case, we have:
\be
f(u)=g_m(||u||)~~\forall m\in U_0~~\forall u\in U\cap T_m\cM~~,
\ee
where $g_m:\R_{\geq 0}\rightarrow \R$ is the $m$-section of $g$:
\be
g_m(u)\eqdef g(m,||u||)~~\forall u\in U\cap T_m\cM~~\forall m\in U_0~~.
\ee
\end{definition}

\begin{eg}
The rescaled Hubble function $\cH$ is norm reducible.
\end{eg}

\begin{definition}
Given an $\cE$-valued local cosmological observable $s\in
\cC^0(U,\cE^k)$ with $U\subset E^k$ and a curve $\varphi:I\rightarrow
\cM$ such that $j^k(\varphi)(I)\subset U$, the {\em evaluation} of $s$
along $\varphi$ is defined though:
\be
s_\varphi \eqdef j^k(\varphi)^\ast(s)\in \cC^0(I,\varphi^\ast(\cE))~~.
\ee
\end{definition}

\noindent In the definition above, we used the isomorphism
$(j^k(\varphi))^\ast(\cE^k)=(j^k(\varphi))^\ast(\pi_k^\ast(\cE))\simeq
\varphi^\ast(\cE)$, which follows from the relation $\pi_k\circ
j^k(\varphi)=\varphi$. Notice that $s_\varphi$ can be viewed as a map
$s_\varphi:I\rightarrow \cE$ which satisfies $s_\varphi(t)\in
\cE_{\varphi(t)}$ for all $t\in I$.

\subsubsection{The cosmological section, cosmological shell and dynamical reduction of local observables.}

The cosmological equation \eqref{eomsingle} defines the {\em
  cosmological section} $\fs:E^1\rightarrow E^2$ of the fiber bundle
$\pi_{21}:E^2\rightarrow E ^1$, which has the following expression in
the canonical local coordinates
$(x_0^i=x^i,x_1^i=\dot{x}^i,x_2^i=\ddot{x}^i)_{i=1,\ldots, d}$ on
$E^2$ induced by local coordinates $x=(x^i)_{i=1,2}$ on $\cM$ (here
the dots denote formal derivatives).
\beqan
\label{fs}
&& \fs^{x_2^i}(x,x_1)=-\Gamma^i_{jk}(x)x_1^jx_1^k-\frac{1}{M_0} \left[\cG_{kl}(x)x_1^k x_1^l+2\Phi(x)\right]^{1/2}x_1^i- \cG^{ij}(x)(\pd_j \Phi)(x)~~\nn\\
&& \fs^{x_1^i}(x,x_1)=x_1^i\\
&& \fs^{x_0^i}(x,x_1)=x_0^i~~,\nn
\eeqan
where $\fs^{x^i_k}\eqdef x^i_k\circ \fs$ for $k=0,1,2$ and $i=1,2$. Here $\Gamma^i_{jk}$ are the
Christoffel symbols of $\cG$ in the local coordinates $x^i$ on
$\cM$. The image of this section is the closed submanifold $\fS$ (called
the {\em cosmological shell}) of the total space of $E^2=J^2(\cM)$ defined
by the cosmological equation. The cosmological shell has equations:
\ben
\label{shelleq}
x_2^i= {\hat \fs}_1^i(x,x_1)~~,
\een
where we set ${\hat \fs}_1^i\eqdef \fs^{x_2^i}$. Formal differentiation of \eqref{shelleq} with respect to time gives:
\be
x_3^i=(\pd_{x_0^j}{\hat \fs}_1^i)(x_0,x_1)x_1^j+(\pd_{x_1^j} {\hat \fs}_1^i)(x_0,x_1){\hat \fs}_1^j(x_0,x_1)\eqdef {\hat \fs}_2^i(x_0,x_1)
\ee
and generally:
\ben
\label{eomsk}
x_{k+1}^i={\hat \fs}_k^i(x_0,x_1)~~\forall k\geq 1~~,
\een
for certain maps ${\hat \fs}_k^i:E^1\rightarrow \R$. The map $\fs_k:E_1\rightarrow E^{k+1}$ with components:
\beqa
&& \fs_k^{x_0^i}(x_0,x_1)=x_0^i~~,~~ \fs_k^{x_1^i}(x_0,x_1)=x_1^i\\
&&\fs_k^{x_{l+1}^i}(x_0,x_1)={\hat \fs}_l^i(x_0,x_1)~~\forall l\in \{1,\ldots, k\}
\eeqa
defines a section of the fiber bundle $E^{k+1}\rightarrow
E^1$ and we have $\fs_1=\fs$. Notice that $\fs_k$ depends on the $k$-th jets of $\cG$ and
$\Phi$.

For any $k\geq 1$, the image of $\fs_k$ is a closed submanifold $\fS_k$
of $E^{k+1}$ which gives the $(k-1)$-th prolongation of the
cosmological shell and we have $\fS_1=\fS$. The equations:
\be
x_l^i=\fs_k^{x_l^i}(x_0,x_1)~~(l=0,\ldots, k+1, i=1,2)
\ee
give the $(k-1)$-th prolongation of the cosmological equation. A curve
$\varphi:I\rightarrow \cM$ satisfies the cosmological equation iff
$j^2(\varphi)(I)\subset \fS$, which amounts to the condition
$j^2(\varphi)(t)=\fs(j^1(\varphi)(t))$ for all $i\in I$.  In this
case, we also have $j^{k+1}(\varphi)(I)\subset \fS_k$ for all $k\geq
1$, which amounts to
\ben
\label{jkphi}
j^{k+1}(\varphi)(t)=\fs_k(j^1(\varphi)(t))~~\forall t\in I~~.
\een

\begin{definition}
The {\em dynamical reduction} (or {\em on-shell reduction}) of an
$\cE$-valued cosmological observable $s\in \cC^0(U,\cE^k)$ of order
$k\geq 2$ (where $U$ is an open subset of $E^k$) is the basic
cosmological observable:
\be
s^\red\eqdef \fs_{k-1}^\ast(s)\in \cC^0(\fs_{k-1}^{-1}(U),\cE^1)~~.
\ee
\end{definition}

\noindent In the definition above, we used the relation
$\fs_k^\ast(\cE^k)=\fs_k^\ast(\pi_k^\ast(\cE))\simeq (\pi_k\circ
\fs_k)^\ast(\cE)=\pi^\ast(\cE)=\cE^1$, which follows from the relation
$\pi_k=\pi\circ \pi_{k,1}$ and from the fact that $\fs_k$ is a section
of the fiber bundle $\pi_{k,1}:E^k\rightarrow E^1$. 

\begin{prop}
Consider a $\cE$-valued local cosmological observable $\fs\in
\cC^0(U,\cE^k)$ of order $k\geq 2$, where $U$ is a non-empty open
subset of $E^k$ which satisfies $(\fs_{k-1}\circ \pi_{k,1})(U)\subset
U$. Then for any cosmological curve $\varphi:I\rightarrow \cM$ such
that $j^k(\varphi)(I)\subset U$, we have $\dot{\varphi}(I)\subset
\fs_{k-1}^{-1}(U)$ and:
\ben
\label{sred}
s_\varphi=s^\red_\varphi=\dot{\varphi}^\ast(s^\red)\in \cC^0(I,\varphi^\ast(\cE)) ~~.
\een
\end{prop}

\begin{proof}
  Since $\dot{\varphi}=\pi_{k,1}\circ j^k(\varphi)$, we have
  $\fs_{k-1}(\dot{\varphi}(t))=(\fs_{k-1}\circ
  \pi_{k,1})(j^k(\varphi)(t))\in (\fs_{k-1}\circ \pi_{k,1})(U)\subset
  U$, which implies $\dot{\varphi}(I)\subset \fs_{k-1}^{-1}(U)$. Thus
  $\dot{\varphi}^\ast(s^\red)$ is defined on the entire interval $I$.
  By the definition of $s_\varphi$, we have:
  \be
  s_\varphi=j^k(\varphi)^\ast(s)=j_1^\ast(\varphi)(\fs_{k-1}^\ast(s))=j_1^\ast(\varphi)(s^\red)=\dot{\varphi}^\ast(s^\red)~~,
  \ee
  where we used \eqref{jkphi} and the fact that $j_1(\varphi)=\dot{\varphi}$. 
\end{proof}

\noindent The condition that $s$ be continuous (though natural for
most observables of physical interest) is inessential in the
discussion above. It is easy to see that the results of this
subsection also apply to discontinuous sections of $\cE^k$, which
define ``discontinuous local observables of order $k$''. We will
encounter some examples of such in latter sections.

\begin{remark}
\label{rem:jets}
Since $E_k=J^k(\cM)$ consists of $k$-jets of pointed curves at $t=0$,
any $\cE$-valued observable $s\in \cC^0(U,\cE^k)$ of order $k$ is
completely determined by specifying its evaluation $s_\varphi$ on {\em
  arbitrary} curves $\varphi:I\rightarrow U$. We will use this fact in
latter sections to describe second order observables through their
evaluation on arbitrary curves instead of giving their jet bundle
description.
\end{remark}

\subsubsection{The cosmological energy function.}

\begin{definition} The {\em cosmological energy function}
is the scalar basic local cosmological observable $E:T\cM\rightarrow
\R_{>0}$ defined through:
\be
E(u)\eqdef \frac{1}{2}||u||^2+\Phi(\pi(u))~~\forall u\in T\cM~~.
\ee
The {\em cosmological kinetic and potential energy functions} and the
{\em cosmological potential energy function} are the basic scalar
local cosmological observables $E_\kin,E_\pot:T\cM\rightarrow \R_{\geq
  0}$ defined through:
\be
E_\kin(u)\eqdef \frac{1}{2}||u||^2~~,~~E_\pot(u)\eqdef \Phi(\pi(u))~~\forall u\in T\cM~~.
\ee
\end{definition}

\noindent With these definitions, we have:
\be
E=E_\kin+E_\pot~~.
\ee
Notice that $E_\pot$ coincides with the natural lift of $\Phi$ to $T\cM$:
\be
E_\pot=\Phi^\v\eqdef \Phi\circ \pi~~.
\ee
\noindent Also notice the relation:
\be
\cH=\frac{1}{M_0}\sqrt{2E}~~.
\ee

\begin{prop}
\label{prop:dissip}
The evaluation of the cosmological energy along any cosmological curve
$\varphi:I\rightarrow \cM$ satisfies the {\em cosmological dissipation
  equation}:
\ben
\label{DissipEq}
\frac{\dd
E_\varphi(t)}{\dd t}=-\frac{\sqrt{2E_\varphi(t)}}{M_0}||\dot{\varphi}(t)||^2\Longleftrightarrow \frac{\dd \cH_\varphi(t)}{\dd t}=-\frac{||\dot{\varphi}(t)||^2}{M_0^2}~~\forall t\in I~~.
\een
In particular, $E_\varphi$ is a strictly decreasing function of $t$ if $\varphi$ is not constant.
\end{prop}

\begin{proof}
Follows immediately by using the cosmological equation \eqref{eomsingle}.
\end{proof}

\noindent The cosmological dissipation equation is equivalent with:
\ben
\label{dcH}
\frac{\dd \cH_\varphi(t)}{\dd t}=-||\dot{\varphi}(t)||^2~~.
\een
Integrating \eqref{dcH} from $t_1$ to $t_2$ along the cosmological curve $\varphi$ gives:
\be
\cH_\varphi(t_1)-\cH_\varphi(t_2)=\int_{t_1}^{t_2}\dd t ||\dot{\varphi}(t)||^2
\ee
Since the integral in the right hand side is minimized by the smallest length geodesic
which connects the points $\varphi(t_1)$ to $\varphi(t_2)$ in $\cM$, we have:
\be
\int_{t_0}^{t_1}\dd t ||\dot{\varphi}(t)||^2\geq \mathrm{d}(\varphi(t_1),\varphi(t_2))^2~~,
\ee
where $\mathrm{d}$ is the distance function \cite{Petersen} of
$(\cM,\cG)$. This gives the {\em Hubble inequality}:
\ben
\label{cHineq}
\cH_\varphi(t_1)-\cH_\varphi(t_2)\geq \mathrm{d}(\varphi(t_1),\varphi(t_2))^2~~.
\een

\section{The adapted frame of the non-critical submanifold}
\label{sec:adapted}

In this section, we discuss a certain frame of $T\cM_0$ which is induced by $\cG$ and $\Phi$. 
Let: 
\ben
n\eqdef \frac{\grad \Phi}{||\grad \Phi||}\in \fX(\cM_0)
\een
be the normalized gradient vector field of $\Phi$, which is
well-defined on the non-critical submanifold $\cM_0$. Let $\tau\in
\fX(\cM_0)$ be the normalized vector field which is orthogonal to $n$
and chosen such that the frame $(n,\tau)$ of $T\cM_0$ is positively
oriented at every point of $\cM_0$. In the notations of Appendix
\ref{app:signed}, we have $\tau=N_n$. Thus (see equation \eqref{Nv}):
\ben
\label{taudef}
\tau=(\iota_n\vol\vert_{\cM_0})^\sharp~~,
\een
where $\vol$ is the volume form of $(\cM,\cG)$.

\begin{definition}
The oriented orthonormal frame $(n,\tau)$ of $T\cM_0$ is called the
{\em adapted frame} of the oriented scalar triple $(\cM,\cG,\Phi)$.
\end{definition}

\noindent Recall that the gradient flow and level set curves of $\Phi$
are mutually orthogonal on the noncritical locus $\cM_0$. Hence these
curves determine an orthogonal mesh on the Riemannian manifold
$(\cM_0,\cG\vert_{\cM_0})$.

\begin{definition}
\label{def:mesh}
The {\em potential mesh} of the oriented scalar triple
$(\cM,\cG,\Phi)$ is the oriented mesh determined by the gradient flow
and level set curves of $\Phi$ on the noncritical set $\cM_0$, where
we orient gradient flow lines towards increasing values of $\Phi$
and level set curves such that their normalized tangent
vectors generate the vector field $\tau$.
\end{definition}

\subsection{The characteristic functions of a scalar triple}

\noindent With Definition \ref{def:mesh}, the frame $(n,\tau)$ is the
oriented orthonormal frame defined by the potential mesh. For any
point $m\in \cM_0$, let $\mu(m)$ and $\lambda(m)$ respectively be the
{\em signed} curvatures (see Definition \ref{def:xi} in Appendix
\ref{app:signed}) at $m$ of the unique gradient and level set curves
of $\Phi$ which pass through $m$. This gives smooth functions
$\mu,\lambda\in \cC^\infty(\cM_0)$.

\begin{definition}
The functions $\mu,\lambda\in \cC^\infty(\cM_0)$ are called the {\em first and second
characteristic functions} of the oriented scalar triple
$(\cM,\cG,\Phi)$.
\end{definition}

\noindent Since $N_n=\tau$ and $N_\tau=-n$, the Frenet-Serret
relations \eqref{KN}, \eqref{NK} for the gradient flow and level set
curves give the {\em characteristic equations} of the oriented scalar
triple $(\cM,\cG,\Phi)$:
\beqan
\label{FSadapted}
&&\nabla_n n=\mu \tau~~,~~\nabla_n\tau=-\mu n\nn\\
&& \nabla_\tau \tau=-\lambda n~~,~~\nabla_\tau n=\lambda \tau~~.
\eeqan
For any vector field $X\in \fX(\cM_0)$, we denote by:
\be
\Phi_X\eqdef (\dd\Phi)(X)=X(\Phi) \in \cC^\infty(\cM)
\ee
the directional derivative of $\Phi$ with respect to $X$. Then the
gradient of $\Phi$ expands as follows in any orthonormal frame $(X,Y)$
of $T\cM$:
\ben
\label{gradPhiExp}
\grad\Phi=\cG(X,\grad\Phi)X+\cG(Y,\grad\Phi)Y=(\dd\Phi)(X) X+(\dd\Phi)(Y)Y=\Phi_X X+\Phi_Y Y~~.
\een
For the adapted frame, we have:
\beqan
\label{PhinPhitau}
&&\Phi_n=n(\Phi)=(\dd\Phi)(n)=\cG(\grad\Phi,n)=||\grad\Phi||=||\dd\Phi||\nn\\
&&\Phi_\tau=\tau(\Phi)=(\dd\Phi)(\tau)=\cG(\grad\Phi,\tau)=0~~.
\eeqan
For any vector fields $X,Y$ defined on $\cM$, we set:
\ben
\label{HessXY}
\Phi_{XY}\eqdef\Hess(\Phi)(X,Y)=(\nabla_X \dd\Phi)(Y)=X(Y(\Phi))-(\dd\Phi)(\nabla_X Y)=X(Y(\Phi))-(\nabla_XY)(\Phi)~~.
\een
Notice that $\Phi_{XY}=\Phi_{YX}$ since the Levi-Civita connection
$\nabla$ of $(\cM,\cG)$ is torsion-free.

\begin{definition}
The function $\Phi_{nn}\in \cC^\infty(\cM_0)$ is called the {\em third
  characteristic function} of the scalar triple $(\cM,\cG,\Phi)$.
\end{definition}

\begin{prop}
\label{prop:HessPhiAdapted}
The following relations hold on the non-critical submanifold $\cM_0$:
\beqan
\label{HessPhiAdapted}
\Phi_{nn}=n(||\dd\Phi||)~~&,&~~\Phi_{\tau\tau}=\lambda ||\dd\Phi||\nn\\
\Phi_{n\tau}=\mu ||\dd\Phi||&=&\tau(||\dd\Phi||)=\Phi_{\tau n}
\eeqan
In particular, the following relations hold on $\cM_0$:
\ben
\label{lambdamu}
\lambda=\frac{\Phi_{\tau\tau}}{||\dd\Phi||}~~,~~\mu=\frac{\Phi_{n\tau}}{||\dd\Phi||}~~.
\een
\end{prop}

\begin{remark}
The last two relations in \eqref{HessPhiAdapted} imply:
\ben
\label{mueqPhi}
\mu=\tau(\log||\dd\Phi||)~~.
\een
\end{remark}

\begin{proof}
Equations \eqref{PhinPhitau} give:  
\ben
\label{ntauPhi}
n(\Phi)=||\dd\Phi||~~,~~\tau(\Phi)=0~~.
\een
Relations \eqref{HessPhiAdapted} follow immediately from
\eqref{FSadapted} and \eqref{HessXY} upon using \eqref{ntauPhi}. The
second and third relation in \eqref{HessPhiAdapted} give
\eqref{lambdamu}.
\end{proof}

\subsection{Adapted coordinates on the nondegenerate set}

\begin{definition}
\label{def:cU}
The {\em nondegenerate set} of $(\cM,\cG,\Phi)$ is the following open subset of $\cM_0$:
\ben
\label{cUdef}
\cU\eqdef \{m\in \cM_0~~\vert~~\mu(m)\neq 0\}=\{m\in \cM_0~~\vert~~\Phi_{n\tau}(m)\neq 0\}~~.
\een
\end{definition}

\noindent Notice that $\cU$ is empty iff all gradient lines of $\Phi$
are geodesics. In this
case, we say the scalar triple $(\cM,\cG,\Phi)$ is {\em degenerate}.
For the remainder of this paper, we assume that $(\cM,\cG,\Phi)$ is a
nondegenerate scalar triple. Consider the non-negative function:
\be
\Psi\eqdef ||\dd\Phi||\in \cC^\infty(\cM,\R_{\geq 0})~~,
\ee
whose vanishing locus coincides with $\Crit\Phi$ and let $\beta\eqdef (\Phi,\Psi):\cM\rightarrow \R_{>0}\times \R_{\geq 0}$. 

\begin{prop}
\label{prop:basecoords}
The following relations hold on the non-critical submanifold $\cM_0$:
\beqa
&&\grad \Phi=||\dd\Phi||n\\
&& \grad ||\dd\Phi|| =\Phi_{nn} n+\mu ||\dd \Phi||\tau~~.
\eeqa
\end{prop}

\begin{proof}
The first relation is obvious. The second follows from:
\be
\cG(\grad||\dd \Phi||,n)=n(||\dd\Phi||)=\Phi_{nn}~~,~~\cG(\grad||\dd\Phi||,\tau)=\tau(||\dd\Phi||)=\mu ||\dd\Phi||~~,
\ee
where we used \eqref{HessPhiAdapted}. 
on $\cU$.
\end{proof}

\begin{cor}
If the map $\beta_\cU$ is proper, then it is a covering map. 
\end{cor}

\begin{proof}
Follows immediately from \cite[Lemma 2]{Ho} using the fact that $\R^2$
and $\cM$ are locally compact space when endowed with the manifold
topology.
\end{proof}

\begin{cor}
The map $\beta$ has rank $2$ on $\cU$, rank $1$ on $\cM_0\setminus
\cU$ and is locally constant on $\Crit \Phi=\cM\setminus \cM_0$.  In
particular the restriction $\beta_\cU\eqdef (\Phi,\Psi):\cU\rightarrow
\cB$ is a surjective local diffeomorphism, where $\Psi\eqdef
||\dd\Phi||$ and we defined $\cB\eqdef \beta(\cU)$.
\end{cor}

\noindent Since $\beta_\cU$ is a local diffeomorphism, it provides
local coordinates on a small enough neighborhood (contained in $\cU$)
of each point of $\cU$.

\begin{definition}
The coordinates $(\Phi,\Psi=||\dd\Phi||)$ defined on a small enough
neighborhood of a point $m\in \cU$ of the non-degenerate set are
called {\em adapted coordinates}.
\end{definition}

\subsection{Basic observables defined by the adapted frame}
\label{sunsec:adaptedobs}

\subsubsection{The signed characteristic angle and characteristic signature.}

\noindent The adapted frame defines two natural discontinuous local
observables which are related to each other. We refer the reader to
Appendix \ref{app:signed} for the definition of the signed angle and
relative chirality of pairs of nonzero vectors on an oriented
Riemannian manifold, which are used in the definition below.

\begin{definition}
\label{def:signed}
Let $u\in \dot{T}\cM_0$ be a non-zero vector tangent to $\cM$ at a
noncritical point $m\in \cM_0$.
\begin{itemize}
\item The signed angle:
\be
\theta(u)\eqdef \uptheta_{n(m)}(u)\in (-\pi,\pi]
\ee
of $u$ relative to the unit gradient vector
$n(m)\in T_m\cM$ is called the {\em signed characteristic angle} of $u$
\item The chirality of $u$ with respect to
$n(m)$:
\be
\sigma(u)\eqdef \upsigma_{n(m)}(u)=\sign \left[\sin\theta(u)\right]=\sign \, \cG(\tau(m),u)\in \{-1,0,1\}
\ee
is called the {\em characteristic signature} of $u$.
\end{itemize}
\end{definition}

\noindent The signed characteristic angle and characteristic
signature define functions $\theta:\dot{T}\cM_0\rightarrow
(-\pi,\pi]$ and $\sigma:{\dot T}\cM_0\rightarrow \{-1,0,1\}$ which we
  call the {\em signed characteristic angle function} and {\em
    characteristic signature function} of $(\cM,\cG,\Phi)$. These
  are {\em discontinuous} basic local observables of the model. 

\subsubsection{The lift of the adapted frame.}

\noindent Recall the abstract oriented Frenet frame $(T,N)$ of
$F\vert_{\dot{T}\cM}$ defined in Appendix \ref{app:signed}, which
gives a universal geometric description of the oriented Frenet frame
of curves trough Finsler vector fields defined on the slit tangent
bundle of $\cM$.

\begin{definition}
The {\em tautological lift} of a vector field $X\in \fX(\cM)$ is
the Finsler vector field $X^\v \in \Gamma(T\cM,F)$ given by:
\be
X^\v(u)\eqdef X(\pi(u))~~\forall u\in T\cM~~.
\ee
\end{definition}

\noindent The adapted frame $(n,\tau)$ of $T\cM_0$ lifts to a frame
$(n^\v,\tau^\v)$ of the restricted tautological bundle $F_0\eqdef
F\vert_{\cM_0}=\pi^\ast (T\cM_0)$. Let $T_0\eqdef
T\vert_{\dot{T}\cM_0}$. The signed angle and chirality maps
$\uptheta:\dot{T}\cM\times_{\cM}\dot{T}\cM\rightarrow (-\pi,\pi]$ and
  $\upsigma:\dot{T}\cM\times_{\cM}\dot{T}\cM\rightarrow \{-1,0,1\}$ of
  the oriented Riemann surface $(\cM,\cG)$ (see Definition
  \ref{def:normsignedanglemap}) restrict to maps $\uptheta_0$ and
  $\upsigma_0$ defined on $\cM_0$, which are the signed angle and
  chirality maps of the oriented Riemannian manifold
  $(\cM_0,\cG\vert_{\cM_0})$. The following immediate statement
  expresses $\theta$ and $\sigma$ through $n^\v$, $T_0$ and
  $\uptheta_0,\upsigma_0$.

\begin{prop}
\label{prop:signedfunction}  
The signed characteristic angle and characteristic signature functions
of $(\cM,\cG,\Phi)$ satisfy:
\ben
\label{tnv}
\theta= \uptheta_0\circ (n^\v, T_0)~~,~~\sigma\eqdef \upsigma_0\circ (n^\v,T_0)~~,
\een
where $n^\v$ and $T_0$ are viewed as maps from $\dot{T}\cM_0$ to
$\dot{T}\cM_0$ and the map $(n^\v,T_0):\dot{T}\cM_0\rightarrow
\dot{T}\cM_0\times_{\cM_0} \dot{T}\cM_0$ is defined through:
\be
(n^\v,T_0)(u)=(n^\v(u),T(u))=\left(n(\pi(u)),\frac{u}{||u||}\right)\in {\dot T}_{\pi(u)}\cM\times \dot{T}_{\pi(u)}\cM~~\forall u\in \dot{T}\cM_0~~.
\ee
\end{prop}

\begin{remark}
For simplicity of notation, we will omit the composition sign in relations such as \eqref{tnv} from now on. 
\end{remark}

\noindent The following result follows immediately from Proposition
\ref{prop:rotation} of Appendix \ref{app:signed}.

\begin{prop}
\label{prop:movingadapted}
The following relations hold on $\dot{T}\cM_0$:
\beqan
\label{movingadaptedTcM}
T&=&~n^\v\cos\theta+\tau^\v\sin\theta\nn\\
N&=&-n^\v\sin\theta+\tau^\v\cos\theta~~.
\eeqan
\end{prop}

\subsubsection{The adapted speed components.}

\noindent Let $I\in \Gamma(T\cM,F)$ be the {\em tautological Finsler field},
which is defined though:
\be
I(u)\eqdef u~~\forall u\in T\cM~~
\ee
and let $I_0\eqdef I\vert_{T\cM_0}$. Since $u=\cN(u) T(u)$ for all
$u\in \dot{T}\cM$, we have:
\ben
\label{IT}
I\vert_{\dot{T}\cM}=\cN\vert_{\dot{T}\cM} T~~.
\een

\begin{definition}
The {\em adapted speed component functions} are the basic local observables
$v_n,v_\tau:T\cM_0\rightarrow \R$ defined through:
\ben
\label{udec0}
I_0=v_n n^\v+v_\tau \tau^\v~~.
\een
\end{definition}

\noindent Relation \eqref{udec0} amounts to:
\ben
\label{udec}
u=v_n(u)n(\pi(u))+v_\tau(u)\tau(\pi(u))~~\forall u\in T\cM_0~~
\een
and we have:
\ben
\label{vcomp}
v_n(u)=\cG(n^\v(u),u)~~,~~v_\tau(u)=\cG(\tau^\v(u),u)~~\forall u\in T\cM_0~~.
\een

\noindent Notice that $(v_n,v_\tau):T\cM_0\rightarrow \R^2$ is the
linear system of fiberwise coordinates defined by the adapted frame
$(n,\tau)$ of $T\cM_0$ in the sense of Subsection
\ref{subsec:fiberwise}.

\begin{prop}
\label{prop:vntau}
The following relations hold on $\dot{T}\cM_0$:
\ben
\label{vnvtau}
\cN=\sqrt{v_n^2+v_\tau^2}~~,~~v_n=\cN\cos\theta~~,~~v_\tau=\cN\sin\theta~~,
\een
where $\cN:T\cM\rightarrow \R_{\geq 0}$ is the norm function. The
first of these relations holds on all of $T\cM_0$.
\end{prop}

\begin{proof}
Follows immediately from \eqref{udec}, \eqref{vcomp} and the first equation in \eqref{movingadaptedTcM} using
\eqref{IT}.   
\end{proof}

\subsubsection{The tangential and positive normal derivative of the scalar potential.}

\noindent For any function $f\in \cC^\infty(\cM)$, we define:
\be
f^\v\eqdef \pi^\ast(f)=f\circ \pi \in \cC^\infty(T\cM)~~.
\ee
For any Finsler vector field $W\in \Gamma(T\cM,F)$, define
$\Phi_W\in \cC^\infty(T\cM)$ through:
\be
\Phi_W(u)\eqdef (\dd\Phi)(\pi(u))(W(u))=\cG((\grad\Phi)^\v(u),W(u))~~\forall u\in T\cM~~.
\ee
With these definitions, we have:
\be
\Phi_{X^\v}=(\Phi_X)^\v~~\forall X\in T\cM~~.
\ee

\begin{definition}
The basic local observables $\Phi_T,\Phi_N\in \cC^\infty(\dot{T}\cM)$ are
called the {\em tangential derivative} and {\em positive normal
  derivative} of $\Phi$.
\end{definition}

\begin{prop}
\label{prop:PhiTN}
The following relations hold on $\dot{T}\cM_0$:
\beqan
\label{PhiDirectionalTcM}
&& \Phi_T=\cG((\grad\Phi)^\v,T)=+||\dd\Phi||^\v\cos\theta\nn\\
&& \Phi_N=\cG((\grad\Phi)^\v,N)=-||\dd\Phi||^\v\sin\theta~~.
\eeqan
\end{prop}

\begin{proof}
Follows immediately using \eqref{movingadaptedTcM} upon noticing that $(\grad \Phi)^\v(u)=||\dd\Phi||^\v n(u)$ for all $u\in T\cM_0$. 
\end{proof}

\subsubsection{Evaluation along a curve.}

The formulation given above encodes in a universal manner certain objects and relations
which are generally presented using a given curve in $\cM_0$. This universal formulation
is made possible by the fact that jets of curves generate the bundles $J^k(\cM)$ (see Remark
\ref{rem:jets}). The usual formulation is recovered by considering the evaluation of the
basic local observables defined above a long a curve. 

Given a noncritical curve $\varphi:I\rightarrow \cM_0$ and a
regular time $t\in I_\reg$, we define:
\be
n_\varphi(t)\eqdef n(\varphi(t))\in T_{\varphi(t)}\cM_0~~,~~\tau_\varphi(t)\eqdef \tau(\varphi(t))\in T_{\varphi(t)}\cM_0~~.
\ee
Let $T_\varphi(t)$ be the unit tangent vector to $\varphi$
at time $t$ and $N_\varphi(t)$ be the positive normal vector to
$\varphi$ at time $t$, which is chosen such that the orthonormal basis
$(T_\varphi(t),N_\varphi(t))$ of $T_{\varphi(t)}\cM_0$ is positively
oriented. We have (see Appendix \ref{app:signed}):
\be
N_\varphi(t)=N(\dot{\varphi}(t))~~,~~T_\varphi(t)=T(\dot{\varphi}(t))~~.
\ee

\begin{definition}
\label{def:signed1}
Consider a noncritical curve $\varphi:I\rightarrow \cM_0$. Then:
\begin{itemize}
\item The signed characteristic angle:
\be
\theta_\varphi(t)\eqdef \theta(T_\varphi(t))=\uptheta_{n_\varphi(t)}(T_\varphi(t))\in (-\pi,\pi]~~(t\in I_\reg)
\ee
of $T_\varphi(t)$ is called the {\em signed characteristic angle} of $\varphi$ at the regular time $t$.
\item The characteristic signature:
\be
\sigma_\varphi(t)\eqdef \sigma(T_\varphi(t))=\sign \left[\sin\theta_\varphi(t)\right]=\sign \, \cG(\tau_\varphi(t),T_\varphi(t))~~(t\in I_\reg)
\ee
of $T_\varphi(t)$ is called the {\em characteristic signature} of $\varphi$ at the regular time $t$.
\item The functions $v_n,v_\tau:I\rightarrow \R$ defined through:
\be
v_n^\varphi=v_n\circ \dot{\varphi}~~,~~v_\tau^\varphi=v_\tau\circ \dot{\varphi}~~(t\in I).
\ee
are called the {\em adapted speed components of $\varphi$}.
\end{itemize}
\end{definition}

\noindent For any noncritical curve $\varphi:I\rightarrow \cM_0$, relation \eqref{udec0} gives:
\ben
\label{speedadapted}
\dot{\varphi}=v^\varphi_n n_\varphi +v^\varphi_\tau \tau_\varphi~~\forall t\in I~~
\een
and we have:
\ben
\label{vnvtauphi}
v^\varphi_n\eqdef \cG(n_\varphi,\dot{\varphi})=||\dot{\varphi}||\cos\theta_\varphi~~,~~v^\varphi_\tau\eqdef \cG(\tau_\varphi,\dot{\varphi})=||\dot{\varphi}||\sin\theta_\varphi~~
\een
Propositions \ref{prop:movingadapted}, \ref{prop:vntau} and \ref{prop:PhiTN} imply:

\begin{prop}
Consider a noncritical curve $\varphi:I\rightarrow \cM_0$. Then the
following relations hold for $t\in I_\reg$:
\beqan
\label{movingadapted}
T_\varphi(t)&=&\,n_\varphi(t)\cos\theta_\varphi(t)+\tau_\varphi(t)\sin\theta_\varphi(t)\nn\\
N_\varphi(t)&=&-n_\varphi(t) \sin\theta_\varphi(t)+\tau_\varphi(t)\cos\theta_\varphi(t)~~.
\eeqan
In particular, we have:
\ben
\label{sprods}
\cG(n_\varphi(t),T_\varphi(t))=\cos\theta_\varphi(t)~~,~~\cG(n_\varphi(t),N_\varphi(t))=-\sin\theta_\varphi(t)~~\forall t\in I_\reg~~.
\een
Moreover, we have:
\ben
\label{tantheta}
||\dot{\varphi}||=\sqrt{(v^\varphi_n)^2+(v^\varphi_\tau)^2}~~\forall t\in I
\een
and:
\be
(\dd \Phi)(T_{\varphi}(t))=\Phi_T(\dot{\varphi}(t))~~,~~(\dd \Phi)(N_{\varphi}(t))=\Phi_N(\dot{\varphi}(t))~~\forall t\in I_\reg~~.
\ee
\end{prop}

\subsection{Adapted and modified phase space coordinates}

\noindent Consider the following functions defined on $T\cM$:
\be
\varPhi\eqdef \Phi^\v~~,~~\varPsi\eqdef \Psi^v=||\dd\Phi||^\v~~.
\ee
Notice that $\varPhi$ is smooth on $T\cM$ while $\Psi$ is smooth on $T\cM_0$.
Recall that $\cH$ denotes the rescaled Hubble function of $\cM$. 

\begin{lemma}
\label{prop:auxdcoords}  
The functions $(v_n,v_\tau,\varPhi,\varPsi)$ and $(v_n,v_\tau,
\cH,\varPsi)$ form systems of local coordinates on a small enough
neighborhood of any point of the open set $T\cU\subset T\cM_0$, where $\cU$
is the nondegenerate set of $(\cM,\cG,\Phi)$ (see Definition
\ref{def:cU}).
\end{lemma}

\begin{proof}
  Follows immediately from Proposition \ref{prop:basecoords} upon using
  the fact that $v_n,v_\tau$ give (linear) fiber coordinates on $T\cM_0$.
\end{proof}

\begin{definition}
The local coordinates $(v_n,v_\tau,\varPhi,\varPsi)$ defined on a
small enough neighborhood of any point of $T\cU$ are called {\em
  adapted phase space coordinates}.
\end{definition}

\noindent 
If $u\in T\cU$ and $U\subset \cU$ is a neighborhood of $\pi(u)$ on
which $\Phi$ and $\Psi$ are coordinates, then
$(v_n,v_\tau,\varPhi,\varPsi)$ are coordinates on the neighborhood
$TU$ of $u$.  Adapted phase space coordinates give a special
coordinate system defined on $TU$ whose associated fiberwise
coordinate system $(v_n,v_\tau)$ is linear; the latter is defined by
the frame $(n,\tau)$ of $T\cM_0$ (which generally is {\em not}
holonomic).

\begin{prop}
The functions $(v_n,v_\tau,\cH,\varPhi)$ give local coordinates on a
small enough vicinity of any point of $T\cU$.
\end{prop}

\begin{proof}
Since $\Phi$ is strictly positive, the relation:
\ben
\label{cHcoord}
\cH=\frac{1}{M_0}(v_n^2+v_\tau^2+2\varPhi)^{1/2}\Longleftrightarrow \varPhi=\frac{1}{2}(M_0^2\cH^2-v_n^2-v_\tau^2)
\een
gives a smooth change of coordinates
$(v_n,v_\tau,\varPhi,\varPsi)\leftrightarrow (v_n,v_\tau,\cH,\varPsi)$
on a small enough vicinity of every point of $T\cU$. Hence
$(v_n,v_\tau,\cH,\varPsi)$ is a local coordinate system on a small
enough vicinity of any point of $T\cU$.
\end{proof}

\begin{definition}
The local coordinates $(v_n,v_\tau,\cH,\varPsi)$ defined on a small
enough neighborhood of any point of $T\cU$ are called {\em modified
  phase space coordinates}.
\end{definition}

\begin{remark}
Relation \eqref{cHcoord} gives:
\ben
\label{kappav}
\kappa=\frac{v_n^2+v_\tau^2}{M_0^2\cH^2-(v_n^2+v_\tau^2)}=\frac{1}{\frac{M_0^2\cH^2}{v_n^2+v_\tau^2}-1}~~.
\een
\end{remark}

\subsection{The cosmological semispray in adapted phase space coordinates}

\noindent In this subsection, we determine the form of the cosmological
equation (and hence of the cosmological semispray) in adapted
coordinates.  We will first write the cosmological equation in
modified phase space coordinates and then change to adapted phase
space coordinates.  Let $\varphi:I\rightarrow \cM_0$ be a noncritical
cosmological curve. Recall that:
\be
\Phi_\varphi\eqdef \Phi\circ \varphi~~,~~\Psi_\varphi\eqdef \Psi\circ \varphi~~,~~\cH_\varphi\eqdef \cH\circ \varphi=\frac{1}{M_0}[(v^\varphi_n)^2+(v^\varphi_\tau)^2+2\Phi_\varphi]^{1/2}
\ee
and:
\ben
v^\varphi_n\eqdef v_n\circ \dot\varphi~~,~~v^\varphi_\tau\eqdef v_\tau\circ \dot{\varphi}~~,~~n_\varphi\eqdef n\circ \varphi~~,~~\tau_\varphi\eqdef \tau\circ \varphi~~.
\een
and define:
\be
\lambda_\varphi\eqdef \lambda\circ \varphi~~,~~\mu_\varphi\eqdef \mu\circ \varphi~~.
\ee
We remind the reader that $\Psi\eqdef ||\dd\Phi||$. When $\varphi$ is both noncritical and nondegenerate, we have
(see \eqref{movingadapted}):
\beqa
T_\varphi&=&~n_\varphi\cos\theta_\varphi+\tau_\varphi\sin\theta_\varphi\nn\\
N_\varphi&=&-n_\varphi \sin\theta_\varphi+\tau_\varphi\cos\theta_\varphi~~,
\eeqa
where $\theta_\varphi:I\rightarrow (-\pi, \pi]$ is the signed
  characteristic angle of $\varphi$.

\begin{lemma}
\label{lemma:cosmadapted}
The cosmological equation for a curve $\varphi:I\rightarrow \cU$ whose image is contained in $\cU$ is
equivalent with:
\beqan
\label{bsys}
&& \dot{v}^\varphi_n-\lambda_\varphi (v^\varphi_\tau)^2-\mu_\varphi v^\varphi_n v^\varphi_\tau+\cH_\varphi v^\varphi_n+\Psi_\varphi=0~~\nn\\
&& \dot{v}^\varphi_\tau+\mu_\varphi (v^\varphi_n)^2+\lambda_\varphi v^\varphi_nv^\varphi_\tau+\cH_\varphi v^\varphi_\tau=0~~,
\eeqan
together with the following two equations:
\beqan
\label{asys0}
&& \frac{\dd \cH_\varphi}{\dd t}=-\frac{1}{M_0^2}[(v_n^\varphi)^2+(v_\tau^\varphi)^2]\nn\\
&& \frac{\dd \Psi_\varphi}{\dd t}=\frac{1}{2}(\Phi_{nn\varphi} v_n^\varphi+\mu_\varphi \Psi_\varphi v^\varphi_\tau)~~,
\eeqan
where $\Phi_{nn\varphi}\eqdef \Phi_{nn}\circ \varphi$.
\end{lemma}

\begin{proof}
The speed
  $\dot{\varphi}:I\rightarrow T\cU$ expands in the adapted frame as:
\ben
\dot{\varphi}=v^\varphi_n n_\varphi +v^\varphi_\tau \tau_\varphi~~.
\een
We have:
\be
\nabla_t  n_\varphi   =(\mu_\varphi v^\varphi_n+\lambda_\varphi v^\varphi_\tau) \tau_\varphi\nn~~,~~\nabla_t  \tau_\varphi =-(\mu_\varphi v^\varphi_n+\lambda_\varphi v^\varphi_\tau)n_\varphi~~,
\ee
where we used \eqref{FSadapted}. Thus:
\be
\nabla_t \dot{\varphi}=[\dot{v}^\varphi_n-\lambda_\varphi (v^\varphi_\tau)^2-\mu_\varphi v^\varphi_n v^\varphi_\tau] n_\varphi+
[\dot{v}^\varphi_\tau+\mu_\varphi (v^\varphi_n)^2+\lambda_\varphi v^\varphi_n v^\varphi_\tau]\tau_\varphi~~.
\ee
Hence the cosmological equation \eqref{eomsingle} implies \eqref{bsys}.
The first equation in \eqref{asys0} is the dissipation
equation \eqref{DissipEq}. To derive the second equation in \eqref{asys0}, we compute:
\beqa
&&\frac{\dd ||(\dd\Phi)(\varphi(t))||}{\dd t}=\frac{\cG(\nabla_t[(\grad\Phi)(\varphi(t))],(\grad\Phi)(\varphi(t)))}{2||(\dd\Phi)(\varphi(t))||}\nn\\
&=&\frac{\Hess(\Phi)(\varphi(t))(\dot{\varphi}(t),(\grad\Phi)(\varphi(t)))}{2||(\dd\Phi)(\varphi(t))||}=\frac{1}{2}\Hess(\Phi)(\varphi(t))(\dot{\varphi}(t),n_\varphi(t))\nn\\
&=&\frac{1}{2}\left[\Phi_{nn}(\varphi(t)) v_n^\varphi(t)+\Phi_{n\tau}(\varphi(t)) v_\tau^\varphi(t)\right]~~,
\eeqa
where we used \eqref{speedadapted}. The second equation in
\eqref{asys0} now follows from the relation above upon using
\eqref{HessPhiAdapted}. Hence the cosmological equation for $\varphi$
implies \eqref{bsys} and \eqref{asys0}. The system formed by these
four first order ODEs is equivalent with the cosmological equation on
$T\cU$ since $(v_n,v_\tau,\cH,\Psi)$ is a local coordinate system around every
point of $T\cU$.
\end{proof}

\noindent The previous Lemma can be reformulated as follows.

\begin{prop}
The cosmological equation for curves whose image is contained in $\cU$
is equivalent with the following system written in local coordinates
$(v_n,v_\tau,\cH,\varPsi)$ in a small enough vicinity of any point $u\in T\cU$:
\beqan
\label{adsys}
&& \dot{v}_n-\lambda^\v v_\tau^2-\mu^\v v_n v_\tau+\cH v_n+\varPsi=0~~\nn\\
&& \dot{v}_\tau+\mu^\v v_n^2+\lambda^\v v_nv_\tau+\cH v_\tau=0~~\nn\\
&& \dot{\cH}=-\frac{1}{M_0^2}(v_n^2+v_\tau^2)\\
&& {\dot \varPsi}=\frac{1}{2}(\Phi^\v_{nn} v_n+\mu^\v \varPsi v_\tau)~~,\nn
\eeqan
where $\Phi_{nn},\lambda$ and $\mu$ are expressed locally as functions of
$v_n,v_\tau,\cH$ and $\varPsi$ using the fact that $(\Phi,\Psi)$
are local coordinates on a vicinity of $\pi(u)\in \cU$ and relations \eqref{cHcoord}.
\end{prop}

\noindent More precisely, the quantities $\Phi_{nn},\lambda$ and $\mu$
can be expressed locally around $\pi(u)$ as functions of $\Phi$ and
$\Psi$ and hence locally around $u$ as functions of $\varPhi$ and
$\varPsi$.  Using relations \eqref{cHcoord}, their lifts to $\cU$
locally become functions of $\varPhi,\cH$ and $v_n,v_\tau$.

\begin{prop}
\label{prop:cosmadapted}
The cosmological equation for curves whose image is contained in $\cU$
is equivalent with the following system written in adapted phase space
coordinates $(v_n,v_\tau,\varPhi,\varPsi)$ on a small enough vicinity
of any point $u\in T\cU$:
\beqan
\label{asys}
&& \dot{v}_n=\lambda^\v v_\tau^2+\mu^\v v_n v_\tau-\frac{1}{M_0}(v_n^2+v_\tau^2+2\varPhi)^{1/2} v_n-\varPsi~~\nn\\
&& \dot{v}_\tau=-\mu^\v v_n^2-\lambda^\v v_nv_\tau-\frac{1}{M_0}(v_n^2+v_\tau^2+2\varPhi)^{1/2}v_\tau~~\nn\\
&& \dot{\varPhi}=\varPsi v_n~~\\
&& {\dot \varPsi}=\frac{1}{2}(\Phi_{nn}^\v v_n+\mu^\v \varPsi v_\tau)~~,\nn
\eeqan
where $\Phi_{nn}$, $\lambda$ and $\mu$ are expressed locally around
$\pi(u)$ in terms of $\Phi$ and $\Psi$ (and hence $\Phi_{nn}^\v$,
$\lambda^\v$ and $\mu^\v$ are expressed locally around $u$ in terms of
$\varPhi$ and $\varPsi$).
\end{prop}

\begin{proof}
The first two equations in \eqref{adsys} imply:
\ben
\label{vdotv}
v_n\dot{v}_n+v_\tau\dot{v}_\tau=-\cH (v_n^2+v_\tau^2)-\varPsi v_n~.
\een
Relation \eqref{cHcoord} implies:
\be
\dot{\varPhi}=M_0^2\cH \dot{\cH}-(v_n\dot{v}_n+v_\tau\dot{v}_\tau)=-\cH(v_n^2+v_\tau^2)-(v_n\dot{v}_n+v_\tau\dot{v}_\tau)=\varPsi v_n~~,
\ee
where we used Proposition \ref{prop:dissip} and \eqref{vdotv}. This shows
that the third equations of \eqref{adsys} and \eqref{asys} are
equivalent modulo the first two equations. 
\end{proof}

\begin{cor}
\label{cor:Sadapted}
The cosmological semispray has the following form in adapted phase
space coordinates around every point $u\in T\cU$:
\be
S=S^{v_n}\frac{\pd}{\pd v^n}+S^{v_\tau}\frac{\pd}{\pd v^\tau}+S^{\varPhi}\frac{\pd}{\pd \varPhi}+S^\varPsi\frac{\pd}{\pd\varPsi}~~,
\ee
where:
\beqan
\label{Sad}
&& S^{v_n}=\lambda^\v v_\tau^2+\mu^\v v_n v_\tau-\frac{1}{M_0}(v_n^2+v_\tau^2+2\varPhi)^{1/2} v_n-\varPsi\nn\\
&& S^{v_\tau}=-\mu^\v v_n^2-\lambda^\v v_nv_\tau-\frac{1}{M_0}(v_n^2+v_\tau^2+2\varPhi)^{1/2}v_\tau\nn\\
&& S^\varPhi=\varPsi v_n\\
&& S^\varPsi=\frac{1}{2}(\Phi_{nn}^\v v_n+\mu^\v \varPsi v_\tau)~~,\nn
\eeqan
where $\Phi_{nn}$, $\lambda$ and $\mu$ are expressed locally around
$\pi(u)$ in terms of $\Phi$ and $\Psi$ (and hence $\Phi_{nn}^\v$,
$\lambda^\v$ and $\mu^\v$ are expressed locally around $u$ in terms of
$\varPhi$ and $\varPsi$).
\end{cor}

\noindent These results show that the functions $\lambda(\Phi,\Psi)$,
$\mu(\Phi,\Psi)$ and $\Phi_{nn}(\Phi,\Psi)$ (which can be defined on a
small enough vicinity of $\pi(u)$) characterize the cosmological
semispray on a small enough vicinity of any point of $u\in T\cU$. Hence cosmological
dynamics on $T\cU$ is locally determined by these three functions. This
opens the possibility of studying the local cosmological dynamics of two-field
models using the ``universal form'' \eqref{asys} of their first order equations.
In the next section, we rewrite \eqref{asys} in certain special local coordinate
systems which are induced by the so-called ``fundamental basic observables''
of the model, which have more direct physical interpretation.

\section{Fundamental basic observables}
\label{sec:fobs}

\subsection{The fundamental parameters of a smooth curve in $\cM$}

\noindent Let $\varphi:I\rightarrow\cM$ be a smooth curve and
$T_\varphi(t)$, $N_\varphi(t)$ be the unit tangent and positive unit
normal vectors to $\varphi$ at a regular time $t\in I_\reg$ (see
Appendix \ref{app:signed}).

\subsubsection{Fundamental first order parameters.}
We start by considering first order parameters, whose value at $t\in
I$ depends only only the tangent vector $\dot{\varphi}(t)$.

\begin{definition}
\label{def:Fundamental1Parameters}
The {\em fundamental first order parameters} of a curve
$\varphi:I\rightarrow \cM$ are defined as follows:
\begin{itemize}
\item The {\em first IR parameter} of $\varphi$ at time $t\in I$ is
  the non-negative real number:
\ben
\label{kappapardef}  
\kappa_\varphi(t)\eqdef \frac{||\dot{\varphi}(t)||^2}{2\Phi(\varphi(t))}\in \R_{\geq 0}~~.
\een
\item The {\em rescaled conservative parameter} of $\varphi$ at non-critical time $t\in
  I_\noncrit$ is the non-negative real number:
\ben
\label{hcpardef}
\hc_\varphi(t)\eqdef \frac{\cH_\varphi(t) ||\dot{\varphi}(t)||}{||(\dd\Phi)(\varphi(t))||}\in \R_{\geq 0}~~.
\een
\item The {\em signed characteristic angle} of
  $\varphi$ at a regular and noncritical time $t\in I_\reg\cap
  I_\noncrit $ is the characteristic angle of $T_{\varphi(t)}$, i.e. the signed angle of the tangent vector
  $T_\varphi(t)$ relative to the unit gradient vector $n_\varphi(t)$
  (see Definition \ref{def:signed}):
\be
\theta_\varphi(t)\eqdef \theta(T_\varphi(t))=\uptheta_{n_\varphi(t)}(T_\varphi(t))\in (-\pi,\pi]~~.
\ee
\end{itemize}
\end{definition}

\noindent The first IR parameter equals the ratio between the kinetic
and potential energy of $\varphi$ and plays a crucial role in
reference \cite{ren}. As we recall in Appendix \ref{app:param}, this
parameter carries the same information as the ordinary first slow roll
parameter used in the physics literature. The conservative parameter
is the ratio of the norms of the friction and gradient terms of
the cosmological equation and controls what we call the
quasi-conservative and strongly dissipative approximations (see
Section \ref{sec:cons}).

\subsubsection{Fundamental second order parameters.}
We next consider parameters which also depend on the covariant
acceleration of $\varphi$. For this, we start by introducing a vector
which compares the acceleration term of the cosmological equation to
the norm of the friction term at a regular time $t\in I_\reg$. The
fundamental second order parameters will be defined as the components
of this vector along $T_\varphi(t)$ and $N_\varphi(t)$.

\begin{definition}
\label{def:hveta}
The {\em rescaled opposite relative acceleration} of a curve
$\varphi:I\rightarrow \cM$ at the regular time $t\in I_\reg$ is the vector:
\ben
\label{hvetapardef}
\hveta_\varphi(t)\eqdef -\frac{\nabla_t \dot{\varphi}(t)}{\cH_\varphi(t)
  ||\dot{\varphi}(t)||}\in T_{\varphi(t)}\cM~~.
\een
\end{definition}

\noindent Let
\be
\hOmega_\varphi(t)\eqdef -\frac{\cG(N_\varphi(t),\nabla_t\dot{\varphi}(t))}{||\dot{\varphi}(t)||}=-\cG(N_\varphi(t),\nabla_t T_\varphi(t))~~\forall t\in I_\reg
\ee
be the {\em signed turning rate} of $\varphi$ at the regular time $t$ (see Appendix \ref{app:signedcurv}).

\begin{remark}
\label{rem:secondorder}
Consider the following open submanifold of $E^2=J^2(\cM)$:
\be
J^2_\reg(\cM)\eqdef \{\phi\in J^2(\cM)~\vert~\pi_{21}(\phi)\neq 0\}~~,
\ee
where $\pi_{21}:J^2(\cM)\rightarrow J^1(\cM)$ is the natural
projection. Since the curve $\varphi$ is arbitrary, we have:
\be
\hveta_\varphi(t)=\underline{\hveta}(j^2(\varphi)(t))~~\forall t\in I_\reg~~,
\ee
where $\underline{\hveta}\in \Gamma(J^2_\reg(\cM),F^2)$ is a
$T\cM$-valued second order local observable in the sense of Subsection
\ref{subsec:obs}. When $\varphi$ is a regular curve, we have
$j^2(\varphi)(I)\subset J^2_\reg(\cM)$ and $\hveta_\varphi$ is the
evaluation of $\underline{\hveta}$ along $\varphi$. In this case,
$\hOmega_\varphi$ is the evaluation along $\varphi$ of a second order
scalar observable $\underline{\hOmega}\in \cC^\infty(J^2_\reg(\cM))$.
\end{remark}

\begin{prop}
\label{prop:hvetapar}
For any curve $\varphi:I\rightarrow \cM$ and any $t\in I_\reg$, we have:
\ben
\label{hvetapar}
\hveta_\varphi(t)=-\frac{1}{\cH_\varphi(t)}\frac{\dd\log ||\dot{\varphi}(t)||}{\dd t} T_\varphi(t)+\frac{\hOmega_\varphi(t)}{\cH_\varphi(t)} N_\varphi(t)\in T_{\varphi(t)}\cM~~.
\een
\end{prop}

\begin{proof}
Follows by projecting the left hand side on the vectors $T_\varphi(t)$
and $N_\varphi(t)$.
\end{proof}

\begin{definition}
\label{def:Fundamental2Parameters}
The {\em fundamental second order parameters} of a curve
$\varphi:I\rightarrow \cM$ are defined as follows:
\begin{itemize}
\item The {\em rescaled second slow roll parameter} of $\varphi$ at
  a regular time $t\in I_\reg$ is the real number:
\ben
\label{hetapnormspeed}
\hetap_\varphi(t)\eqdef \cG(T_\varphi(t),\hveta_\varphi(t))=-\frac{\dd\log ||\dot{\varphi}(t)||}{\cH_\varphi(t)\dd t } \in \R~~.
\een
\item The {\em rescaled signed relative turning rate} of $\varphi$ at a regular time $t\in
I_\reg$ is the real number:
\ben
\label{homegapardef}
\homega_\varphi(t)\eqdef \cG(N_\varphi(t),\hveta_\varphi(t))=\frac{\hOmega_\varphi(t)}{\cH_\varphi(t)}\in \R~~.
\een
\end{itemize}
The {\em fundamental second order local observables} are the corresponding
functions $\underline{\hetap},\underline{\homega}\in
\cC^\infty(J^2_\reg(\cM))$ (see Remark \ref{rem:secondorder}).
\end{definition}

\noindent We have:
\be
\hveta_\varphi(t)=\hetap_\varphi(t) T_\varphi(t)+\homega_\varphi(t) N_\varphi(t)~~.
\ee
Relation \eqref{OmegaDef} of Appendix \ref{app:signed} gives:
\be
\homega_\varphi(t)=-\frac{||\dot{\varphi}(t)||}{\cH_\varphi(t)}\xi_\varphi(t)~~\forall t\in I_\reg~~,
\ee
where $\xi_\varphi$ is the {\em signed} curvature of $\varphi$ defined
in Appendix \ref{app:signedcurv}.

Let us extract the on-shell reduction of $\underline{\hetap}$ and $\underline{\homega}$ using the
cosmological equation for $\varphi$. Projecting the cosmological equation \eqref{eomsingle} for a
curve $\varphi:I\rightarrow \cM$ on the vectors $T_\varphi(t)$ and
$N_\varphi(t)$ at a regular time $t\in I_\reg$ gives respectively the {\em
  adiabatic equation}:
\be
\cG(T_\varphi(t), \nabla_t \dot{\varphi}(t))+\cH(\dot{\varphi}(t)) ||\dot{\varphi}(t)||+ \Phi_T(\dot{\varphi}(t))=0~~\forall t\in I_\reg
\ee
and the {\em entropic equation}:
\be
\cG(N_\varphi(t),\nabla_t \dot{\varphi}(t))+\Phi_N(\dot{\varphi}(t))=0~~\forall t\in I_\reg~~.
\ee
These equations can be written as:
\ben
\label{adent}
\hetap(t)=1+\frac{\Phi_T(\dot{\varphi}(t))}{\cH(\dot{\varphi}(t))||\dot{\varphi}(t)||}~~,~~\homega_\varphi(t)=\frac{\Phi_N(\dot{\varphi}(t))}{\cH(\dot{\varphi}(t))||\dot{\varphi}(t)||}~~~~
\forall t\in I_\reg~~.
\een
These relations show that the dynamical reductions of $\underline{\hetap}$ and $\underline{\homega}$
are given by:
\ben
\label{homegahetapred}
\underline{\hetap}^\red=1+\frac{\Phi_T}{\cH\cN}~~,~~\underline{\homega}^\red=\frac{\Phi_N}{\cH\cN}~~.
\een
Together with Definition \ref{def:Fundamental1Parameters}, this leads to the following.

\begin{definition}
\label{def:FundamentalObservables}
The {\em fundamental basic observables} of the oriented scalar triple
$(\cM,\cG,\Phi)$ are the basic local scalar observables defined as follows:
\begin{itemize} 
\item The {\em first IR function} of $(\cM,\cG,\Phi)$ is
  the map $\kappa:T\cM\rightarrow \R_{\geq 0}$ defined through:
\ben
\label{kappadef}
\kappa\eqdef \frac{\cN^2}{2\varPhi}~~.
\een
\item The {\em rescaled conservative function} of $(\cM,\cG,\Phi)$ is the
map $\hc:T\cM_0\rightarrow \R_{\geq 0}$ defined through:
\ben
\label{cdef}
\hc\eqdef \frac{\cN \cH}{\varPsi}=\frac{\cN\sqrt{\cN^2+2\varPhi}}{M_0\varPsi}~~.
\een
\item The {\em signed characteristic angle function} of
  $(\cM,\cG,\Phi)$ is the map $\theta:\dot{T}\cM_0 \rightarrow (-\pi,\pi]$
  defined through (see Definition \ref{def:signed} and Proposition \ref{prop:signedfunction}):
\ben
\label{thetadef}  
\theta\eqdef \uptheta_0(n^\v,T_0)~~.
\een
\item The {\em rescaled second slow roll function} of $(\cM,\cG,\Phi)$ is
  the map $\hetap:\dot{T}\cM\rightarrow \R$ defined through:
\ben
\label{etadef}  
\hetap\eqdef 1+\frac{\Phi_T}{\cH \cN}~~.
\een
\item The {\em rescaled turn rate function} is the map $\homega:{\dot
  T}\cM\rightarrow \R$ defined through:
\ben
\label{homegadef}
\homega=\frac{\Phi_N}{\cH \cN}~~.
\een
\end{itemize}
\end{definition}

\begin{remark}
Since $\Phi$ is strictly positive, the function $\hc$
vanishes only on the zero section of $T\cM_0$ and hence is strictly
positive on $\dot{T}\cM_0$. This function tends to $+\infty$ on the
topological frontier of $T\cM_0$ in $T\cM$, which coincides with the
$\pi$-preimage of the critical locus:
\be
\mathrm{Fr}(T\cM_0)=\overline{T\cM_0}\setminus T\cM_0=\pi^{-1}(\Crit\Phi)=\{u\in T\cM~\vert~\pi(u)\in \Crit\Phi\}~~.
\ee Also notice that the limits of $\hetap$ and $\homega$ at a point
which lies on the zero section of $T\cM_0$ depend on the direction
from which one approaches that point inside $T\cM_0$ and the
directional limits equal plus or minus infinity when the directional
derivatives of $\Phi$ at that point with in that direction
(respectively in its positive normal direction) are nonzero.
\end{remark}

\begin{prop}
\label{prop:etaomegaobs}
Suppose that $\varphi:I\rightarrow \cM$ is a cosmological curve. Then we have:
\ben
\label{kappacosm}
\kappa_\varphi(t)=\kappa(\dot{\varphi}(t))~~\forall t\in I
\een
and:
\ben
\label{ccosm}
\hc_\varphi(t)=\hc(\dot{\varphi}(t))~~\forall t\in I_\noncrit~~
\een
as well as:
\beqan
\label{etacosm}
&& \hetap_\varphi(t)=\heta(\dot{\varphi}(t))~~,~~\homega_\varphi(t)=\homega(\dot{\varphi}(t))~~\forall t\in I_\reg.
\eeqan
\end{prop}

\begin{proof}
Relations \eqref{kappacosm} and \eqref{ccosm} are obvious, while
\eqref{etacosm} follows from \eqref{homegahetapred} and the results of
Subsection \ref{subsec:obs}.
\end{proof}

\noindent The basic fundamental observables introduced above are not
independent. Below, we discuss some obvious relations between
them. More interesting relations will arise in the next section.

\subsection{Relations between $\cN$, $\cH$, $\hc$ and $\kappa$}

\noindent The basic observables $\kappa$ and $\hc$ are norm reducible
in the sense of Definition \ref{def:normred}. We can express $\cN$ in
terms of $\kappa$ as:
\ben
\label{Nkappa}
\cN=\sqrt{2\varPhi \kappa}~~,
\een
which gives: 
\ben
\label{kapparels}
\cH=\frac{\sqrt{2\varPhi}}{M_0}(1+\kappa)^{1/2}~~.
\een
To express $\hc$ in terms of $\kappa$, we define:

\begin{definition}
\label{def:Xi}
The {\em characteristic one-form} of the model $(M_0,\cM,\cG,\Phi)$ is
the following one-form defined on $\cM$:
\ben
\label{Xidef}
\Xi\eqdef \frac{M_0\dd\Phi}{2\Phi}\in \Omega^1(\cM)~~.
\een
\end{definition}

\begin{prop}
\label{prop:ckappa}
The following relations hold on $T\cM_0$:
\ben
\label{NH}
\cN\cH=\varPsi \hc
\een
\ben
\label{ckappa}
\hc=\frac{\left[\kappa(1+\kappa)\right]^{1/2}}{||\Xi||^\v}~~
\een
\ben
\label{kappac}
\kappa=\frac{1}{2}\left[-1+\sqrt{1+4(||\Xi||^\v)^2 \hc^2}\right]~~.
\een

\end{prop}

\begin{proof}
The first relation follows from the definition \eqref{cdef} of $\hc$,
while the second relation follows from the first using \eqref{Nkappa},
\eqref{kapparels} and the definition \eqref{Xidef} of $\Xi$.  The last
relation follows by solving \eqref{ckappa} for $\kappa$.
\end{proof}

\section{Natural fiberwise coordinates}
\label{sec:fibcoords}

The natural basic observables introduced in the previous section allow
us to define fiberwise coordinate systems on certain open subsets of
$T\cM_0$.

\subsubsection{Fiberwise norm-angle coordinates.}
Let $L_n^\pm\rightarrow \cM_0$ be the closed half line
bundles generated by the vector fields $\pm n$:
\beqan
\label{Ldef}
&& L^+\eqdef \{u\in T\cM_0~\vert~\theta(u)=0\}=\sqcup_{m\in \cM_0} \R_{\geq 0} n(m)~~\\
&&L^-\eqdef \{u\in T\cM_0~\vert~\theta(u)=\pi\}=\sqcup_{m\in \cM_0} \R_{\geq 0} (- n(m))~~
\eeqan
and notice that $(\cN,\theta)$ is a fiberwise coordinate system on
$T\cM_0\setminus L^-$ in the sense of Definition
\ref{def:fiberwise}.

\begin{definition}
The pair $(\cN,\theta)$ is called the system of {\em fiberwise norm-angle
  coordinates} on $T\cM_0\setminus L^- \subset \dot{T}\cM_0$.
\end{definition}

\noindent These fiberwise coordinates define a polar coordinate system
in each fiber of $\dot{T}\cM_0$.

\subsubsection{Fiberwise conservative coordinates.}

\begin{lemma}
\label{lemma:NHc}
The following relations hold on $T\cM_0$:
\ben
\label{Nc}
\cN=\sqrt{-\varPhi+\sqrt{\varPhi^2+M_0^2 \varPsi^2 \hc^2}}~~.
\een
\ben
\label{Hc}
\cH=\frac{1}{M_0}\left[\varPhi+\sqrt{\varPhi^2+M_0^2 \varPsi^2 \hc^2}\right]~~.
\een
\end{lemma}

\begin{proof}
Follows by combining \eqref{kappac} with with \eqref{Nkappa} and \eqref{kapparels}.
\end{proof}

\begin{prop}
\label{prop:conscoords}
The pair of functions $(\hc,\theta)$ is a fiberwise coordinate system on
$T\cM_0\setminus L^-$, which we call {\em fiberwise conservative
  coordinates}.
\end{prop}

\begin{proof}
The values of $\theta(u)$ and $\hc(u)$, determine the direction and
norm of the non-zero tangent vector $u\in \dot{T}\cM_0$ in each fiber
of $\dot{T}\cM_0$, where the second statement follows from equation
\eqref{Nc}.
\end{proof}

\noindent Notice that $\hc$ can be expressed in terms of $\cN$ and $\varPhi$, $\varPsi$ as:
\be
\hc=\frac{1}{M_0\varPsi}\sqrt{(\cN^2+\varPhi)^2-\varPhi^2}=\frac{1}{M_0\Psi}\sqrt{\cN^2(\cN^2+2\varPhi)}~~.
\ee

\begin{cor}
\label{cor:q}
The functions: 
\ben
\label{qdef}
q_1\eqdef \hc \cos\theta~~,~~q_2\eqdef \hc \sin\theta~~.
\een
are fiberwise coordinates on $\dot{T}\cU$, which we call {\em
  fiberwise Cartesian conservative coordinates}.
\end{cor}

\subsubsection{Fiberwise roll-turn coordinates.}

\begin{lemma}
\label{lemma:rtcons}
The following relations hold on $\dot{T}\cM_0$:
\ben
\label{etaomegactheta}
\hetap=1+\frac{\cos\theta}{\hc}~~,~~\homega=-\frac{\sin\theta}{\hc}~~
\een
\ben
\label{pos}
(1-\hetap)^2+\homega^2=\frac{1}{\hc^2}~~
\een
\ben
\label{comegaeta}
\hc=\frac{1}{\sqrt{(1-\hetap)^2+\homega^2}}~~,~~\sin\theta=-\frac{\homega}{\sqrt{(1-\hetap)^2+\homega^2}}~~,~~\cos\theta=-\frac{1-\hetap}{\sqrt{(1-\hetap)^2+\homega^2}}~~.
\een
Moreover, we we have $\sign(\homega)=-\sigma$ and:
\ben
\label{NHoe}
\cN\cH=\frac{\varPsi}{\sqrt{(1-\hetap)^2+\homega^2}}~~.
\een
as well as:
\ben
\label{omegaetatheta}
\homega\cos\theta=(1-\hetap)\sin\theta~~.
\een
\end{lemma}

\begin{remark}
Notice that $\hetap(u)=1$ iff $|\theta(u)|=\frac{\pi}{2}$ i.e. iff the
non-zero vector $u\in \dot{T}\cM_0$ is orthogonal to the unit gradient
vector $n(\pi(u))$. In this case, we have $\homega(u)=\pm
\frac{1}{\hc(u)}\neq 0$.
\end{remark}

\begin{proof}
Relations \eqref{etaomegactheta} follow from Proposition \ref{prop:PhiTN}
and the definitions of $\hetap$ and $\homega$.  These relations
give:
\ben
\label{thetac}
\cos\theta=-(1-\hetap)\hc~~,~~\sin\theta=-\homega \hc~~. 
\een
Using these equations, the identity $1=\cos^2\theta+\sin^2\theta$
implies \eqref{pos} and \eqref{comegaeta}. The second relation in
\eqref{etaomegactheta} gives $\sign(\homega)=-\sigma$, where we remind
the reader that $\sigma=\sign(\sin\theta)$. Relation \eqref{NH} and
the first relation in \eqref{comegaeta} imply \eqref{NHoe}. The last
two relations in \eqref{comegaeta} imply \eqref{omegaetatheta}.
\end{proof}

\begin{prop}
\label{prop:rtcoords}
The pair of functions $(\hetap,\homega)$ is a fiberwise coordinate
system on $\dot{T}\cM_0$ in the sense of Definition
\ref{def:fiberwise}, which we call {\em fiberwise roll-turn
  coordinates}.
\end{prop}

\begin{proof}
The values of $\hetap(u)$ and $\homega(u)$ determine and are
determined by those of $\theta(u)$ and $\hc(u)$ by Lemma
\ref{lemma:rtcons}, so the conclusion follows from Proposition
\ref{prop:conscoords}.
\end{proof}

\begin{cor}
\label{cor:zeta}
The functions $\zeta\eqdef 1-\hetap$ and $\homega$ give fiberwise coordinates
defined on $\dot{T}\cM_0$, which we call {\em modified roll-turn coordinates}.  
\end{cor}

\noindent The relation between fiberwise roll-turn and conservative
coordinates can be described geometrically as follows.
Recall the Cartesian fiberwise conservative coordinates
of Corollary \ref{cor:q}.

\begin{prop}
\label{prop:inversion}
In any fiber of $\dot{T}\cM_0$, the point with Cartesian conservative
coordinates $(q_1,q_2)$ is the opposite of the inversion of the point
with modified roll-turn coordinates $(\zeta,\homega)$ with respect to the
origin.
\end{prop}

\begin{proof}
The quantities $\hc=\sqrt{q_1^2+q_2^2}$ and $\theta$ are the radial
distance from the origin and polar angle in the $(q_1,q_2)$ plane
while $\sqrt{\zeta^2+\homega^2}$ is the radial distance from the
origin in the $(\zeta,\homega)$ plane.  Relations \eqref{comegaeta}
are equivalent with:
\ben
\label{q12}
q_1=-\frac{\zeta}{\zeta^2+\homega^2}~~,~~q_2=-\frac{\homega}{\zeta^2+\homega^2}~~.
\een
\end{proof}

\begin{remark}
In every fiber of $T\cM_0$, the zero tangent vector corresponds to the
origin of the $(q_1,q_2)$ plane and to the point at infinity in the
$(\zeta,\homega)$ and $(\hetap,\homega)$ planes.
\end{remark}

\begin{prop}
\label{prop:NHkappaomegaeta}
The following relations hold on $\dot{T}\cM_0$:
\ben
\label{speednormomegaeta}
\cN=\sqrt{-\varPhi+\sqrt{\varPhi^2+\frac{M_0^2\varPsi^2}{(1-\hetap)^2+\homega^2}}}~~
\een
\ben
\label{cHomegaeta}
\cH=\frac{1}{M_0}\sqrt{\varPhi+\sqrt{\varPhi^2+\frac{M_0^2\varPsi^2}{(1-\hetap)^2+\homega^2}}}~~.
\een
\ben
\label{kappaomegaeta}
\kappa=\frac{1}{2}\left[-1+\sqrt{1+\frac{4(||\Xi||^\v)^2}{(1-\hetap)^2+\homega^2}}\right]~~.
\een
Moreover, $\kappa,\homega$ and $\hetap$ satisfy:
\ben
\label{Xikappaetaomega}
(1-\hetap)^2+\homega^2=\frac{(||\Xi||^\v)^2}{\kappa(1+\kappa)}~~\mathrm{on}~~\dot{T}\cM_0~~.
\een
\end{prop}

\begin{proof}
The first three relations follow from Lemma \ref{lemma:NHc} and
equation \eqref{kappac} using the first identity in \eqref{comegaeta}
of Lemma \ref{lemma:rtcons}. Relation \eqref{Xikappaetaomega} follows
from \eqref{ckappa} and \eqref{pos}.
\end{proof}

\begin{remark}
Relation \eqref{Xikappaetaomega} shows that fixing $\kappa,\hetap$ and
$\homega$ determines $||\Xi||=\frac{M_0}{2}||\dd\log \Phi||$.  This
gives an Eikonal equation for $\log\Phi$ (with index of refraction
$\frac{2}{M_0}||\Xi||$) on the connected Riemannian manifold
$(\cM,\cG)$, which can be used to determine the homothety class of $\Phi$
when the critical locus is known. 
\end{remark}

\subsubsection{Fiberwise slow roll coordinates.}

Consider the following open subsets of $\dot{T}\cM_0$:
\be
T^\pm\cM_0\eqdef \{u\in \dot{T}\cM_0~\vert~\sigma(u)=\pm 1 \}=\{u\in \dot{T}\cM_0~\vert~\sign\theta(u)=\pm 1\}~~,
\ee
where $\sigma(u)$ is the characteristic signature of $u$ (see
Definition \ref{def:signed}). We have:
\be
T^\pm \cM_0=\{u\in \dot{T}\cM_0~\vert~\sign(q_1(u))=\pm 1\}~~,
\ee
where $(q_1,q_2)$ are the Cartesian conservative coordinates of
Corollary \ref{cor:q}. Moreover:
\be
T^+\cM_0\cup T^-\cM_0=T\cM_0\setminus L~~,
\ee
where $L=L^+\cup L^-\subset T\cM_0$ is the line bundle generated by
the vector field $n$.  The sets $T^+\cM_0$ and $T^-\cM_0$ are the
total spaces of fiber sub-bundle of $T\cM_0$ whose typical fiber is an
open half-plane.

\begin{lemma}
\label{lemma:srcons}
The following relations hold on $T^\lambda\cM_0$ for each $\lambda\in \{-1,1\}$:
\ben
\label{msr1}
\hetap=1+\frac{\cos\theta}{\hc}~~,~~\kappa=\frac{1}{2}\left[-1+\sqrt{1+\frac{M_0^2\varPsi^2}{\varPhi^2} \hc^2}\right]
\een
\ben
\label{msr2}
\hc=\frac{2\varPhi}{M_0\varPsi} \sqrt{\kappa(1+\kappa)}~~,~~\theta=\lambda \arccos\left[\frac{2\varPhi}{M_0\varPsi}(\hetap-1)\sqrt{\kappa(1+\kappa)}\right]~~.
\een
\end{lemma}

\begin{proof}
The first relation in \eqref{msr1} is the first equation in
\eqref{etaomegactheta}, while the second follows immediately from
\eqref{kappac} and the definition \eqref{Xidef} of $\Xi$.  The first
relation in \eqref{msr2} follows from \eqref{ckappa}. The second relation in
\eqref{msr2} follows from the first relations in \eqref{msr1} and
\eqref{msr2} and the fact that $\sign \theta=\lambda$ on $T^\lambda
\cM_0$.
\end{proof}

\begin{prop}
\label{prop:srcoords}
For any $\lambda\in \{-1,1\}$, the pair of functions $(\kappa,
\hetap)$ is a fiberwise coordinate system on $T^\lambda \cM_0$ in the
sense of Definition \ref{def:fiberwise}, which we call {\em fiberwise
  slow roll coordinates}.
\end{prop}

\begin{proof}
Follows immediately from Lemma \ref{lemma:srcons} and Proposition \ref{prop:conscoords}.
\end{proof}

\begin{remark}
Recall the Cartesian conservative coordinates of Corollary
\ref{cor:q}. Relations \eqref{msr1} can be written as:
\be
\hetap=1+\frac{q_1}{q_1^2+q_2^2}~~,~~\kappa=\frac{1}{2}\left[-1+\sqrt{1+\frac{M_0^2\varPsi^2}{\varPhi^2} (q_1^2+q_2^2)}\right]
\ee
and define a branched double covering map $\R^2\rightarrow \R^2$ which is
invariant under the transformation $(q_1,q_2)\rightarrow (q_1,-q_2)$.
The ramification set is the line with equation $q_2=0$, while the
branching locus is the line with equation $\eta=1$. In each fiber of $T\cM_0$,
the ramification set corresponds to the intersection of $L$ with that fiber, whose
complement in the fiber is the union of the corresponding fibers of $T^\pm \cM_0$. 
\end{remark}

\begin{cor}
\label{cor:msr}
Let
\ben
\label{hdef}
h\eqdef 2\sqrt{\kappa(1+\kappa)}:T\cM\rightarrow \R_{\geq 0}~~.
\een  
For all $\lambda\in \{-1,1\}$, the pair $(h,\hetap)$ is a system
of fiberwise coordinates on $T^\lambda\cM_0$, which we call {\em
  modified slow roll coordinates}.
\end{cor}

\begin{proof}
Since $\kappa$ is positive on $\dot{T}\cM_0$, relation \eqref{hdef}
determines $\kappa$ in terms of $h$ as:
\be
\kappa=\frac{1}{2}\left[-1+\sqrt{1+h^2}\right]~~.
\ee
The conclusion follows using Proposition \ref{prop:srcoords}.
\end{proof}

\section{Natural phase space coordinates}
\label{sec:phase}

Recall the notion of {\em special local coordinate system} on $T\cM$
introduced in Definition \ref{def:special}. In Sections
\ref{sec:adapted} and \ref{sec:fibcoords}, we showed that each of the
pairs $(v_n,v_\tau)$, $(\cN,\theta)$, $(\hc,\theta)$,
$(\hetap,\homega)$, $(\kappa,\hetap)$ and $(h,\hetap)$ is a system
fiberwise coordinates on an open submanifold of $T\cM_0$.

\begin{definition}
A special phase space coordinate system $(x^1,x^2,x^3,x^4)$ defined on
an open subset of $T\cU$ is called {\em natural} if
$(x^3,x^4)=(\varPhi,\varPsi)$.
\end{definition}

\noindent A natural coordinate system is completely specified by its
two fiberwise coordinates. In the following, we consider the
following natural phase space local coordinate systems:

\begin{itemize}
\item The {\em phase space adapted coordinates}
  $(v_n,v_\tau,\varPhi,\varPsi)$ on neighborhoods of $T\cU$, which
  already appeared in Section \ref{sec:adapted}. These coordinates are
  fiberwise linear.
\item The {\em norm-angle coordinates} $(\cN,\theta,\varPhi,\varPsi)$
  on neighborhoods of $\dot{T}\cU$
\item The {\em conservative coordinates} $(c,\theta,\varPhi,\varPsi)$
  on neighborhoods of $\dot{T}\cU$
\item The {\em roll-turn coordinates}
  $(\hetap,\homega,\varPhi,\varPsi)$ on neighborhoods of $\dot{T}\cU$
\item The {\em slow roll coordinates}
  $(\kappa,\hetap,\varPhi,\varPsi)$ on neighborhoods of $T^\pm\cU$
\item The {\em modified slow roll coordinates}
  $(h,\hetap,\varPhi,\varPsi)$ on neighborhoods of $T^\pm\cU$.
\end{itemize}

\noindent The cosmological semispray in adapted phase space
coordinates was given in Corollary \ref{cor:Sadapted}. Below, we give
the components of $S$ in norm-angle, conservative, roll-turn and
modified slow roll coordinates. These follow by applying the corresponding
changes of variables to components of $S$ in adapted coordinates. 

\subsection{The cosmological semispray in norm-angle coordinates}

\begin{prop}
The cosmological semispray has the following form in norm-angle coordinates:
\be
S=S^{\cN}\frac{\pd}{\pd \cN}+S^{\theta}\frac{\pd}{\pd \theta}+S^{\varPhi}\frac{\pd}{\pd \varPhi}+S^\varPsi\frac{\pd}{\pd\varPsi}~~,
\ee
where:
\beqan
&& S^\cN=-\varPsi  \cos\theta-\frac{\cN \sqrt{\cN^2+2 \varPhi }}{M_0}\nn\\
&& S^\theta=\frac{\varPsi  \sin\theta}{\cN}-\mu \cN \cos \theta -\lambda \cN \sin\theta\nn\\  
&& S^\varPhi=\cN \varPsi  \cos\theta\\
&& S^\varPsi= \frac{1}{2} \cN (\mu \varPsi  \sin\theta +\Phi_{nn}^\v \cos\theta )~~.\nn
\eeqan
\end{prop}

\begin{proof}
Follows from Proposition \ref{prop:Scan} using the relations above.
\end{proof}

\subsection{The cosmological semispray in conservative coordinates}

\begin{prop}
\label{prop:Scons}
The restriction of the cosmological semispray to $\dot{T}\cU$ has the following form in conservative coordinates:
\be
S=S^{c}\frac{\pd}{\pd c}+S^{\theta}\frac{\pd}{\pd \theta}+S^{\varPhi}\frac{\pd}{\pd \varPhi}+S^\varPsi\frac{\pd}{\pd\varPsi}~~,
\ee
where:
{\scriptsize \beqan
&& S^c=-\frac{2}{M_0^2\varPsi} \sqrt{\left(\varPhi^2+M_0^2 \varPsi^2 \hc^2\right)
    \left(-\varPhi +\sqrt{\varPhi^2+ M_0^2 \varPsi^2\hc^2 }\right)}\nn\\
&&- \frac{1}{2M_0\varPsi^2}\sqrt{\varPhi +\sqrt{\varPhi^2+M_0^2 \varPsi^2 \hc^2
      }} \left[2 \varPsi ^2 \cos\theta +\left(-\varPhi +\sqrt{\varPhi ^2+M_0^2\varPsi^2 \hc^2 }\right) (\mu \varPsi
  \sin\theta +\Phi_{nn}^\v\cos\theta )\right]\nn\\
&& S^\theta=\frac{\varPsi \sin\theta}{\sqrt{-\varPhi
    +\sqrt{\varPhi^2+M_0^2 \varPsi^2 \hc^2}}}-(\lambda \sin\theta +\mu\cos\theta)\sqrt{-\varPhi +\sqrt{\varPhi^2+M_0^2 \varPsi^2 \hc^2}} \nn\\
&& S^\varPhi=\varPsi  \cos\theta \sqrt{-\varPhi +\sqrt{\varPhi^2+M_0^2 \varPsi ^2 \hc^2}}\\
&& S^\varPsi=\frac{1}{2}(\mu \varPsi  \sin\theta +\Phi_{nn}^\v \cos\theta) \sqrt{-\varPhi +\sqrt{\varPhi^2+M_0^2 \varPsi ^2 \hc^2}} \nn~~.
\eeqan}
\end{prop}

\subsection{The cosmological semispray in roll-turn coordinates}

\begin{prop}
\label{prop:Srt}
The restriction of the cosmological semispray to $\dot{T}\cU$ has the following form in roll-turn coordinates:
\be
S=S^{\hetap}\frac{\pd}{\pd \hetap}+S^{\homega}\frac{\pd}{\pd \homega}+S^{\varPhi}\frac{\pd}{\pd \varPhi}+S^\varPsi\frac{\pd}{\pd\varPsi}~~,
\ee
where:
\!\!\!\!{\scriptsize \beqan
&&\!\!\!\!\!\!\!\!\!\!\!\!\!\!S^{\hetap}\!=\!-\frac{\varPsi\homega^2}{\sqrt{(1-\hetap)^2+\homega^2}\sqrt{-\varPhi+\sqrt{\varPhi^2+\frac{M_0^2\varPsi^2}{(1-\hetap)^2+\homega^2}}}}+
\homega[\lambda\homega+(1-\hetap)\mu]\sqrt{\frac{-\varPhi+\sqrt{\varPhi^2+\frac{M_0^2\varPsi^2}{(1-\hetap)^2+\homega^2}}}{(1-\hetap)^2+\homega^2}}
\nn\\
&&+4|1-\hetap|\varPsi \sqrt{(1-\hetap)^2+\homega^2} \sqrt{\varPhi^2+\frac{M_0^2\varPsi^2}{(1-\hetap)^2+\homega^2}}\sqrt{-\varPhi+\sqrt{\varPhi^2+\frac{M_0^2\varPsi^2}{(1-\hetap)^2+\homega^2}}}\nn\\
&&+M_0\sqrt{\varPhi+\sqrt{\varPhi^2+\frac{M_0^2\varPsi^2}{(1-\hetap)^2+\homega^2}}} \left[2 (1-\hetap)^2\varPsi^2-(1-\hetap)(\mu\homega\varPsi +(1-\hetap)\Phi_{nn}^\v)\left(
  \varPhi-\sqrt{\varPhi^2+\frac{M_0^2\varPsi^2}{(1-\hetap)^2+\homega^2}}\right)\right]\nn\\
&&\!\!\!\!\!\!\!\!\!\!\!\!\!\!S^{\homega}\!=\!-2\sign(1-\hetap) \frac{\homega}{M_0^2\varPsi}\sqrt{(1-\hetap)^2+\homega^2}\sqrt{\varPhi^2+\frac{M_0^2\varPsi^2}{(1-\hetap)^2+\homega^2}}
\sqrt{-\varPhi+\sqrt{\varPhi^2+\frac{M_0^2\varPsi^2}{(1-\hetap)^2+\homega^2}}}\nn\\
&& -\frac{(1-\hetap)\homega\varPsi}{\sqrt{(1-\hetap)^2+\homega^2}\sqrt{-\varPhi+\sqrt{\varPhi^2+\frac{M_0^2\varPsi^2}{(1-\hetap)^2+\homega^2}}}}+(1-\hetap)(\lambda\homega+(1-\hetap)\mu)
\sqrt{\frac{-\varPhi+\sqrt{\varPhi^2+\frac{M_0^2\varPsi^2}{(1-\hetap)^2+\homega^2}}}{(1-\hetap)^2+\homega^2}}\nn\\
&&+\frac{\homega(\mu\homega\varPsi+(1-\hetap)\Phi_{nn}^\v)}{2M_0\varPsi^2}\sqrt{\varPhi+\sqrt{\varPhi^2+\frac{M_0^2\varPsi^2}{(1-\hetap)^2+\homega^2}}}
\left(\sqrt{\varPhi^2+\frac{M_0^2\varPsi^2}{(1-\hetap)^2+\homega^2}}-\varPhi-2(1-\hetap)\varPsi^2\right)\nn\\
&&\!\!\!\!\!\!\!\!\!\!\!\!\!\!S^\varPhi\!=\!|1-\hetap|\varPsi\sqrt{\frac{-\varPhi+\sqrt{\varPhi^2+\frac{M_0^2\varPsi^2}{(1-\hetap)^2+\homega^2}}}{(1-\hetap)^2+\homega^2}} \nn\\
&&\!\!\!\!\!\!\!\!\!\!\!\!\!\!S^\varPsi\!=\!\frac{1}{2}\sign(1-\hetap)[\mu\homega\varPsi+(1-\hetap)\Phi_{nn}^\v]\sqrt{\frac{-\varPhi+\sqrt{\varPhi^2+\frac{M_0^2\varPsi^2}{(1-\hetap)^2+\homega^2}}}{(1-\hetap)^2+\homega^2}}~~.
\eeqan}
\end{prop}

\subsection{The cosmological semispray in modified slow roll coordinates}

\begin{prop}
\label{prop:Scan}
The restriction of the cosmological semispray to $T^+\cU$ has the following form in modified slow roll coordinates:
\be
S=S^{\hetap}\frac{\pd}{\pd \hetap}+S^{h}\frac{\pd}{\pd h}+S^{\varPhi}\frac{\pd}{\pd \varPhi}+S^\varPsi\frac{\pd}{\pd\varPsi}~~,
\ee
where:
{\scriptsize \beqan
&&\!\!\!\!\!\!\!\!\!\!\!\!\!\!S^{\hetap}=-\frac{1}{2 M_0 h \varPhi^4\varPsi^2 \sqrt{-1+\sqrt{1+h^2 \varPsi^4}}}\Big[2 M_0^2\varPhi^{5/2}\varPsi^2 +2\varPhi^{3/2}\varPsi\left(-1+\sqrt{1+h^2 \varPsi ^4}\right)
    \left(h^2(1-\hetap)^2 \varPsi^2-M_0^2 \varPhi^2\right) \lambda- \nn\\
&& -(1-\hetap)\left(-4\varPhi^{5/2}-4h^2\sqrt{\varPhi}(-1+\hetap +\varPhi^2)\varPsi^4+4 \sqrt{\varPhi^5(1+h^2 \varPsi^4)}\right)+\nn\\
&& +3(1-\hetap) h\varPsi^2\left(-\varPhi^{3/2} \sqrt{-h^2(-1+\hetap)^2\varPsi^2+M_0^2\varPhi^2}+\sqrt{\varPhi^3 (1+h^2 \varPsi^4)(-h^2 (-1+\hetap)^2 \varPsi^2+M_0^2 \varPhi^2)}\right) \mu -\nn\\
&& -h^2\varPsi^2(-1+\hetap)^2 \left(-\varPhi^{3/2}+\sqrt{\varPhi^3(1+h^2\varPsi^4)}\right) \Phi_{nn}^\v \Big]\nn
\eeqan
\beqan  
&&\!\!\!\!\!\!\!\!\!\!\!\!\!\!\!S^h\!=\!\frac{\varPsi^2}{2M_0\sqrt{\frac{ \varPhi^3 (1+\sqrt{1+h^2})(1+h^2\varPsi^4)}{1+h^2}}}
\Big[-4\varPhi ^2\sqrt{\left(1+\sqrt{1+h^2}\right) \left(1+h^2 \varPsi ^4\right)\left(-1+\sqrt{1+h^2 \varPsi ^4}\right)}-\nn\\
  && -2 (-1+\hetap ) h^2 \varPsi^4 \sqrt{\left(1+\sqrt{1+h^2}\right) \left(-1+\sqrt{1+h^2 \varPsi^4}\right)}+\nn\\
&&+ 2(-1+\hetap)\varPsi^2 \left(\sqrt{\left(-1+\sqrt{1+h^2}\right)\left(1+\sqrt{1+h^2 \varPsi ^4}\right)} +\sqrt{\left(1+h^2\right)
    \left(-1+\sqrt{1+h^2}\right) \left(1+\sqrt{1+h^2 \varPsi
      ^4}\right)}\right)+\nn\\
&&+ \varPhi \Big(\sqrt{-\left(1+\sqrt{1+h^2}\right)\left(1+\sqrt{1+h^2 \varPsi ^4}\right)\left(h^2(-1+\hetap)^2\varPsi^2-M_0^2\varPhi ^2\right)}-\nn\\
 &&-\sqrt{-\left(1+\sqrt{1+h^2}\right)\left(1+h^2 \varPsi^4\right)\left(1+\sqrt{1+h^2\varPsi^4}\right)\left(h^2(-1+\hetap )^2 \varPsi^2-M_0^2\varPhi^2\right)}+\nn\\
&& +\varPsi^2 \sqrt{-\left(-1+\sqrt{1+h^2}\right) \left(-1+\sqrt{1+h^2 \varPsi
      ^4}\right) \left(h^2 (-1+\hetap )^2 \varPsi ^2-\varPhi ^2
    M_0^2\right)}+\nn\\
&&+\varPsi^2\sqrt{-\left(1+h^2\right)\left(-1+\sqrt{1+h^2}\right)\left(-1+\sqrt{1+h^2 \varPsi ^4}\right)\left(h^2(-1+\hetap )^2 \varPsi^2-\varPhi^2 M_0^2\right)}\Big) \mu +\nn\\
&& -(1-\hetap ) \varPhi \Big(h^2 \varPsi^2\sqrt{\left(1+\sqrt{1+h^2}\right) \left(-1+\sqrt{1+h^2 \varPsi^4}\right)}+\sqrt{\left(-1+\sqrt{1+h^2}\right)\left(1+\sqrt{1+h^2 \varPsi ^4}\right)}+\nn\\
  &&+\sqrt{\left(1+h^2\right)\left(-1+\sqrt{1+h^2}\right) \left(1+\sqrt{1+h^2 \varPsi^4}\right)} -\sqrt{\left(-1+\sqrt{1+h^2}\right) \left(1+h^2 \varPsi^4\right) \left(1+\sqrt{1+h^2 \varPsi^4}
    \right)}-\nn\\
  && - \sqrt{\left(1+h^2\right) \left(-1+\sqrt{1+h^2}\right) \left(1+h^2 \varPsi ^4\right)\left(1+\sqrt{1+h^2 \varPsi ^4}\right)}\Big) \Phi_{nn}^\v\Big]\nn\\
&&\!\!\!\!\!\!\!\!\!\!\!\!\!\!S^\varPhi\!=\!-\frac{h (1-\hetap ) \varPsi^2}{M_0} \sqrt{\frac{-1+\sqrt{1+h^2 \varPsi ^4}}{\varPhi }}\nn\\
&&\!\!\!\!\!\!\!\!\!\!\!\!\!\!S^\varPsi\!=\!\frac{\varPsi}{2M_0}  \sqrt{\mu \frac{-1+\sqrt{1+h^2 \varPsi ^4}}{\varPhi }} \left(\sqrt{M_0^2\varPhi ^2 -h^2 (1-\hetap )^2 \varPsi ^2} - h ( 1-\hetap ) \Phi_{nn}^\v \right)~~.
\eeqan}
\end{prop}

\noindent The formidable expressions given above in roll-turn
and modified slow roll coordinates can be expanded to first order in
various regions to extract much simpler expressions which can be used
to perform, for example, a general study of the second order slow roll
regime or of the slow roll-rapid turn regime.

\section{Constant roll conditions and constant roll manifolds}
\label{sec:SR}

\noindent The basic observable $\kappa$ and $\hetap$ allow one to
describe various loci in $T\cM$ which are relevant to the slow roll
and constant roll approximations used in cosmology. Below, we give a
geometric description of some of these regions.

\begin{definition}
The {\em first slow roll condition} is the condition $\kappa(u)\ll 1$
on the tangent vector $u\in T\cM$.
\end{definition}

\noindent Since $\kappa(u)=0$ iff $u=0$ (see eq. \eqref{Nkappa}), the
first slow roll condition is satisfied iff $u$ lies close to the image
$Z$ of the zero section of $T\cM$. Hence $Z$ plays the role of {\em
  first slow roll manifold}.  Notice that a cosmological curve
$\varphi:I\rightarrow \cM$ satisfies $\dot{\varphi}(t)\in Z$ iff $t\in
I_\sing$.

\begin{definition}
The {\em (adiabatic) constant roll condition} with parameter $C\in \R$
is the condition $|\hetap(u)-C|\ll 1$ on the nonzero tangent vector
$u\in \dot{T}\cM$.
\end{definition}

\noindent The constant roll condition with parameter zero is called the
{\em second slow roll condition}, while the constant roll condition
with parameter one is called the {\em ultra slow roll
  condition}. The combination of the first and second slow roll
conditions is called the {\em second order slow roll condition}.
By \eqref{thetac}, the second slow roll condition
$|\hetap(u)|\ll 1$ holds for $u\in \dot{T}\cM_0$ iff
$\cos\theta(u)\approx -\hc(u)$, which in particular requires that $u$
points towards decreasing values of $\Phi$.

\begin{definition}
We say that a cosmological curve $\varphi:I\rightarrow \cM$ {\em
  satisfies the first slow (resp. second slow, ultra,
  constant with parameter $C$) roll condition} at time $t\in I$
if the tangent vector $\dot{\varphi}(t)\in T_{\varphi(t)}\cM$
satisfies the corresponding condition.
\end{definition}

\begin{prop}
\label{prop:SRcond}
The first slow roll condition $\kappa\ll 1$ is equivalent on $\dot{T}\cM_0$ with:
\ben
\label{SR1cond}
\sqrt{(1-\hetap)^2+\homega^2}\gg 2||\Xi||^\v~~.
\een
Moreover, the second order slow roll conditions are equivalent on $\dot{T}\cM_0$ with:
\ben
\label{SRcond}
|\hetap|\ll 1~~\mathrm{and}~~\sqrt{1+\homega^2}\gg 2||\Xi||^\v~~.
\een
\end{prop}

\begin{proof}
Follows immediately from relation \eqref{kappaomegaeta} of Proposition \ref{prop:NHkappaomegaeta}.
\end{proof}

\subsection{Constant roll manifolds}

\begin{definition}
The (adiabatic) {\em constant roll manifold} of $(\cM,\cG,\Phi)$ at
parameter $C$ is the following set:
\be
\cZ(C)\eqdef \overline{\{u\in \dot{T}\cM_0~\vert~\hetap(u)=C\}}~~,
\ee
where the overline denotes closure in $T\cM_0$.
\end{definition}

\noindent The constant roll manifolds at parameters $\cZ(0)$ and
$\cZ(1)$ are called the {\em second slow roll} and {\em ultra slow
  roll} manifolds, respectively.  The intersection of $\cZ(C)$ with a
fiber of $T\cM_0$ is called the {\em constant roll curve} at parameter
$C$ of that fiber. For $C=0,1$ this is called respectively the {\em
  slow roll curve} or {\em ultra slow roll curve}.

\begin{lemma}
\label{lemma:croll}
The following relations hold on $\dot{T}\cM_0$, where in the second
equality we assume $\hetap\neq 1$:
\ben
\label{unorm}
\cN=\sqrt{\varPhi\left[-1+\sqrt{1+4\hc^2 (||\Xi||^\v)^2}\right]}=
\sqrt{\varPhi\left[-1+\sqrt{1+\frac{4 (||\Xi||^\v)^2}{(1-\hetap)^2} \cos^2\theta}\right]}~~.
\een
\end{lemma}

\begin{proof}
The first equality follows from \eqref{Nc} and \eqref{Xidef}. The
second follows the first and \eqref{etaomegactheta}.
\end{proof}

\noindent Recall the first relation in \eqref{etaomegactheta}, which can be written as:
\ben
\label{costheta}
\cos\theta=\hc(\hetap-1)~~.
\een

\begin{prop}
For any $C\neq 1$, the constant roll curve of parameter $C$ in a fiber
of $T\cM_0$ is closed simple curve which contains the origin of that
fiber. Hence $\cZ(C)$ is a 3-manifold which contains the image of the
zero section of $T\cM_0$.
\end{prop}

\begin{proof}
Setting $\hetap=C$ in Lemma \ref{lemma:croll} shows that the constant
roll curve $\Gamma_C(m)$ of parameter $C$ in the fiber $T_m\cM$ (where
$m\in \cM_0$) satisfies the implicit equation:
\ben
\label{unorm1}
||u||^2=\Phi(m)\left[-1+\sqrt{1+\frac{4 ||\Xi(m)||^2}{(C-1)^2} \cos^2\theta(u)}\right]~~(u\in T_m\cM)
\een
and hence is a portion of the curve $\Gamma'_C(m)\subset T_m\cM$ which
consists of all solutions of this equation. Since
$||u||^2=v_n(u)^2+v_\tau(u)^2$ and $\theta$ is the signed angle
between $u$ and $n(m)$, relation \eqref{unorm1} is the equation of
$\Gamma'_C(m)$ in the polar coordinates associated to the Cartesian
coordinate system defined on $T_m\cM$ by $v_n$ and $v_\tau$, whose
coordinate axes are the lines generated by the vectors $n(m)$ and
$\tau(m)$. The curve $\Gamma'_C(m)$ passes through the origin of $T_m\cM$
at $\theta=\pm \frac{\pi}{2}$. The equation is invariant under the
transformations $\theta\rightarrow -\theta$ and $\theta\rightarrow
\pi-\theta$ (with $u$ unchanged). The second of these symmetries shows
that $\Gamma'_C(m)\setminus\{0\}$ has two connected components which lie
on the two sides of the $\tau$-axis. The closure of each connected
component (which is obtained by adding the origin) is a closed simple
curve. Relation \eqref{costheta} gives:
\be
\cos\theta(u)=\hc(u)(C-1)~~,
\ee
which implies:
\be
\sign(\frac{\pi}{2}-|\theta(u)|)=\sign[\cos\theta(u)]=\sign(C-1)
\ee
since $c(u)>0$. This shows that $\Gamma_C(m)\setminus\{0\}$ coincides
with one of the connected components of $\Gamma'_C(m)\setminus\{0\}$,
which implies that $\Gamma_C(m)$ is a closed simple curve.
\end{proof}

\begin{remark}
The constant roll curve $\Gamma_C(m)$ degenerates to the line $\R\tau(m)$
orthogonal to $n(m)$ inside $T_m\cM$ in the ultra slow roll limit
$C\rightarrow 1$.  In this case, equation \eqref{speednormomegaeta}
gives:
\be
||u||=\sqrt{\Phi(m)\left[-1+\sqrt{1+\frac{4||\Xi(m)||^2}{\homega(u)^2}}\right]}~~,
\ee
where $\homega(u)\in \R^\times$. When $\hetap(u)\rightarrow \pm \infty$, the
curve degenerates to the origin of $T_m\cM$. 
\end{remark}

\noindent Consider now the special case of slow roll and
ultra-slow roll curves. Recall the fiberwise Cartesian conservative
coordinates $q_1\eqdef \hc\cos\theta, q_2\eqdef
\hc\sin\theta:\dot{T}\cM_0\rightarrow \R$ of Corollary \ref{cor:q}.

\begin{lemma}
\label{lemma:usr}
We have $\hetap=1$ in $\dot{T}\cM_0$ iff $q_1=0$. Moreover, we have
$\hetap=0$ in $\dot{T}\cM_0$ iff the value of $(q_1,q_2)$ lies on the
circle of radius $1/2$ centered at the point $(-1/2,0)$ in the
$(q_1,q_2)$ plane. In this case, the position of $(q_1,q_2)$ on this
circle determines the value of $\homega$.
\end{lemma}

\begin{proof}
The first statement follows immediately from eqs. \eqref{q12} (where we remind
the reader that $\zeta\eqdef 1-\hetap$). The same relations show that
$\hetap=0$ is equivalent with:  
\be
q_1=-\frac{1}{1+\homega^2}~~,~~q_2=-\frac{\homega}{1+\homega^2}~~
\ee
with $\omega\in \R$. This implies $\homega=\frac{q_2}{q_1}$ and:
\be
(q_1+\frac{1}{2})^2+q_2^2=\frac{1}{4}~~.
\ee
\end{proof}

\begin{figure}[H]
\centering
\begin{minipage}{.48\textwidth}
\vskip 0.3cm 
\centering \includegraphics[width=.95\linewidth]{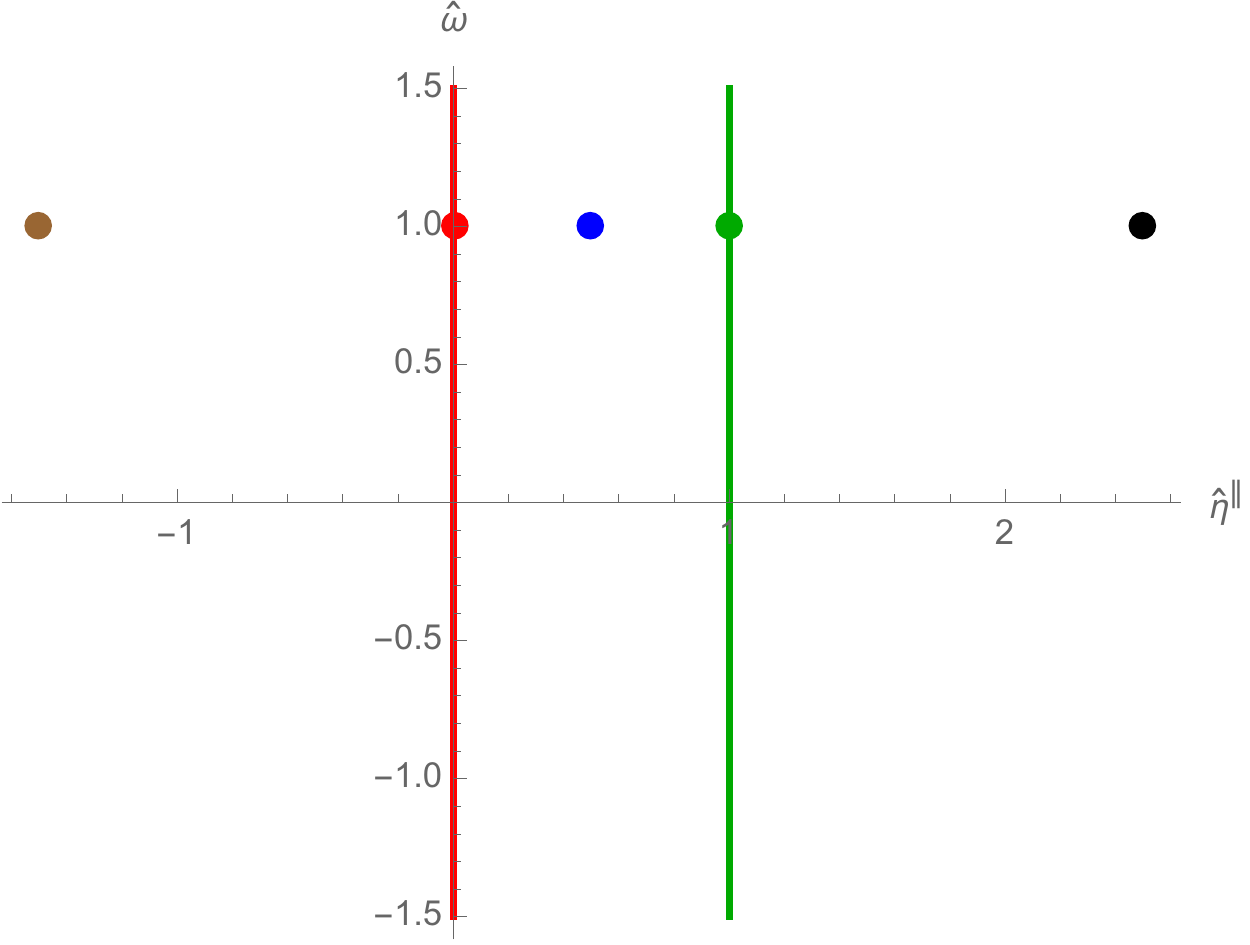}
\vskip 1.2cm
\subcaption{Five points in the $(\hetap,\homega)$ plane. The green
  and red points lie on the slow roll and ultra slow roll curves (shown
  respectively in red and green).}
\end{minipage}\hfill
\begin{minipage}{.48\textwidth}
\vskip -3mm
\centering \includegraphics[width=.8\linewidth]{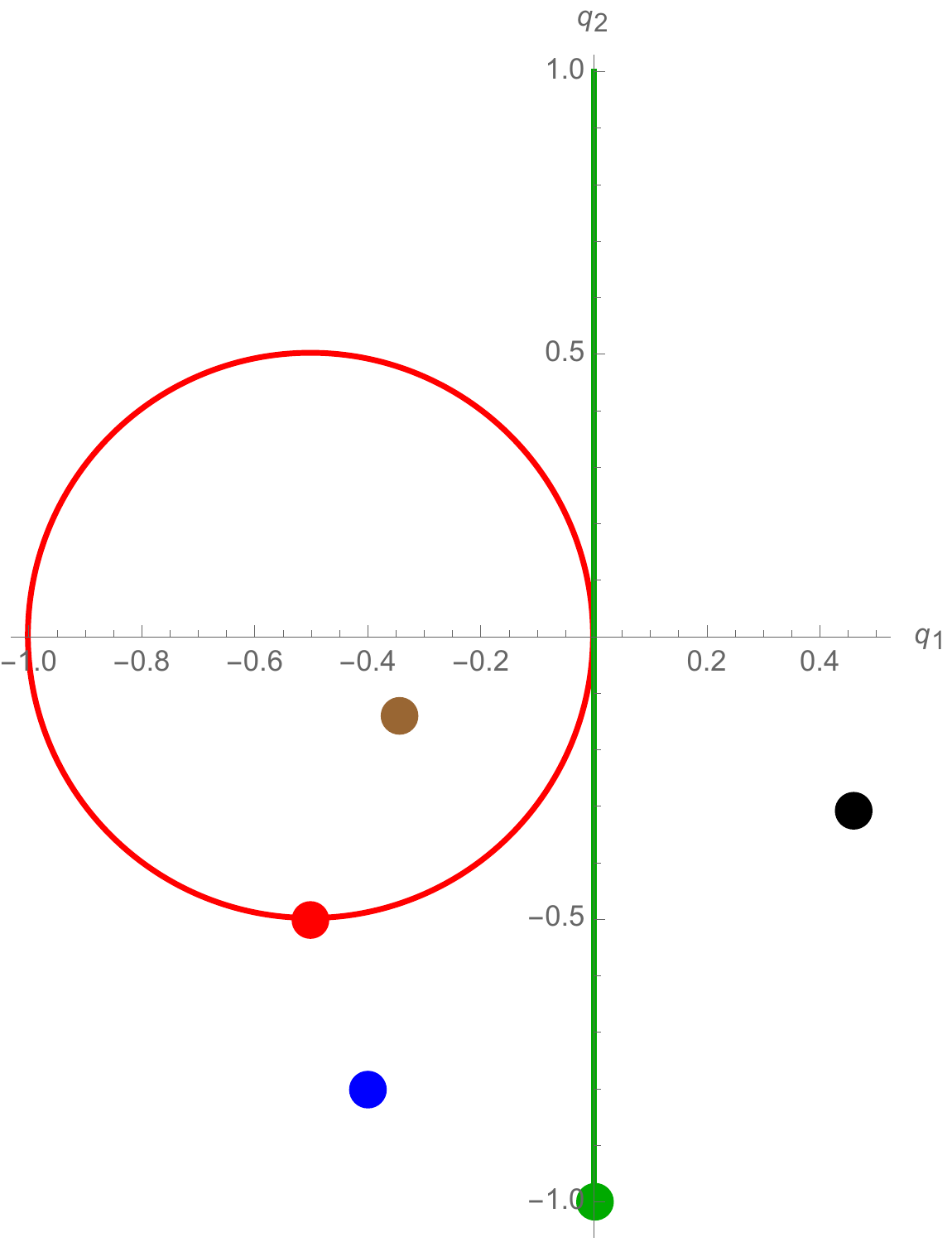}  
\subcaption{The image of the same five points in the $(q_1,q_2)$
  plane. In these coordinates, the slow roll curve (shown in red)
  becomes a circle while the ultra slow
  roll curve (shown in green) coincides with the vertical axis.}
\end{minipage}
\caption{The slow roll and ultra slow roll curves in the
  $(\hetap,\homega)$ and $(q_1,q_2)$ planes.}
\label{fig:natural}
\end{figure}

\begin{cor}
The second slow roll manifold $\cZ(0)$ is a closed three-dimensional
submanifold of $T\cM_0$ which intersects each fiber of $T\cM_0$ in a
simple closed curve. This submanifold contains the image $Z_0\eqdef
Z\cap T\cM_0$ of the zero section of $T\cM_0$. The ultra slow roll
manifold $\cZ(1)$ is a closed three-dimensional manifold of $T\cM$ which
intersects each fiber along a line and is tangent to $\cZ(0)$ along $Z_0$.
\end{cor}

\subsection{Simplifications under slow roll}

\subsubsection{Simplifications in the second slow roll regime.}

\begin{prop}
\label{prop:rtSR}
 Suppose that $|\hetap|\ll 1$. Then the following relations hold on ${\dot T}\cM_0$:
\ben
\label{tanthetadappr}
\homega\approx \tan\theta\approx -\sigma\sqrt{\frac{1}{\hc^2}-1}
\een
and:
\ben
\label{csthetaomega}
\cos\theta\approx -\frac{1}{\sqrt{1+\homega^2}}\approx -\hc~~,~~\sin\theta\approx -\frac{\homega}{\sqrt{1+\homega^2}}\approx -\homega \hc\approx \sigma\sqrt{1-\hc^2}~~.
\een
In particular, we have $\cos\theta< 0$ and hence a vector $u\in \dot{T}\cM_0$
which satisfies $|\hetap(u)|\ll 1$ points towards decreasing values of
$\Phi$. Moreover, we have $|\theta(u)|\rightarrow \frac{\pi}{2}$ in
the rapid turn limit $|\homega|\rightarrow \infty$.
\end{prop}

\begin{proof}
Follows immediately from Lemma \ref{lemma:rtcons}.
\end{proof}

\noindent Hence in the second slow roll regime the reduced turning
rate $\homega_\varphi(t)$ of a non-critical cosmological curve
$\varphi:I\rightarrow \cM_0$ at a regular time $t\in I_\reg$
determines and is determined by the signed characteristic angle
$\theta_\varphi(t)$ and in this regime we
have $|\theta_\varphi(t)|> \frac{\pi}{2}$, i.e. the tangent vector to
the curve at time $t$ points towards decreasing values of $\Phi$.

\begin{prop}
\label{prop:vnvtauSR}
Suppose that $|\hetap|\ll 1$. Then the following relations hold on $\dot{T}\cM_0$:
\beqan
\label{dotphicostheta}
&&\cN\approx \sqrt{-\varPhi +\sqrt{\varPhi^2+\frac{M_0^2\varPsi^2}{1+\homega^2}}}=\sqrt{-\varPhi +\sqrt{\varPhi^2+M_0^2\varPsi^2\cos^2\theta}}\\
&& \cH\approx \frac{1}{M_0}\sqrt{\varPhi +\sqrt{\varPhi^2+\frac{M_0^2\varPsi^2}{1+\homega^2}}}=\frac{1}{M_0}\sqrt{\varPhi +\sqrt{\varPhi^2+M_0^2\varPsi^2\cos^2\theta}}~~\nn
\eeqan
and:
\ben
\label{HN}
\cN\cH\approx \frac{\varPsi}{\sqrt{1+\homega^2}}\approx -\varPsi \cos\theta~~.
\een
Moreover, we have:
\beqan
\label{vSR1}
&&v_n\approx -\frac{1}{\sqrt{1+\homega^2}}\sqrt{-\varPhi +\sqrt{\varPhi^2+\frac{M_0^2\varPsi^2}{1+\homega^2}}}=\sqrt{-\varPhi +\sqrt{\varPhi^2+M_0^2\varPsi^2\cos^2\theta}}
\cos\theta\nn\\
&&v_\tau\approx -\frac{\homega}{\sqrt{1+\homega^2}}\sqrt{-\varPhi +\sqrt{\varPhi^2+\frac{M_0^2\varPsi^2}{1+\homega^2}}} =\sqrt{-\varPhi +\sqrt{\varPhi^2+M_0^2\varPsi^2\cos^2\theta}}
\sin\theta~~.\nn
\eeqan
\end{prop}

\begin{remark}
Notice that in the second slow roll regime we have $\cos\theta\approx
-\hc$ and $\sin\theta=\sigma \sqrt{1-\hc^2}$. This allows us to
express the quantities in Proposition \ref{prop:vnvtauSR} in terms of
$\hc$.
\end{remark}

\begin{proof}
Propositions \ref{prop:NHkappaomegaeta} and \ref{prop:rtSR} imply
relations \eqref{dotphicostheta}, while \eqref{NHoe} implies
\eqref{HN}.  The last two relations in the statement
follow from the first equation in \eqref{dotphicostheta}
by using \eqref{vnvtau} and \eqref{csthetaomega}.
\end{proof}

\subsubsection{Simplifications in the second order slow roll regime.}

\noindent The second equation in \eqref{dotphicostheta} of Proposition
\ref{prop:vnvtauSR} shows that, in the second slow roll regime, the
first slow roll condition $\kappa\ll 1$ (which is equivalent with
$\cH\approx \frac{\sqrt{2\varPhi}}{M_0}$) amounts to the requirement:
\be
-M_0\frac{\varPsi}{\varPhi}\cos\theta\ll 1~~\Longleftrightarrow -2||\Xi||^\v \cos\theta\ll 1~~.
\ee

\noindent Let:
\be
V=M_0\sqrt{2\Phi}
\ee
be the {\em classical effective potential} of the model
$(M_0,\cM,\cG,\Phi)$ (see reference \cite{ren}).

\begin{prop}
In the second order slow roll regime $\kappa\ll 1$ and
$|\hetap|\ll 1$, we have:
\ben
\label{normspeedSR}
\cN \approx -||\dd V||^\v\cos\theta\approx ||\dd V||^\v \hc ~~,
\een
and:
\ben
\label{vnvtauSR}
v_n\approx -||\dd V||^\v\cos^2\theta~~,~~v_\tau\approx -||\dd V||^\v\cos\theta\sin\theta~~.
\een
\end{prop}

\begin{proof}
Relation \eqref{HN} gives:
\ben
\label{normspeedSR2}
\cN \approx -\frac{\varPsi}{\cH}\cos\theta~~.
\een
Since we also impose the first slow roll condition, we have $\cH\approx
\frac{\sqrt{2\varPhi}}{M_0}$ and \eqref{normspeedSR2} becomes \eqref{normspeedSR}.
Combining this with \eqref{vnvtau} gives \eqref{vnvtauSR}.
\end{proof}

\section{The quasi-conservative and strongly dissipative regimes}
\label{sec:cons}

\begin{definition}
The {\em quasi-conservative condition} on a vector $u\in T\cM_0$ is
the condition $c(u)\ll 1$.  The {\em strong dissipation condition} on
$u$ is the condition $c(u)\gg 1$.
\end{definition}

\begin{definition}
We say that a noncritical cosmological curve $\varphi:I\rightarrow
\cM_0$ satisfies the quasi-conservative or strong dissipation
condition at time $t\in I$ when $c_\varphi(t)=c(\dot{\varphi}(t))$ is
much smaller (respectively much larger) than one.
\end{definition}

\noindent The cosmological equation \eqref{eomsingle} can be written as:
\be
\nabla_t\dot{\varphi}(t)+||\dd\Phi||\left[c_\varphi(t) T_\varphi(t) +n_\varphi(t)\right]=0~~.
\ee
When $\varphi$ satisfies the quasi-conservative condition at time $t_0$,
this can be approximated by the {\em conservative equation}:
\be
\nabla_t\dot{\varphi}(t)+(\grad\Phi)(\varphi(t))=0
\ee
for $t$ close to $t_0$. The conservative equation describes motion of
a particle in the Riemannian manifold $(\cM,\cG)$ under the influence
of the potential $\Phi$ (which corresponds to a conservative dynamical
system). When $\varphi$ satisfies the strong dissipation condition at time
$t_0$, the cosmological equation can be approximated by the {\em
  modified geodesic equation}:
\be
\nabla_t\dot{\varphi}(t)+\cH_\varphi(t)\dot{\varphi}(t)=0
\ee
for $t$ close to $t_0$. The solutions of this equation are
reparameterized geodesics of $(\cM,\cG)$. Some more detail about the
quasi-conservative and strongly dissipative approximations can be
found in \cite{cons}.

\begin{remark}
Recall equation \eqref{pos}:
\ben
\label{mrel}
(\hetap-1)^2+\homega^2=\frac{1}{\hc^2}~~.
\een
This shows that the {\em rapid turn condition} $|\homega|\gg 1$ is
equivalent with $\frac{1}{\hc}\gg \sqrt{1+(\hetap-1)^2}$, which
reduces to the quasi-conservative condition $\hc\ll 1$ in the second
slow roll regime $|\hetap|\ll 1$. Hence the rapid turn and
quasi-conservative limits are equivalent within the second slow roll
regime. Recall that the condition $\hetap \approx 1$ defines the
adiabatic ultra slow roll regime. In this case, we have
$\homega\approx -\frac{\sigma}{\hc} \neq 0$ and $\theta\approx\sigma
\frac{\pi}{2}$.
\end{remark}

The first relation in \eqref{etaomegactheta} requires:
\be
\hetap\in \left[1-\frac{1}{\hc},1+\frac{1}{\hc}\right]
\ee
and can be written as \eqref{costheta}. In particular, the interval
within which $\hetap$ can take values is centered at $1$ and
constrained by the value of $\hc$. This interval is very large in the
quasi-conservative regime $\hc\ll 1$ and becomes narrow in the
strongly dissipative regime $\hc\gg 1$, when $\hetap$ is forced to be
close to one. In particular, the adiabatic ultra slow roll
approximation $\hetap\approx 1$ is accurate in the strongly
dissipative regime.

\section{The horizontal approximation in adapted coordinates}
\label{sec:horadapted}

In this section, we study the horizontal approximation defined by the
adapted local coordinates on a neighborhood of a point of $T\cU$,
where $\cU$ is the non-degenerate set of $(\cM,\cG,\Phi)$.

\subsection{The adapted speed component functions}

\begin{lemma}
\label{lemma:homega}
The following relations hold on $\dot{T}\cM_0$:
\ben
\label{vetaomega}
v_n= -\frac{\varPsi}{\cH}\frac{1-\hetap}{(1-\hetap)^2+\homega^2}~~,~~v_\tau= -\frac{\varPsi}{\cH}\frac{\homega}{(1-\hetap)^2+\homega^2}~~.
\een
and:
\ben
\label{kappaetaomega}
\kappa=\frac{1}{2\varPhi}\left(\frac{\varPsi}{\cH}\right)^2\frac{1}{(1-\hetap)^2+\homega^2}~~.
\een
\end{lemma}

\begin{proof}
Combining \eqref{vnvtau} with \eqref{comegaeta} and
\eqref{NHoe} gives \eqref{vetaomega}, which implies
\eqref{kappaetaomega} upon using $\kappa=\frac{\cN^2}{2\varPhi}=\frac{\v_n^2+v_\tau^2}{2\varPhi}$.
\end{proof}

\begin{prop}
The following relations hold on $\dot{T}\cM_0$:
\beqan
\label{vetaomega2}
&&v_n= -\frac{M_0\varPsi}{\sqrt{\varPhi+\sqrt{\varPhi^2+\frac{M_0^2\varPsi^2}{(1-\hetap)^2+\homega^2}}}} \frac{1-\hetap}{(1-\hetap)^2+\homega^2}\nn\\
&&v_\tau= - \frac{M_0\varPsi}{\sqrt{\varPhi+\sqrt{\varPhi^2+\frac{M_0^2\varPsi^2}{(1-\hetap)^2+\homega^2}}}}\frac{\homega}{(1-\hetap)^2+\homega^2}~~.
\eeqan
\end{prop}

\begin{proof}
Equations \eqref{vetaomega} give \eqref{vetaomega2} upon using
Proposition \ref{prop:NHkappaomegaeta}.
\end{proof}

\subsection{The adapted parameters and adapted functions}

\noindent Let:
\be
\!Q\!\eqdef\!\{\!u\in T\cM_0~\vert~v_n(u)v_\tau(u)\!=\!0\!\}\!=\!\{u\in T\cM_0~\vert~\theta(u)\not\in \!\{-\frac{\pi}{2},\!0,\!\frac{\pi}{2},\!\pi\!\}\}\!=\!T\cM_0\setminus (L_n\cup L_\tau)
\ee
and notice that $Z\subset Q$. Here $L_n$ and $L_\tau$ are the line bundles generated by the vector fields $n$ and $\tau$ on $\cM_0$.

\begin{definition}
A noncritical curve $\varphi:I\rightarrow \cM_0$ is called {\em
  admissible} if its canonical lift $\dot{\varphi}:I\rightarrow
T\cM_0$ does not intersect $Q$.
\end{definition}

\noindent Notice that any admissible curve is nondegenerate. 

\begin{definition}
The {\em adapted parameters} of an admissible curve
$\varphi:I\rightarrow \cM_0$ are defined through:
\ben
\label{adaptedparam}
\f^n_\varphi\eqdef -\frac{\dot{v}_n^\varphi}{\cH_\varphi v^\varphi_n}~~\mathrm{and}~~\f^r_\varphi\eqdef -\frac{\dot{v}^\varphi_\tau}{\cH_\varphi v^\varphi_\tau}~~.
\een
\end{definition}

\begin{definition}
\label{def:fnftau}
The {\em adapted functions} $\f^n,\f^\tau:T\cM_0\setminus Q\rightarrow \R$ of
$(\cM,\cG,\Phi)$ are the basic local scalar observables defined through:
\beqan
\label{fnftau}
\!\!\!&&\!\!\!\!\!\!\!\!\!\f^n\!\eqdef\! -\frac{S^{v_n}}{\cH v_n}\!=\!1+\frac{\varPsi-\lambda^\v v_\tau^2-\mu^\v v_n v_\tau}{\cH v_n}\!=\!1+
\frac{\varPsi}{\cH\cN\cos\theta}\!-\!\frac{\cN (\lambda^\v \sin\theta+\mu^\v\cos\theta)}{\cH}\tan\theta\nn\\
\!\!\!&&\!\!\!\!\!\!\!\!\!\f^\tau\!\eqdef\!-\frac{S^{v_\tau}}{\cH v_\tau}=1+\frac{\mu^\v v_n^2+\lambda^\v v_nv_\tau}{\cH v_\tau}=1+\frac{\cN (\lambda^\v\sin\theta+\mu^\v \cos\theta)}{\cH \tan\theta}~~,
\eeqan
where $S^{v_n}$ and $S^{v_\tau}$ are the vertical components of the
cosmological semispray in adapted phase space coordinates (see
Corollary \ref{cor:Sadapted}).
\end{definition}

\begin{prop}
\label{prop:fobs}
Suppose that $\varphi$ is an admissible cosmological curve. Then we have:
\be
\f^n_\varphi=\f^n\circ \dot{\varphi}~~\mathrm{and}~~\f^\tau_\varphi=\f^\tau\circ \dot{\varphi}~~.
\ee
\end{prop}

\begin{proof}
Follows immediately from Lemma \ref{lemma:cosmadapted}.
\end{proof}

\subsection{Relation to roll-turn fiberwise coordinates}

\begin{prop}
\label{prop:etaf}
The following relations hold on $T\cM_0\setminus Q$:
\ben
\label{hetatheta}
\f^n \cos^2\theta  + \f^\tau \sin^2\theta =1+\frac{\cos\theta}{\hc}=\hetap~~.
\een
and:
\ben
\label{hetapf}
\hetap=\frac{[1+2(\f^n-\f^\tau)\hc^2]\pm\sqrt{[1+2(\f^n-\f^\tau)\hc^2]^2-4(\f^n-\f^\tau)\hc^2
[\f^\tau+(\f^n-\f^\tau)\hc^2]}}{2(\f^n-\f^\tau)\hc^2}~~.
\een
\end{prop}

\begin{proof}
Relation \eqref{hetatheta} follows immediately from \eqref{fnftau} upon using \eqref{NH} and the first relation in \eqref{etaomegactheta}.
Equation \eqref{costheta} implies
$\sin^2\theta=1-\cos^2\theta=1-\hc^2(\hetap-1)^2$. Thus
\eqref{hetatheta} is equivalent with:
\ben
\label{eqintmd}
\hetap=\f^\tau+(\f^n-\f^\tau)(\hetap-1)^2 \hc^2~~,
\een
which amounts to a quadratic equation for $\hetap$:
\be
(\f^n-\f^\tau) \hc^2 (\hetap)^2-[1+2(\f^n-\f^\tau)\hc^2]\hetap+\f^\tau+(\f^n-\f^\tau)\hc^2=0~~.
\ee
The solutions of this equation are \eqref{hetapf}.
\end{proof}

\begin{remark}
When $|\f^n-\f^\tau|\cos^2\theta \ll |\f^\tau|$ , relation \eqref{hetapf} reduces to:
\be
\hetap\approx \f^\tau~~
\ee
as follows from \eqref{eqintmd}.
\end{remark}

\begin{prop}
\label{prop:f}
The following relations hold on $T\cM_0\setminus Q$:
\beqan
\label{fnftauomega}
&&\f^n= \hetap-\frac{\homega^2}{1-\hetap}+\frac{1}{\cH^2}\frac{\homega}{1-\hetap}
\frac{\Phi_{\tau\tau}^\v \omega+\Phi_{n\tau}^\v(1-\hetap)}{(1-\hetap)^2+\homega^2} \nn\\
&&\f^\tau=1-\frac{1}{\cH^2}\frac{1-\hetap}{\homega} \frac{\Phi_{\tau\tau}^\v \omega+\Phi_{n\tau}^\v(1-\hetap)}{(1-\hetap)^2+\homega^2}~~,
\eeqan
where:
\ben
\label{cHetaomega}
\cH=\frac{1}{M_0}\sqrt{\varPhi+\sqrt{\varPhi^2+\frac{M_0^2\varPsi^2}{(1-\hetap)^2+\homega^2}}}~~.
\een
Moreover, we have:
\ben
\label{hetathetaom}
\hetap=\frac{(1-\hetap)^2\f^n +\homega^2 \f^\tau}{(1-\hetap)^2+\homega^2}~~.
\een
\end{prop}

\begin{proof}
Using relation \eqref{NH} and Lemma \ref{lemma:rtcons}, we compute:
\ben
\label{eq1}
1+\frac{\varPsi}{\cH\cN\cos\theta}=\frac{1}{\hc \cos\theta}=1+\frac{(1-\hetap)^2+\homega^2}{\hetap-1}=\hetap+\frac{\homega^2}{\hetap-1}
\een
and:
\ben
\label{eq2}
-\frac{\cN (\lambda^\v \sin\theta+\mu^\v\cos\theta)}{\cH}=\frac{\cN}{\cH} \frac{\lambda^\v \homega+\mu^\v(1-\hetap)}{\sqrt{(1-\hetap)^2+\homega^2}}=
\frac{1}{\cH^2}\frac{\Phi_{\tau\tau}^\v\homega+\Phi_{n\tau}^\v (1-\hetap)}{(1-\hetap)^2+\homega^2}
\een
where we used \eqref{lambdamu} and the relation:
\be
\frac{\cN}{\cH}=\frac{\varPsi}{\cH^2}\hc=\frac{\varPsi}{\cH^2} \frac{1}{\sqrt{(1-\hetap)^2+\homega^2}}~~,
\ee
which follows from \eqref{NH} and Lemma
\ref{lemma:rtcons}. The latter also gives:
\ben
\label{tan}
\tan\theta=\frac{\homega}{1-\hetap}~~.
\een
Using this, \eqref{eq1} and \eqref{eq2} in the definition of $\f^n$
gives the first equation in \eqref{fnftauomega}. The second equation
in \eqref{fnftauomega} follows from \eqref{eq2}, \eqref{tan} and the
definition of $\f^\tau$.  Relation \eqref{cHetaomega} was proved
before (see \eqref{cHomegaeta}). Relation \eqref{hetathetaom} follows
from \eqref{hetatheta} by using Lemma \ref{lemma:rtcons}.
\end{proof}

\begin{remark}
Suppose that $\f^n$ and $\f^\tau$ are fixed and view
\eqref{fnftauomega} as a system of equations for $\hetap$ and
$\homega$, where $\cH$ is expressed in terms of the latter using
\eqref{cHetaomega}. Then this system is equivalent with the equation:
\ben
\label{eq3}
\f^n= \hetap-\frac{\homega^2}{1-\hetap}+M_0^2\frac{\homega}{1-\hetap}\frac{\Phi_{\tau\tau}^\v \omega+\Phi_{n\tau}^\v(1-\hetap)}{
  [(1-\hetap)^2+\homega^2]\left[\varPhi+\sqrt{\varPhi^2+\frac{M_0^2\varPsi^2}{(1-\hetap)^2+\homega^2}}\right]}~~.
\een
(which is obtained using \eqref{cHetaomega} from the first equation in
\eqref{fnftauomega}) together with equation \eqref{hetathetaom}, which
can be written as:
\ben
\label{omegaf}
\homega^2=-(1-\hetap)^2\frac{\f^n-\hetap}{\f^\tau-\hetap}~~
\een
and can be used to eliminate $\homega$ in terms of $\hetap$.
Substituting this in \eqref{eq3} gives an algebraic equation for
$\hetap$ whose coefficients depend on $\f^n$ and
$\f^\tau$. Picking a solution of this algebraic equation, relation
\eqref{omegaf} determines $\homega$ as a
function of the adapted parameters.
\end{remark}

\begin{cor}
The intersection of the second slow roll manifold $\cZ(0)$ with $T\cM_0\setminus Q$ has implicit equations:
\ben
\f^n=-\homega^2 \f^\tau~~,~~\f^\tau=1-\frac{M_0^2}{\homega\left[\varPhi+\sqrt{\varPhi^2+\frac{M_0^2\varPsi^2}{1+\homega^2}}\right]} \frac{\Phi_{\tau\tau}^\v \homega+\Phi_{n\tau}^\v}{1+\homega^2}~~.
\een
where $\homega\in \R^\times$. 
\end{cor}

\begin{cor}
\label{cor:fSR2}
The following relations hold on $T\cM_0\setminus Q$ in the slow roll approximation $|\hetap|\ll 1$:
\ben
\label{fnftauomegaSR2}
\f^n\approx -\homega^2+\frac{\homega}{\cH^2}
\frac{\Phi^\v_{\tau\tau} \homega+\Phi_{n\tau}^\v}{1+\homega^2}~~,~~\f^\tau\approx -\frac{\f^n}{\homega^2}=1-\frac{1}{\cH^2 \homega} \frac{\Phi^\v_{\tau\tau} \homega+\Phi_{n\tau}^\v}{1+\homega^2}~~,
\een
where:
\ben
\label{cHetaomegaSR2}
\cH\approx \frac{1}{M_0}\sqrt{\varPhi+\sqrt{\varPhi^2+\frac{M_0^2\varPsi^2}{1+\homega^2}}}~~.
\een
Moreover, we have:
\ben
\label{hetathetaomSR2}
\hetap\approx \frac{\f^n +\homega^2 \f^\tau}{1+\homega^2}~~.
\een
\end{cor}

\begin{cor}
\label{cor:fSR}
The following relations hold on $T\cM_0\setminus Q$ in the second {\em order} slow roll approximation $\kappa\ll 1$ and $|\hetap|\ll 1$:
\ben
\label{fnftauomegaSR}
\f^n\approx -\homega^2+\frac{M_0^2}{2\varPhi}\homega
\frac{\Phi_{\tau\tau}^\v \omega+\Phi_{n\tau}^\v}{1+\homega^2}~~,~~\f^\tau\approx -\frac{\f^n}{\homega^2}=1-\frac{M_0^2}{2\varPhi \homega} \frac{\Phi^\v_{\tau\tau} \omega+\Phi^\v_{n\tau}}{1+\homega^2}~~.
\een
\end{cor}

\begin{proof}
Follows from Corollary \ref{cor:fSR2} upon recalling that
the first slow roll condition is equivalent with $\cH\approx
\frac{\sqrt{2\varPhi}}{M_0}$.
\end{proof}

\subsection{The adapted manifold}

\begin{definition}
The {\em adapted manifold} is the submanifold $\cZ_a$ of
$T\cM_0\setminus Q$ defined by the {\em strict adapted conditions}:
\ben
\label{sadaptedcond}
\f^n=\f^\tau=0~~.
\een
\end{definition}

\noindent Let $W_a$ be the the flat Ehresmann connection defined by
the adapted phase space coordinates in the vicinity of a point of
$T\cU\setminus Q$, which we shall call {\em the adapted connection}.
Definition \ref{def:fnftau} shows that the adapted manifold coincides
with the mean field manifold of $S$ relative to $W_a$, which has
equations $S^{v_n}=S^{v_\tau}=0$. Relation \eqref{hetathetaom} shows that we have
$\cZ_a\subset \cZ(0)$. 

\begin{prop}
\label{prop:sadapted}
The adapted manifold $\cZ_a$ has the following implicit equations in
fiberwise phase space roll-turn coordinates $(\hetap,\homega,\varPhi,\varPsi)$:
\beqan
\label{Za}
&&\hetap=0\\
&&(2 \lambda^\v \varPhi +\varPsi) \omega ^4+2 \mu^\v \varPhi \omega ^3+(2\lambda^\v \varPhi +\varPsi -M_0^2 (\lambda^\v)^2 \varPsi)\omega ^2 \nn\\
&&+2\mu^\v\left(\varPhi -M_0^2 \lambda^\v \varPsi \right)\omega
-M_0^2 (\mu^\v)^2\varPsi =0~~. \nn
\eeqan
In particular, $\cZ_a$ is a multisection of $T\cM_0$ which is contained in the second slow roll manifold $\cZ(0)$. 
\end{prop}

\begin{proof}
When $\f^n=\f^\tau=0$, Proposition \ref{prop:f} gives $\hetap=0$ and:
\beqan
\label{sadaptedeq}
\frac{1}{\cH^2} \frac{\varPhi_{\tau\tau}^\v \homega+\varPhi_{n\tau}^\v}{1+\homega^2}= \homega~~,
\eeqan
where:
\ben
\label{csHadapted}
\cH= \frac{1}{M_0}\sqrt{\varPhi+\sqrt{\varPhi^2+\frac{M_0^2\varPsi^2}{1+\homega^2}}}~~.
\een
Combining these equations gives:
\be
M_0^2 \frac{\Phi_{\tau\tau}^\v \homega+\Phi_{n\tau}^\v}{1+\homega^2}= \homega\left[\varPhi+\sqrt{\varPhi^2+\frac{M_0^2\varPsi^2}{1+\homega^2}}\right]~~,
\ee
which implies:
\be
\left(M_0^2 \frac{\varPhi_{\tau\tau}^\v \homega+\varPhi_{n\tau}^\v}{\homega(1+\homega^2)}-\varPhi\right)^2=\varPhi^2+\frac{M_0^2\varPsi^2}{1+\homega^2}~~.
\ee
The latter amounts to the quartic equation:
\ben
M_0^4 (\varPhi_{\tau\tau}^\v \homega+\varPhi_{n\tau}^\v)^2-2M_0^2\varPhi (\varPhi_{\tau\tau}^\v \homega+\varPhi_{n\tau}^\v)\homega (1+\homega^2)-M_0^2\varPsi^2\homega^2(1+\homega^2)=0~~,
\een
which is equivalent to the second equation in \eqref{Za} upon using \eqref{HessPhiAdapted}.
\end{proof}

\begin{cor}
\label{cor:sadaptedSR1}
When $\hetap=0$, the first slow roll regime occurs when
\ben
\label{omegain}
\sqrt{1+\homega^2}\gg M_0\frac{\varPsi}{\varPhi}~~.
\een
In this regime, the implicit equations of the adapted manifold $Z_a$ reduce to:
\ben
\label{ZaSR1}
\hetap=0~~,~~\homega^3+\left(1-\frac{M_0^2}{2\varPhi}\varPhi_{\tau\tau}^\v \right)\homega-\frac{M_0^2}{2\varPhi}\varPhi^\v_{n\tau}=0~~.
\een
In particular, the first slow roll regime can be realized along $\cZ_a$
iff the cubic equation in \eqref{ZaSR1} admits a real solution
$\homega$ which satisfies \eqref{omegain}.
\end{cor}

\begin{proof}
When $\hetap=0$, Proposition \ref{prop:SRcond} shows that the first slow roll condition $\kappa\ll 1$ amounts to:
\be
\sqrt{1+\homega^2}\gg 2||\Xi||^\v=M_0\frac{\varPsi}{\varPhi}~~.
\ee
In this case, we have $\cH\approx \frac{\sqrt{2\varPhi}}{M_0}$ and \eqref{sadaptedeq} becomes:
\be
\frac{M_0^2}{2\varPhi}(\varPhi_{\tau\tau}^\v \homega+\varPhi_{n\tau}^\v)\approx \homega(1+\homega^2)~~, 
\ee
which is equivalent with the second equation in \eqref{ZaSR1}.
\end{proof}

The cubic equation in \eqref{ZaSR1} is in depressed form:
\be
\homega^3+p \homega+q=0~~,
\ee
where:
\be
p\eqdef 1-\frac{M_0^2}{2\varPhi}\varPhi_{\tau\tau}^\v~~,~~q\eqdef -\frac{M_0^2}{2\varPhi}\varPhi^\v_{n\tau}<0~~.
\ee
The equation has discriminant $\Delta$ given by:
\be
-\Delta=4 p^3+27 q^2=4\left(1-\frac{M_0^2}{2\varPhi}\Phi_{\tau\tau}^\v\right)^3+\frac{27}{4} \frac{M_0^4}{\varPhi^2}(\Phi^\v_{n\tau})^2~~.
\ee
Distinguish the cases:
\begin{enumerate}[1.]
\item $\Delta<0$. Then the equation has one real root and two
complex roots (which are conjugate to each other). In this case, the
real root is given by Cardano's formula:
\be
\homega=\sqrt[3]{-\frac{q}{2}+\sqrt{\frac{q^2}{4}+\frac{p^3}{27}}}+\sqrt[3]{-\frac{q}{2}-\sqrt{\frac{q^2}{4}+\frac{p^3}{27}}}
\ee
and condition \eqref{omegain} requires this root to satisfy:
\be
\homega^2 \gg  M_0^2\frac{\varPsi^2}{\varPhi^2}-1~~,
\ee
which gives a complicated condition on $\Phi$. 
\item $\Delta\geq 0$. Then the equation has three real roots
  (some of which may coincide) and the condition that one of these roots
  satisfies \eqref{omegain} amounts to the requirement that the
  largest root satisfies that condition. Since $\Delta\geq 0$ and $q<0$, we have
  $p<0$ (i.e. $ \frac{M_0^2}{2\varPhi}\Phi_{\tau\tau}^\v> 1$) and
  $\frac{27 q^2}{4|p|^3}\leq 1$. Let:
\be
A\eqdef 2\sqrt{\frac{|p|}{3}}=\sqrt{\frac{4|p|}{3}}\geq 0~~.
\ee
and:
\be
B\eqdef \frac{3|q|}{A|p|}=\sqrt{\frac{27 q^2}{4|p|^3}} \in [0,1]~~.
\ee
Setting $\phi\eqdef \frac{1}{3}\arccos B\in [0,\frac{\pi}{9}]$, the three real roots are given by Viete's formulas:
\be
\homega_k=A\cos(\phi+\frac{2\pi}{3}k)~~(k=0,1,2)~~
\ee
and correspond to the projection on the $x$ axis of three equidistant
points of the unit circle. It is clear that the largest of the 
roots belongs to the interval
$[A\cos(\frac{\pi}{3}),A]=[\frac{A}{2},A]$. Hence the requirement that a
root satisfies \eqref{omegain} is equivalent with the condition:
\be
1+\frac{A^2}{4}\gg M_0^2 \frac{\Psi^2}{\Phi^2}\Longleftrightarrow 2+M_0^2\frac{\Phi_{\tau\tau}}{2\Phi} \gg 3M_0^2\frac{\Psi^2}{\Phi^2}~~,
\ee
which amounts to:
\be
M_0^2\left(\frac{\Psi^2}{\Phi^2}-\frac{\Phi_{\tau\tau}}{6\Phi}\right)\ll \frac{2}{3}~~.
\ee
\end{enumerate}

\subsection{The horizontal approximation defined by the adapted connection}

\noindent Let $\cU\subset \cM_0$ be the non-critical set of $(\cM,\cG,\Phi)$
(see Definition \ref{def:cU}) and recall that the adapted phase space
coordinates $(v_n,v_\tau, \varPhi=\varPhi,\varPsi=\varPsi)$
give a special coordinate system on $T\cU$. On this open subset of $T\cM_0$, the
horizontal projection of the cosmological semispray has the form:
\be
S_H=S^\varPhi\frac{\pd}{\pd \varPhi}+S^\varPsi\frac{\pd}{\pd \varPsi}~~,
\ee
where (see Corollary \ref{cor:Sadapted}):
\beqan
&& S^\varPhi=\varPsi v_n\\
&& S^\varPsi=\frac{1}{2}(\Phi_{nn}^\v v_n+\mu^\v \varPsi v_\tau)~~.\nn
\eeqan
Hence the horizontal approximation defined by the adapted connection replaces locally the
cosmological equation for curves whose canonical lift lies in $T\cU\setminus Q$ with the system:
\ben
\label{horadapted}
\dot{v}_n=\dot{v}_\tau=0~~,~~\dot{\varPhi}=\varPsi v_n~~,~~\dot{\varPsi}=\frac{1}{2}(\Phi_{nn}^\v v_n+\mu^\v \varPsi v_\tau)~~.
\een
The first two equations show that $v_n$ and $v_\tau$ are constant
while the remaining equations give a system for the time evolution of
$\varPhi$ and $\varPsi$. In these equations, it is understood that the
quantities $\Phi_{nn}$ and $\mu$ are expressed as functions of $\Phi$
and $\Psi$ and hence $\Phi_{nn}^\v$ and $\mu^\v$ are expressed as
functions of $\varPhi$ and $\varPsi$. The solutions of
\eqref{horadapted} are the $W_a$-curves which serve as horizontal
approximants of cosmological curves contained in $\cU$ in the sense of
Section \ref{sec:mf}.  As explained there, the approximation is
accurate for cosmological curves which are close to the adapted
manifold $\cZ_a$.

\subsection{The adapted approximation}

\begin{definition}
The {\em adapted approximation} is defined by the {\em adapted conditions}:
\ben
\label{adaptedcond}
|\f^n|\ll 1~~\mathrm{and}~~|\f^\tau|\ll 1~~.
\een
\end{definition}

\noindent A vector $u\in T\cM_0\setminus Q$ satisfies the adapted
conditions iff it is very close to the adapted manifold $\cZ_a$.
Notice from relation \eqref{hetatheta} that the adapted conditions
imply $|\hetap|\ll 1$, so the adapted approximation implies the second
slow roll approximation; in particular, it implies the relation (see
\eqref{hetathetaomSR2}):
\be
\f^\tau\approx -\frac{\f^n}{\homega^2}~~.
\ee
However, the second slow roll condition $|\hetap|\ll 1$ does {\em not}
imply the adapted conditions (as made manifest by the fact that $\cZ_a$ is a
codimension one submanifold of the second slow roll manifold $\cZ(0)$). 

\section*{Acknowledgments}
\noindent This work was supported by national grant
  PN 19060101/2019-2022 and by a Maria Zambrano Excellency Fellowship
  as well by the COST Action COSMIC WISPers CA21106, supported by COST
  (European Cooperation in Science and Technology).

\appendix

\section{Signed angles, curvatures and turning rates on an oriented Riemann surface}
\label{app:signed}

Let $(\cM,\cG)$ be an oriented, connected and borderless Riemann
surface, whose volume form we denote by $\vol$.  In a
positively-oriented local coordinate system $(x^1,x^2)$ on $\cM$, we
have:
\be
\vol =\sqrt{\det {\hat \cG}}\dd x^1\wedge \dd x^2=\frac{1}{2}\sqrt{\det {\hat {\cG}}}\epsilon_{ij} \dd x^i\wedge \dd x^j~~,
\ee
where $\epsilon_{ij}$ is the Levi-Civita symbol and we use Einstein
summation over repeated indices. We have $\vol(e_j,e_j)=\epsilon_{ij}$
in any orthonormal and positive local frame of $T\cM$.

\subsection{The signed angle between vectors}
\label{app:signedangle}

\noindent Let $v,v_0\in T_m\cM$ be two non-zero tangent vectors to
$\cM$ at a point $m\in \cM$. Let $\Theta(v,v_0)\in [0,\pi]$ be the
ordinary angle between $v$ and $v_0$ in the Euclidean space
$(T_m\cM,\cG_m)$, which is defined through the relation:
\be
\cG(v,v_0)=||v||||v_0|| \cos\Theta(v,v_0)~~.
\ee
Notice that $\Theta(v_0,v)=\Theta(v,v_0)$. 

\begin{definition}
Suppose that $v$ and $v_0$ are linearly independent, i.e. $v$ is not
a multiple of $v_0$. Then the {\em chirality} of $v$ relative to $v_0$
is the signature of the basis $(v_0,v)$ of $T_m\cM$, which is defined
through:
\be
\upsigma_{v_0}(v)\eqdef \sign [\vol(v_0,v)]\in \{-1,1\}~~.
\ee
When $v_0$ and $v$ are linearly dependent, we define $\upsigma_{v_0}(v)=0$. 
\end{definition}

\noindent When $v_0$ and $v$ are linearly independent, we have
$\upsigma_{v_0}(v)=+1$ iff the basis $(v_0,v)$ of $T_m\cM$ is
positively oriented and $\upsigma_{v_0}(v)=-1$ iff this basis is
negatively oriented. Notice that $\sigma_{v_0}(v)=-\sigma_v(v_0)$.

\begin{definition} The {\em signed angle} of $v$ relative to $v_0$ is
the real number $\uptheta_{v_0}(v)\in (-\pi,\pi]$ defined as follows:
\begin{enumerate}
\item If $v$ is not proportional to $v_0$, then we set
$\uptheta_{v_0}(v)\eqdef \upsigma_{v_0}(v)\Theta(v_0,v)\in (-\pi,0)\cup(0,\pi)$.
\item If $v=\alpha v_0$ with $\alpha>0$, then we set
$\uptheta_{v_0}(v)\eqdef 0$.
\item If $v=\alpha v_0$ with $\alpha<0$, then we set
$\uptheta_{v_0}(v)\eqdef \pi$.
\end{enumerate}
\end{definition}

\noindent Notice that $\uptheta_v(v_0)=-\uptheta_{v_0}(v)$ when $v$
and $v_0$ are linearly independent and that
$\uptheta_v(v_0)=\uptheta_{v_0}(v)\in \{0,\pi\}$ iff $v$ and $v_0$ are
linearly dependent. In all cases, we have
$\uptheta_v(v_0)=-\uptheta_{v_0}(v)\! \mod 2\pi$ and hence:
\be
\cos\uptheta_v(v_0)=\cos\uptheta_{v_0}(v)~~,~~\sin\uptheta_v(v_0)=-\sin\uptheta_{v_0}(v)~~.
\ee
We have:
\be
\sign (\sin\uptheta_{v_0}(v))=\upsigma_{v_0}(v)~~,
\ee
with the convention $\sign(0)=0$. 

\begin{definition}
Given any non-zero vector $v\in T_m\cM\setminus \{0\}$, the {\em
  positive unit normal} to $v$ is the unique unit norm vector $N_v\in
T_m\cM$ which is orthogonal to $v$ and satisfies $\upsigma_v(N_v)=+1$,
i.e. the orthonormal basis $(v,N_v)$ of $T_m\cM$ is positively
oriented.  The {\em negative unit normal} to $v$ is the vector $-N_v$.
\end{definition}

\noindent We have:
\ben
\label{Nv}
N_v=\frac{1}{||v||}(\iota_v\vol)^\sharp~~\forall v\in \dot{T}\cM~~.
\een
Notice that:
\be
N_{\alpha v} =\sign(\alpha) N_v~~\forall \alpha\in \R^\times~~.
\ee

\noindent We have $\uptheta_v(N_v)=+\frac{\pi}{2}$ and
$\uptheta_v(-N_v)=-\frac{\pi}{2}$. The following statement is
immediate.

\begin{prop}
\label{prop:rotation}
Let $v_0,v\in T_m\cM$ be two non-vanishing vectors. Then $(v,N_v)$ is
obtained from $(v_0,N_{v_0})$ by a rotation of angle
$\uptheta_{v_0}(v)$, i.e. we have:
\beqa
v &=&~~[\cos \uptheta_{v_0}(v)] v_0+[\sin \uptheta_{v_0}(v)] N_{v_0}\nn\\
N_v&=&\!\!-[\sin \uptheta_{v_0}(v)] v_0+ [\cos \uptheta_{v_0}(v)] N_{v_0}~~.
\eeqa
\end{prop}

\begin{remark}
\label{rem:epsTheta}
We have:
\be
\upsigma_{v_0}(v)=\sign(\cos\uptheta_v(N_{v_0}))=\sign (\cG(N_{v_0}, v))~~,
\ee
with the convention $\sign(0)=0$.
\end{remark}

To give a universal geometric description of the notions introduced
above, let $\pi:T\cM\rightarrow \cM$ be the projection of the tangent
bundle of $\cM$ and $F\eqdef \pi^\ast(T\cM)$ be the tautological
vector bundle over $T\cM$ (which is also known \cite{SLK} as the
Finsler bundle of $\cM$). A local section of $F$ defined on an open
set $U\subset T\cM$ can be viewed as a map $f:U\rightarrow T\cM$ which
satisfies $f(u)\in T_{\pi(u)}\cM$ for all $u\in U$. Thus $f$ is a map
which associates to a tangent vector $u\in T\cM$ another vector,
tangent to $\cM$ at the {\em same} point as $u$. Notice that $F$
carries the natural Euclidean pairing $\cG_F:F\times_{T\cM}
F\rightarrow \R_{T\cM}$ given by the bundle pull-back of $\cG$.  It
also carries an orientation described by the bundle pullback of the
volume form of $(\cM,\cG)$, which is a nowhere-vanishing section of
the determinant line bundle of $F$. We will denote $\cG_F$ by
$\cG$ for ease of notation.

\begin{definition}
\label{def:normsignedanglemap}
The {\em normalization map} and {\em positive normal map} of
$(\cM,\cG)$ are the sections $T,N\in \Gamma(\dot{T}\cM,F)$ defined
through:
\be
T(u)\eqdef \frac{u}{||u||}~~,~~N(u)\eqdef N_u~~\forall u\in \dot{T}\cM~~.
\ee
The {\em chirality map} and {\em signed angle map} of $(\cM,\cG)$ are
the maps $\upsigma:{\dot T}\cM\times_\cM {\dot T}\cM\rightarrow
\{-1,0, 1\}$ and $\uptheta:{\dot T}\cM\times_\cM {\dot T}\cM\rightarrow
(-\pi,\pi]$ defined through:
\be
\upsigma(v_0,v)\eqdef \upsigma_{v_0}(v)~~,~~\uptheta(v_0,v)\eqdef \uptheta_{v_0}(v)~~.
\ee
for all $(v_0,v)\in {\dot T}\cM\times_\cM {\dot T}\cM$~~.
\end{definition}

\noindent Notice that the maps $\upsigma$ and $\uptheta$ are
discontinuous. The following notion encodes oriented Frenet frames of
nondegenerate curves in a universal manner.

\begin{definition}
\label{def:Frenet}
The {\em abstract oriented Frenet frame} of $(\cM,\cG)$ is the
orthononormal frame $(T,N)$ of the vector bundle
$F\vert_{\dot{T}\cM}$.
\end{definition}

\noindent As we show in the next subsection, the Finsler vector fields
$T$ and $N$ give a universal geometric description of the oriented
Frenet frame of nondegenerate curves. With the definitions above, we
have (see Remark \ref{rem:epsTheta}):
\ben
\label{upsigma}
\upsigma=\sign(\cG\circ(N\times_\cM T))~~,
\een
where we view $\cG$ as a map $\cG:T\cM\times_\cM T\cM\rightarrow \R$
and $N,T$ as maps from $\dot{T}\cM$ to $\dot{T}\cM$. Here $N\times_\cM
T: \dot{T}\cM\times_\cM \dot{T}\cM \rightarrow \dot{T}\cM\times_\cM
\dot{T}\cM$ denotes the map defined through:
\be
(N\times_\cM T)(u,v)\eqdef (N(u),T(v))\in {\dot T}_m\cM\times \dot{T}_m\cM~~\forall (u,v)\in \dot{T}_m\cM\times \dot{T}_m\cM~~.
\ee

\subsection{The signed curvature and signed turning rate of a curve in $\cM$}
\label{app:signedcurv}

\noindent Let $(\cM,\cG)$ be an oriented Riemann surface and
$\varphi:I\rightarrow \cM$ be a smooth curve in $\cM$. For any $t\in
I_\reg$, let:
\be
T_\varphi(t)\eqdef \frac{\dot{\varphi(t)}}{||\dot{\varphi}(t)||}\in T_{\varphi(t)}\cM
\ee
be the unit tangent vector to $\varphi$ at time $t$.

\begin{definition}
The {\em positive normal vector} to the curve $\varphi$ at time $t\in
I_\reg$ is the positive unit normal $N_\varphi(t)\eqdef
N_{T_\varphi(t)}\in T_{\varphi(t)}\cM$ to $T_\varphi(t)$. The positive
orthonormal basis $(T_\varphi(t),N_\varphi(t))$ of $T_{\varphi(t)}\cM$
is called the {\em positive Frenet basis} at time $t\in I_\reg$. The
{\em oriented Frenet frame} of $\varphi$ is the frame
$(T_\varphi,N_\varphi)$ of the vector bundle
$(\varphi\vert_{I_\reg})^\ast(T\cM)$ defined by $T_\varphi(t)$ and
$N_\varphi(t)$ as $t$ varies in $I_\reg$.
\end{definition}

\noindent For any $t\in I_\reg$, we have:
\be
T_\varphi(t)=T(\dot{\varphi}(t))~~,~~N_\varphi(t)=N(\dot{\varphi}(t))~~,
\ee
where $T$ and $N$ are the normalization and positive normal maps of
$(\cM,\cG)$ defined in the previous subsection. Accordingly, we have:
\be
T_\varphi=(\dot{\varphi}\vert_{I_\reg})^\ast(T)~~,~~N_\varphi=(\dot{\varphi}\vert_{I_\reg})^\ast(N)~~,
\ee
where $\dot{\varphi}=c(\varphi):I\rightarrow T\cM$ is the canonical
lift of $\varphi$ to $T\cM$. Since $\pi\circ \dot{\varphi}=\varphi$,
we have $\dot{\varphi}^\ast(F)=\varphi^\ast(T\cM)$.

Let $s$ be an increasing proper length parameter along $\varphi$ with
image interval $J$ (thus $s$ gives an orientation-preserving
diffeomorphism ${\hat s}$ from $I$ to $J$) and let $J_\reg\eqdef {\hat
  s}(I_\reg)$. We have $\dot{s}(t)=||\dot{\varphi}(t)||$ for all $t\in
I$, which implies:
\be
T_\varphi(s)=\varphi'(s)~~\forall s\in I_\reg~~,
\ee
where the prime denotes derivative with respect to $s$. The vector
$K_\varphi(s)\eqdef \nabla_s T_\varphi(s)$ is the classical {\em
  curvature vector} of $\varphi$. The relation
$\cG(T_\varphi(s),T_\varphi(s))=1$ implies:
\be
K_\varphi(s)\perp T_\varphi(s)~~\forall s\in J_\reg~~.
\ee
Recall that the ordinary curvature of $\varphi$ at $s\in J_\reg$ is
the {\em strictly positive} quantity $\chi_\varphi(s)$ defined
through:
\be
\chi_\varphi(s)\eqdef ||K_\varphi(s)||=||\nabla_s T_\varphi(s)||>0~~\forall s\in I_\reg~~.
\ee

\begin{definition}
The {\em turning sign} (or {\em chirality}) of the curve $\varphi$ at
$s\in J_\reg$ is the chirality of the vector $K_\varphi(s)$ with
respect to $T_\varphi(s)$:
\be
\upsigma_\varphi(s)\eqdef \upsigma_{T_\varphi(s)}(K_\varphi(s))=\sign \, \cG(N_\varphi(s), K_\varphi(s)) \in \{-1,0,1\}~~,
\ee
where we set $\upsigma_\varphi(s)=0$ when $K_\varphi(s)=0$.
\end{definition}

\noindent Notice that $\upsigma_\varphi(s)=0$ iff
$K_\varphi(s)=0$. When $K_\varphi(s)\neq 0$, the oriented angle
$\uptheta_{T_\varphi(s)}(K_\varphi(s))$ equals
$\upsigma_\varphi(s)\frac{\pi}{2}$. In all cases, we have
$K_\varphi(s)=\sigma_\varphi(s)\chi_\varphi(s)N_\varphi(s)$.

\begin{definition}
\label{def:xi}
The {\em signed curvature} of $\varphi$ at $s\in J_\reg$ is defined
through:
\ben
\label{xidef}
\xi_\varphi(s)\eqdef \upsigma_\varphi(s)\chi_\varphi(s) \in \R~~.
\een
\end{definition}

\noindent Thus $K_\varphi(s)=\xi_\varphi(s) N_\varphi(s)$. 
Notice that $\xi_\varphi$ coincides with the ordinary signed
curvature of plane curves when $(\cM,\cG)$ is the Euclidean plane.
The Frenet-Serret relations of $\varphi$ take the form:
\beqan
K_\varphi(s)\eqdef \nabla_s T_\varphi(s) &=&+\xi_\varphi(s) N_\varphi(s)~~\label{KN}\\
\nabla_s N_\varphi(s)&=&-\xi_\varphi(s) T_\varphi(s)~~\forall s \in J_\reg~~.\label{NK}
\eeqan

\begin{definition}
The {\em rescaled signed turning rate} of $\varphi$ at $t\in I_\reg$ is the
quantity:
\ben
\label{OmegaDef}
\hOmega_\varphi(t)\eqdef -\frac{\cG(N_\varphi(t),\nabla_t\dot{\varphi}(t))}{||\dot{\varphi}(t)||}=-||\dot{\varphi}(t)||\xi_\varphi(t)\in \R~~.
\een
\end{definition}

\begin{prop}
The Fenet-Serret relations of $\varphi$ are equivalent with:
\ben
\label{FS2}
\nabla_t T_\varphi(t) =-\hOmega_\varphi(t)N_\varphi(t)~~,~~\nabla_t N_\varphi(t) = +\hOmega_\varphi(t) T_\varphi(t)~~\forall t\in I_\reg~~.
\een
\end{prop}

\begin{proof}
Follows by direct computation from \eqref{KN} and \eqref{NK} upon
using \eqref{OmegaDef}.
\end{proof}

\noindent The first Frenet-Serret relation can also be
written as:
\ben
\label{covacc2}
\nabla_t\dot{\varphi}(t)=\frac{\dd||\dot{\varphi}(t)||}{\dd t} T_\varphi(t)-||\dot{\varphi}(t)|| \hOmega_\varphi(t) N_\varphi(t)~~\forall t\in I_\reg~~.
\een
Note that $K_\varphi$, $\upsigma_\varphi$, $\chi_\varphi$, $\xi_\varphi$ and
$\hOmega_\varphi$ are second order quantities which cannot be expressed
through objects defined on $T\cM$ unless $\varphi$ satisfies a
geometric second order ODE on $\cM$.

\section{Relation between parameters used in this paper and parameters used in the physics literature}
\label{app:param}

\subsubsection{The first slow roll parameter.}

It is customary to define another second order parameter of a curve as follows:

\begin{definition}
\label{def:epsilon}
The {\em first slow roll parameter} of a curve
$\varphi:I\rightarrow \cM$ is the function
$\bepsilon_\varphi:I\rightarrow \R$ defined through:
\be
\bepsilon_{\varphi}(t)\eqdef -\frac{\dot{H_\varphi}(t)}{H_\varphi(t)^2}=-\frac{1}{H_\varphi(t)}\frac{\dd}{\dd t} \log H_\varphi(t)=3{\hat \bepsilon}_\varphi(t)~~,
\ee
where:
\be
{\hat \bepsilon}_\varphi(t)\eqdef -\frac{1}{\cH_\varphi(t)}\frac{\dd}{\dd t} \log \cH_\varphi(t)
\ee
is the {\em rescaled} first slow roll parameter of $\varphi$.
\end{definition}

\noindent For cosmological curves, the first slow roll parameter
reduces to a function of the first IR parameter and hence does not
give an independent basic observable.

\begin{prop}
\label{prop:epsilon}
Suppose that $\varphi:I\rightarrow \cM$ is a cosmological curve. Then we have:
\be
{\hat \bepsilon}_\varphi(t)=\frac{\kappa_\varphi(t)}{1+\kappa_\varphi(t)}~~\forall t\in I~~.
\ee
\end{prop}

\begin{proof}
Follows immediately from the cosmological equation \eqref{eomsingle}.
\end{proof}

\begin{remark}
When $\varphi$ is a cosmological curve, we have
$\bepsilon_\varphi(t)=0$ iff $\kappa_\varphi(t)=0$ and
$\bepsilon_\varphi(t)=1$ (i.e. ${\hat \bepsilon}_\varphi(t)=1/3$) iff
$\kappa(t)=1/2$.
\end{remark}

The definition of other parameters (beyond the first slow roll
parameter) used in the physics literature (in the conventions of
reference \cite{Paban}) is as follows, where for ease of notation we
don't indicate the curve $\varphi$ on which $s$ is an increasing
length parameter:
\beqa
&&\delta\eqdef \frac{\ddot{s}}{H\dot{s}}~~,~~\eta\eqdef \frac{\dot{\bepsilon}}{H\bepsilon}=2\delta+2\bepsilon~~,~~\xi\eqdef \frac{\dddot{s}}{H^2\ddot{s}}\\
&& \nu\eqdef \frac{1}{\hOmega}\frac{\dd}{\dd t}\left(\frac{\hOmega}{H}\right)~~,~~\bepsilon_V\eqdef \frac{M^2}{2}\frac{||\dd \Phi||^2}{\Phi^2}~~.
\eeqa
Here $\eta$ is the ordinary second slow roll parameter. We have
$\eta=-\delta$. The parameters used in this paper are:
\be
M_0\eqdef M\sqrt{\frac{2}{3}}~~,~~\cH\eqdef 3H~~,~~ \kappa\eqdef \frac{\dot{s}^2}{2\Phi}~~,~~\hetap\eqdef \frac{\eta}{3}=-\frac{1}{3}\delta~~,~~\Xi\eqdef \frac{M_0}{2}\frac{\dd \Phi}{\Phi}
\ee
and $\hat{\bepsilon}\eqdef -\frac{\bepsilon}{3}$, which  was defined above. We have:
\beqa
&& \cH=\sqrt{2\Phi}(1+\kappa)^{1/2}~~,~~\kappa=\frac{\hepsilon}{1-\hepsilon}~~,~~\kappa(\kappa+1)=\frac{\hepsilon}{(1-\hepsilon)^2}~~,~~||\Xi||^2=\hepsilon_V~~.
\eeqa

\subsubsection{A certain formula found in the physics literature.}

\noindent Relations \eqref{cdef} and \eqref{homegadef} give:
\be
\hOmega=-\cH\frac{\sin\theta}{\hc}=\frac{||\dd \Phi||^\v}{\cN}\sin\theta~~.
\ee
As mentioned above, the parameter $\delta$ of \cite{Paban} is
related to our parameters by $\delta=-3\hetap$, while the
ordinary second slow roll parameter used in loc. cit. is given by
$\eta=\frac{\dot\epsilon}{H\epsilon}=2\delta+2\epsilon$.
Using this and the relation $\bepsilon=3{\hat \bepsilon}$, we compute:
\be
1+\frac{\eta}{2(3-\bepsilon)}=\frac{3+\delta_\varphi}{3-\epsilon}=\frac{1-\hetap}{1-\hat{\bepsilon}}
\ee
and:
\be
\frac{\bepsilon}{3(1-\bepsilon/3)^2}=\frac{\hat{\bepsilon}}{(1-\hat{\bepsilon})^2}=\kappa(1+\kappa)~~.
\ee
Moreover, we have $||\dd\Xi||^2=\frac{1}{3}\bepsilon_V$.  This allows
one to write our relation \eqref{mrel} (which is a consequence
\eqref{etaomegactheta}) in the form:
\ben
\label{HPrel}
\frac{\bepsilon_V}{\bepsilon}=\left(1+\frac{\eta}{2(3-\bepsilon)}\right)^2+\frac{\hOmega^2}{9H^2}\frac{1}{(1-\frac{\bepsilon}{3})^2}~~,
\een
which was originally given in \cite{AP,HP}. We stress that this
formidable looking equation is a rewriting of relations
\eqref{etaomegactheta} (the first of which which can be taken as the
      {\em geometric definition} of $\hetap$), being a complicated
      re-formulation of the latter in terms of the parameters
      $\bepsilon,\bepsilon_V, \hOmega$ and $\cH$. The comparative
      simplicity of \eqref{etaomegactheta} and \eqref{mrel}
      underscores the conceptual advantage of using parameters which
      are geometrically meaningful.

\end{document}